\documentclass[aps,pra,twocolumn,nofootinbib, superscriptaddress,10pt]{revtex4-1}

\usepackage{mathdots}
\usepackage{amsfonts,amssymb,amsmath}
\usepackage{mathrsfs}
\usepackage[]{graphics,graphicx,epsfig}
\usepackage{amsthm}
\usepackage{bm}
\usepackage{csquotes}
\MakeOuterQuote{"}
\usepackage{epstopdf}
\usepackage{color}

\newcommand{\lsp}{\hspace{0.1em}}

\def\identity{\leavevmode\hbox{\small1\kern-3.8pt\normalsize1}}

\newtheorem{theorem}{Theorem}

\newtheorem{lemma}{Lemma}
\newtheorem{proposition}{Proposition}

\newcommand{\ket}[1]{\left | #1 \right\rangle}

 \newcommand{\rmd}{\mathrm{d}}
 \newcommand{\rme}{\mathrm{e}}
  
 \newcommand{\rmi}{\mathrm{i}}

 \newcommand{\rmF}{\mathrm{F}}

 \newcommand{\rmT}{\mathrm{T}}

 \newcommand{\caH}{\mathcal{H}}

\newcommand{\scrD}{\mathscr{D}}
\newcommand{\scrH}{\mathscr{H}}

  \newcommand{\scrE}{\mathscr{E}}

\newcommand{\tr}{\operatorname{tr}}
\newcommand{\1}{\operatorname{\uppercase\expandafter{\romannumeral1}}}
\newcommand{\2}{\operatorname{\uppercase\expandafter{\romannumeral2}}}
\newcommand{\3}{\operatorname{\uppercase\expandafter{\romannumeral3}}}
\newcommand{\4}{\operatorname{\uppercase\expandafter{\romannumeral4}}}
\newcommand{\5}{\operatorname{\uppercase\expandafter{\romannumeral5}}}
\newcommand{\6}{\operatorname{\uppercase\expandafter{\romannumeral6}}}
\newcommand{\7}{\operatorname{\uppercase\expandafter{\romannumeral7}}}
\newcommand{\8}{\operatorname{\uppercase\expandafter{\romannumeral8}}}
\newcommand{\9}{\operatorname{\uppercase\expandafter{\romannumeral9}}}

\def\eqref#1{\textup{(\ref{#1})}}  
\newcommand{\eref}[1]{Eq.~\textup{(\ref{#1})}}
\newcommand{\Eref}[1]{Equation~\textup{(\ref{#1})}}
\newcommand{\esref}[1]{Eqs.~\textup{(\ref{#1})}}
\newcommand{\Esref}[1]{Equations~\textup{(\ref{#1})}}

\newcommand{\fref}[1]{Fig.~\ref{#1}}

\newcommand{\tref}[1]{Table~\ref{#1}}

\newcommand{\thref}[1]{Theorem~\ref{#1}}
\newcommand{\Thref}[1]{Theorem~\ref{#1}}
\newcommand{\thsref}[1]{Theorems~\ref{#1}}
\newcommand{\Thsref}[1]{Theorems~\ref{#1}}

\newcommand{\lref}[1]{Lemma~\ref{#1}}
\newcommand{\Lref}[1]{Lemma~\ref{#1}}
\newcommand{\lsref}[1]{Lemmas~\ref{#1}}

\newcommand{\pref}[1]{Proposition~\ref{#1}}
\newcommand{\Pref}[1]{Proposition~\ref{#1}}

\newcommand{\cref}[1]{Conjecture~\ref{#1}}
\newcommand{\Cref}[1]{Conjecture~\ref{#1}}

\newcommand{\rcite}[1]{Ref.~\cite{#1}}
\newcommand{\rscite}[1]{Refs.~\cite{#1}}

\def\<{\langle}  
\def\>{\rangle}  

\begin{document}

\title{Optimal Verification of the Bell State and Greenberger-Horne-Zeilinger States in Untrusted Quantum Networks}

\author{Yun-Guang Han}
\affiliation{State Key Laboratory of Surface Physics and Department of Physics, Fudan University, Shanghai 200433, China}
\affiliation{Institute for Nanoelectronic Devices and Quantum Computing, Fudan University, Shanghai 200433, China}
\affiliation{Center for Field Theory and Particle Physics, Fudan University, Shanghai 200433, China}

\author{Zihao Li}
\affiliation{State Key Laboratory of Surface Physics and Department of Physics, Fudan University, Shanghai 200433, China}
\affiliation{Institute for Nanoelectronic Devices and Quantum Computing, Fudan University, Shanghai 200433, China}
\affiliation{Center for Field Theory and Particle Physics, Fudan University, Shanghai 200433, China}

\author{Yukun Wang}
\affiliation{Department of Computer Science and Technology, China University of Petroleum, Beijing 102249, China}

\author{Huangjun Zhu}
  \email{zhuhuangjun@fudan.edu.cn}
\affiliation{State Key Laboratory of Surface Physics and Department of Physics, Fudan University, Shanghai 200433, China}
\affiliation{Institute for Nanoelectronic Devices and Quantum Computing, Fudan University, Shanghai 200433, China}
\affiliation{Center for Field Theory and Particle Physics, Fudan University, Shanghai 200433, China}

\begin{abstract}
Bipartite and multipartite entangled states are basic ingredients for constructing quantum networks and their accurate verification is crucial to the  functioning of the networks, especially for untrusted networks. Here we propose a simple approach for verifying the Bell state in an untrusted quantum network  in which one party is not honest. Only local projective measurements are required  for the honest party. It turns out each verification protocol is tied to a probability distribution on the Bloch sphere and its performance has an intuitive geometric meaning. This geometric picture enables us to construct the optimal and simplest verification protocols, which  are also very useful to detecting entanglement in the untrusted network. Moreover, we show that our verification protocols can achieve almost the same sample efficiencies as protocols tailored to standard quantum state verification. Furthermore, we establish an intimate connection between the verification of Greenberger-Horne-Zeilinger states and the verification of the Bell state. By virtue of this connection we  construct the optimal protocol for verifying  Greenberger-Horne-Zeilinger states and for detecting genuine multipartite entanglement. 
\end{abstract}

\date{\today}
\maketitle

\section{INTRODUCTION}
Entanglement is the characteristic of quantum mechanics and key resource in quantum information processing \cite{Horodecki2009,Nonlocality2014,Steering2020}. 
As typical examples of bipartite and multipartite entangled states, the Bell state and Greenberger-Horne-Zeilinger (GHZ) states \cite{GHZ1989,GHZ1990} play crucial  roles in numerous quantum information processing tasks and in foundational studies, such as quantum teleportation \cite{Tele1993,Tele1997,Tele2004}, quantum key distribution \cite{E91,DI2007}, quantum random number generation \cite{QRNG2016}, and nonlocality tests \cite{CHSH1969,Pan2000}.
Furthermore, as a special example of graph states \cite{hein2004}, GHZ states are  useful to constructing quantum networks \cite{Kimble2008,Wehner2018} designed for distributed quantum information processing, such as quantum secret sharing \cite{QSS1999,QSS2014}, quantum conference key agreement \cite{Zhao2020GHZ} and distribution \cite{Fu2015GHZ}.

To guarantee the proper functioning of a quantum network, it is essential to verify the entangled state deployed in the network accurately and efficiently, especially for untrusted networks \cite{Eisert2020,Supic2020,Pappa2012,McCut2016,Supic2016,Gheorghiu2017,Lu2020,Unnikrishnan2020,Markham2020}. This scenario has wide applications in quantum information processing, such as one-sided device-independent (DI) quantum key distribution \cite{Branciard2012}, anonymous communication \cite{Unnikrishnan2019,Hahn2020GHZ}, and verifiable quantum secure modulo summation \cite{Hayashi2019}.
Meanwhile, this problem is  tied to the foundational studies on quantum steering in the asymmetric scenario \cite{Steering2007,ExpSteering2010,Steering2015,Steering2020} and the uncertainty principle in the presence of quantum memory \cite{BertCCR10,Zhu21}.

 Unfortunately, not much is known about quantum verification in untrusted networks despite its significance. This is because not all parties in the networks are honest, and the verification problem gets much more complicated in the presence of dishonest parties. In particular, traditional tomographic approaches are not applicable in the network setting even if their low efficiency is tolerable. Also, most alternative approaches, including direct fidelity estimation \cite{Flammia2011} and quantum state verification (QSV) \cite{HayaMT06,Haya09,pallister2018,ZhuH2019AdvS,ZhuH2019AdvL,Takeuchi2018}, are not applicable, although QSV can address the adversarial  scenario in which the source is not trustworthy \cite{Takeuchi2018,ZhuH2019AdvS,ZhuH2019AdvL}. DI QSV \cite{Dimic2021} based on self-testing \cite{Mayers2004,Supic2020,McKague2012,
Yang2014,Kaniewski2016,Hayashi2018,
Vidick2020} can be applied in the network setting in principle, but is too resource consuming and too demanding with current technologies. For the Bell state and GHZ states, optimal verification protocols are known  when all parties are honest \cite{pallister2018,ZhuH2019O,wang2019,li2019_bipartite,yu2019,li2020GHZ,Dangniam2020}. In the network setting,  however, only  suboptimal protocols are known in the literature \cite{Pappa2012,McCut2016,Supic2016,Gheorghiu2017,Unnikrishnan2020}. 

In this paper, we propose a simple approach for verifying the Bell state over an untrusted network in the semi-device-independent (SDI) scenario in which one party is not honest. Only local projective measurements are required  for the honest party. In addition, we establish a simple  connection between verification   protocols of the Bell state and probability distributions on the Bloch sphere and reveal an intuitive geometric interpretation of the performance of each verification protocol. By virtue of this geometric picture, we construct the optimal and simplest protocols for verifying the Bell state, which can also be applied to detecting entanglement in the untrusted network. Moreover, we determine the sample efficiencies of our SDI verification protocols in addition to the guessing probabilities.

Furthermore, we establish an intimate connection between the verification of GHZ states and the verification of the Bell state. Thanks to this connection,  efficient protocols for verifying GHZ states can easily be constructed from the counterparts for the Bell state.  Notably, this connection enables us to  construct the optimal protocol for verifying  GHZ states and for detecting genuine multipartite entanglement (GME). To put our work in perspective, we also provide a detailed comparison between  SDI QSV
considered in this work and standard QSV as well as DI QSV based on  self-testing. For the Bell state and GHZ states, SDI verification can achieve almost the same sample efficiency as standard QSV; by contrast, the sample efficiency in the DI scenario is in general quadratically worse in the infidelity unless there exists a suitable Bell inequality for which 
 the quantum bound coincides with the algebraic bound.

\section{RESULTS}
\noindent\textbf{Verification of the  Bell state}\\
Suppose two distant parties, Alice and Bob, want to create the Bell state $|\Phi\rangle=(|00\>+|11\>)/\sqrt{2}$ as follows:  Bob first prepares $|\Phi\>$ in his lab and then sends one particle of the entangled pair to Alice using a quantum channel.  To verify this state Alice can perform a random projective measurement from a set of accessible measurements and then ask Bob to guess the measurement outcome given the measurement chosen. Each projective measurement is specified by a unit vector $\bm{r}$ on the Bloch sphere, which specifies the two outcomes $P_\pm=(\mathbb{I} \pm \bm{r}\cdot\bm{\sigma})/2$, where $\bm{\sigma}$ is the vector composed of the three Pauli matrices. 
If Bob is honest and prepares the target state $|\Phi\>$, then his reduced states corresponding to the two outcomes $P_{+}$ and $P_-$ have mutually orthogonal supports, so he can guess the measurement outcome with certainty by performing a suitable projective measurement.

If Bob is not honest and tries to prepare a different  state $\rho$ instead of $|\Phi\>$, then his \emph{guessing probability}---the probability of successful guess---would be limited.
In this case, Alice cannot distinguish two states that can be turned into each other by local operations of  Bob; nevertheless, she can verify whether the state prepared is indeed $|\Phi\>$ up to these local operations.
Let $\rho_{\pm}=\tr_A(\rho P_{\pm})$ be the unnormalized reduced states of Bob. To guess the measurement outcome of Alice, Bob can perform a two-outcome POVM $\{E_+, E_-\}$ to distinguish $\rho_+$ and $\rho_-$.
By the Helstrom theorem \cite{Helstrom1976}, the
maximum guessing probability $ \gamma(\rho,\bm{r})$ over all POVMs (or projective measurements) reads $ \gamma(\rho,\bm{r})=(1+\|\rho_+-\rho_-\|_1)/2$.

Recall that a general  two-qubit state has the form
\begin{equation}
\rho=\frac{1}{4}\biggl(\mathbb{I} +\bm{a}\cdot \bm{\sigma} \otimes \mathbb{I}+\mathbb{I}\otimes \bm{b}\cdot \bm{\sigma}+\sum_{j,k}T_{jk}\sigma_j\otimes \sigma_k\biggr), \label{eq:twoqubit}
\end{equation}
where $\sigma_j, \sigma_k$ are Pauli matrices (also denoted by $X,Y,Z$),  $\bm{a}$ and $\bm{b}$ are the Bloch vectors of the two reduced states, and  $T$ is the correlation matrix.
If $\rho$ is pure, then we can deduce
(cf. Supplementary Note 1),
\begin{align}\label{eq:gammaTr}
\gamma(\rho,\bm{r})&=\frac{1}{2}\bigl(1+\|T^\rmT\bm{r}\|\bigr)=\frac{1}{2}\bigl(1+\bigl\|\sqrt{T T^\rmT }\lsp\bm{r}\bigr\|\bigr).
\end{align}
To understand the geometric meaning of $\gamma(\rho,\bm{r})$, note that  the set of vectors $\{\sqrt{T T^\rmT}\lsp\bm{r}: |\bm{r}|=1\}$  forms a rotational ellipsoid, which is called the \emph{correlation ellipsoid} and denoted by $\scrE_\rho$ as illustrated in Fig.~\ref{fig:Correlation-ellipsoid} (cf. the steering ellipsoid  \cite{Jevtic2014, Zhang2019}). 
The  semi-major axis $\bm{v}$ and  semi-minor axis of $\scrE_\rho$ have length 1 and $C$, respectively, where $C$ is the concurrence of $\rho$ \cite{Wootters1998,Horodecki2009}. In addition, the radius $\|\sqrt{T T^\rmT}\lsp\bm{r}\bigr\|$ is determined by $C$ and the angle between $\bm{r}$ and the semi-major axis as follows,
\begin{equation}\label{eq:Tr}
\bigl\|\sqrt{T T^\rmT}\lsp\bm{r}\bigr\|=\|T^\rmT\bm{r}\|=\sqrt{C^2+(1-C^2)(\bm{r}\cdot\bm{v})^2}.
\end{equation}

\begin{figure}
\begin{center}
  \includegraphics[width=7cm]{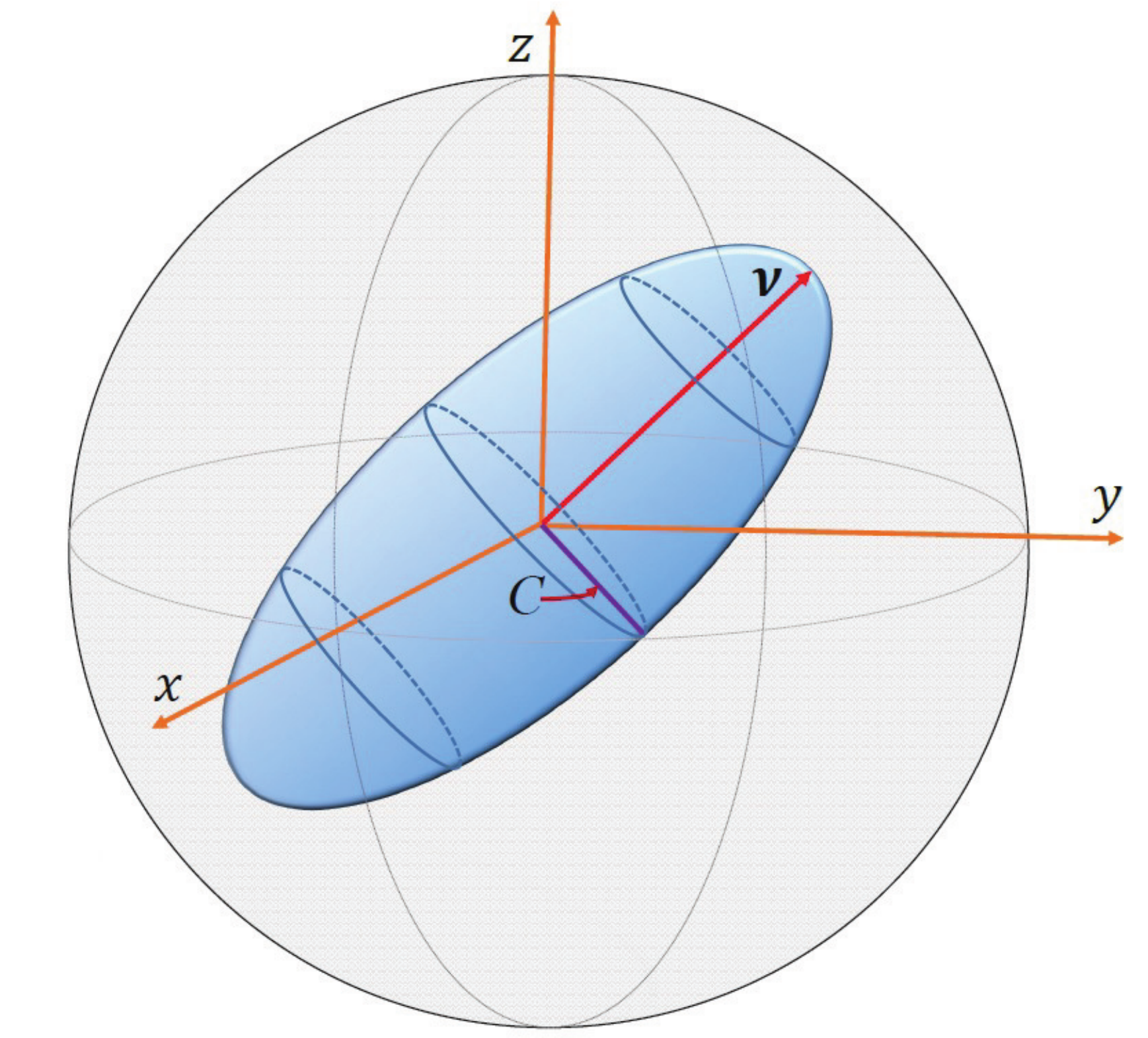}
  \caption{\label{fig:Correlation-ellipsoid}
Geometric illustration of the  $XYZ$ protocol in the Bloch sphere. For a given concurrence $C$, the guessing probability is maximized when the semi-major axis $\bm{v}$ of the correlation ellipsoid parallels one of the eight intelligent directions.     
}
\end{center}
\end{figure}

A verification strategy of Alice is determined by  a probability distribution $\mu$ on the Bloch sphere, which specifies the probability of performing each projective measurement. Given the strategy $\mu$ and the state $\rho$, the maximum average guessing probability of Bob reads
\begin{equation}
\gamma(\rho,\mu):=\int\rmd\mu(\bm{r})\gamma(\rho,\bm{r})=\frac{1}{2}+\frac{1}{2}\int\rmd\mu(\bm{r}) \|T^\rmT\bm{r}\|,
\end{equation}
where the bias is a weighted average of radii of the correlation ellipsoid. Denote by $\gamma_2(C,\mu)$ the maximum guessing probability over all pure states with concurrence at most $C$.  Note that  maximizing $\gamma(\rho,\mu)$ for a given concurrence amounts to choosing a proper orientation of the correlation ellipsoid so as to maximize the weighted average of radii, as illustrated in \fref{fig:Correlation-ellipsoid}. This intuition leads to the following theorem  as proved in Supplementary Note 1.
\begin{theorem}\label{thm:MGP}
Suppose $0\leq C\leq 1$; then
\begin{align}
\gamma_2(C,\mu)=&\frac{1}{2}[1+ g(C,\mu)],\label{eq:gammaCmu}\\
g(C, \mu ) :=&\max_{\bm{v}}\int\rmd\mu(\bm{r}) \sqrt{C^2+(1-C^2)(\bm{r}\cdot\bm{v})^2}, \label{eq:gCmu}
\end{align}
where the maximization in \eref{eq:gCmu} is over all unit vectors.
\end{theorem}
Any unit vector $\bm{v}$ that maximizes the integration in \eref{eq:gCmu} is called an \emph{intelligent direction}. For a given concurrence, the guessing probability is maximized when the major axis of the correlation ellipsoid parallels an intelligent direction. 
When $C=1$, the correlation ellipsoid is a sphere, in which case \thref{thm:MGP} yields $g(C, \mu )=1$ and $\gamma_2(C, \mu )=1$.   When $C=0$, the correlation ellipsoid reduces to a line segment, in which case we can deduce
\begin{align}
g^\ast (\mu):=g(0, \mu )&=\max_{\bm{v}}\int\rmd\mu(\bm{r}) |\bm{r}\cdot\bm{v}|,  \label{eq:gC0} \\
\gamma_2^\ast(\mu ):=\gamma_2(0, \mu )&=\frac{1}{2}+\frac{1}{2}\max_{\bm{v}}\int\rmd\mu(\bm{r}) |\bm{r}\cdot\bm{v}|.  \label{eq:gammaC0}
\end{align}
Notably,  entanglement can be certified in the shared system when the guessing probability surpasses the threshold $\gamma_2^\ast(\mu )$. The relation between the guessing probability and concurrence for various verification protocols are illustrated in Fig.~\ref{fig:Concurrence-Pure}.

\begin{figure}
\begin{center}
  \includegraphics[width=7.5cm]{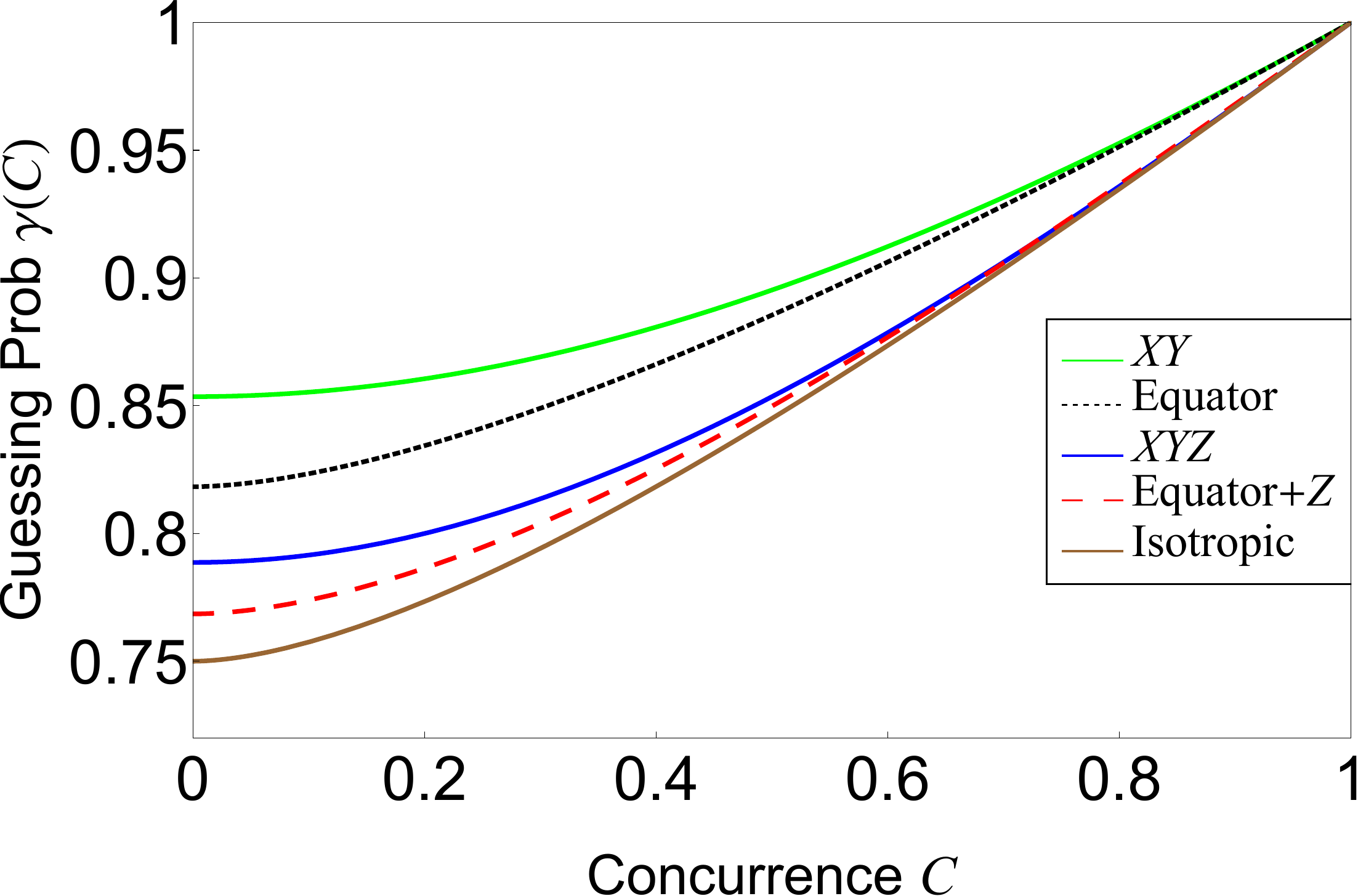}
  \caption{\label{fig:Concurrence-Pure}
  The  guessing probability $\gamma(C)=\gamma_2(C)$ as a function of the  concurrence $C$ for various verification protocols of the Bell state. Here the $XY$ protocol and isotropic protocol are introduced in the main text, while  other protocols are proposed in the Supplementary Material.
}
\end{center}
\end{figure}

\bigskip
\noindent\textbf{Alternative strategies of the adversary}\\
So far we have assumed that  the state $\rho$ prepared by Bob is a two-qubit pure state and $\rho_A:=\tr_B(\rho)$ is supported in the local support of the target Bell state, that is, the subspace spanned by $|0\rangle$ and $|1\rangle$. Can Bob gain any advantage
if $\rho_A$ is not supported in this subspace?
The answer turns out to be negative. Now Alice can first perform the projective measurement $\{P_A, \mathbb{I}-P_A\}$ with $P_A=|0\rangle\langle 0|+|1\rangle \langle 1|$ and then apply a verification protocol as before if she obtains the first outcome and reject otherwise. The maximum guessing probability $\gamma(C,\mu)$ of Bob for any pure state with $C(\rho)\leq C$ is the same as before as shown in the following lemma and proved in Supplementary Note 3.
\begin{lemma}\label{lem:gammaCHighDim}
$\gamma(C,\mu)=\gamma_2(C,\mu)$ for $0\leq C\leq 1$.
\end{lemma}
Note that $\gamma(C,\mu)=1$ when $C\geq 1$, in which case Bob can  prepare the target Bell state. So we can focus on the case $0\leq C\leq 1$. Define $\gamma^\ast(\mu):=\gamma(0,\mu)$, then  $\gamma^\ast(\mu)=\gamma_2^\ast(\mu)$ thanks to \lref{lem:gammaCHighDim}, so the threshold for entanglement detection remains the same as before; cf. \eref{eq:gammaC0}. In conjunction with the convexity of  $\gamma_2(C,\mu)$ in $C$ (cf. \lref{lem:Convex}  in Supplementary Note 2), \lref{lem:gammaCHighDim} implies that 
\begin{align}
\gamma(C,\mu)=\gamma_2(C,\mu)\leq \gamma^\ast(\mu)(1-C)+C,\quad  0\leq C\leq 1, \label{eq:gammaCmuLUB}
\end{align}
which offers the best linear upper bound for $\gamma(C,\mu)$. 
When the distribution $\mu$ is clear from the context, $\gamma_2^\ast(\mu)$ and $\gamma^\ast(\mu)$ are abbreviated as $\gamma_2^\ast$ and $\gamma^\ast$ for simplicity.

Above results can be extended to mixed states, although our main interest are pure states. Let $\hat{\gamma}(C,\mu)$ be the maximum guessing probability of Bob over all states with concurrence at most $C$. Define  $\hat{\gamma}_2(C,\mu)$ in a similar way, but assuming that $\rho_A$ is supported in the support of $P_A$. By the following theorem proved  in Supplementary Note 4, $\hat{\gamma}(C,\mu)$ and $\hat{\gamma}_2(C,\mu)$ are  weighted averages of  $\gamma^\ast(\mu)$ and $\gamma(1,\mu)=1$. 
\begin{theorem}\label{thm:Mixed}
  Suppose $0\leq C\leq 1$; then
  \begin{align}
  \hat{\gamma}(C,\mu)&=\hat{\gamma}_2(C,\mu)=
  (1-C)\gamma^\ast(\mu) +C \nonumber \\
  &=\frac{1+C}{2}+\frac{1-C}{2}\max_{\bm{v}}\int\rmd\mu(\bm{r}) |\bm{r}\cdot\bm{v}|.  \label{eq:gammahat}
  \end{align}
\end{theorem}

\bigskip
\noindent\textbf{Fidelity as the figure of merit}\\
Next, we consider the fidelity as the figure of merit, which is more natural for QSV. Here we assume that Bob controls the whole system except that of Alice, so we can assume that the state $\rho$ prepared by Bob is pure. Define the reduced fidelity
\begin{align}\label{def:reducedfidelity}
F_B(\rho):=\max_{U_B} \hspace*{0.05cm} \<\Phi|(\mathbb{I}_A\otimes U_B)\rho (\mathbb{I}_A\otimes U_B)^\dag |\Phi\>,
\end{align}
where the maximization is taken over all local unitary transformations on $\caH_B$. Denote by $\gamma^\rmF(F,\mu)$ the maximum guessing probability over all pure states with $F_B(\rho)\leq F$. Define $\gamma_2^\rmF(F,\mu)$ in a similar way, but assuming that $\rho_A$ is supported in the support of $P_A$.
It is known that $F_B(\rho)=[1+C(\rho)]/2\geq 1/2$  for any two-qubit pure state $\rho$ satisfying $P_A\rho_A=\rho_A$ \cite{Verstraete2002}.
So $\gamma_2^\rmF(F,\mu)$ is defined only for $1/2\leq F\leq 1$, although $\gamma^\rmF(F,\mu)$ is defined for $0\leq F\leq 1$.

The following theorem  proved in Supplementary Note~5 clarifies the relations between 
$\gamma^\rmF(F,\mu)$, $\gamma_2^\rmF(F,\mu)$, and $\gamma_2(C,\mu)$.
The guessing probabilities $\gamma^\rmF(F,\mu)$  for various verification protocols are illustrated in Fig.~\ref{fig:Fidelity-Pure}. 
\begin{theorem}\label{thm:gammaF}
	Suppose $1/2\leq F\leq 1$; then
	\begin{align}
	\gamma_2^\rmF(F,\mu)&=\gamma_2(2F-1,\mu)\leq 1-2(1-\gamma^\ast)(1-F). \label{eq:gamma2F}
	\end{align} 
	Suppose $0\leq F\leq 1$; then   
	\begin{align}
	\gamma^\rmF(F,\mu)&=\begin{cases}
	2\gamma^\ast F  & 0\leq F< 1/2,\\
\gamma_2(2F-1,\mu) & 1/2\leq F\leq 1,
	\end{cases}  \label{eq:gammaF}\\
	\gamma^\rmF(F,\mu)&\leq 1-2(1-\gamma^\ast)(1-F). \label{eq:gammaFLUB}
\end{align}	
\end{theorem}
\Eref{eq:gammaFLUB}  offers the best linear upper bound for $\gamma^\rmF(F,\mu)$ when $1/2\leq F\leq 1$ and demonstrates the robustness of the verification protocol. \Thsref{thm:MGP} to \ref{thm:gammaF} corroborate the significance of the threshold $\gamma^{\ast}$ in verifying the Bell state and entanglement in the SDI scenario. 
Moreover,  the threshold $\gamma^{\ast}$  determines the sample efficiency, as we shall see shortly. Therefore, $\gamma^{\ast}$ can be regarded as the most important figure of merit for characterizing the performance of a verification protocol.

\begin{figure}
	\begin{center}
		\includegraphics[width=7.5cm]{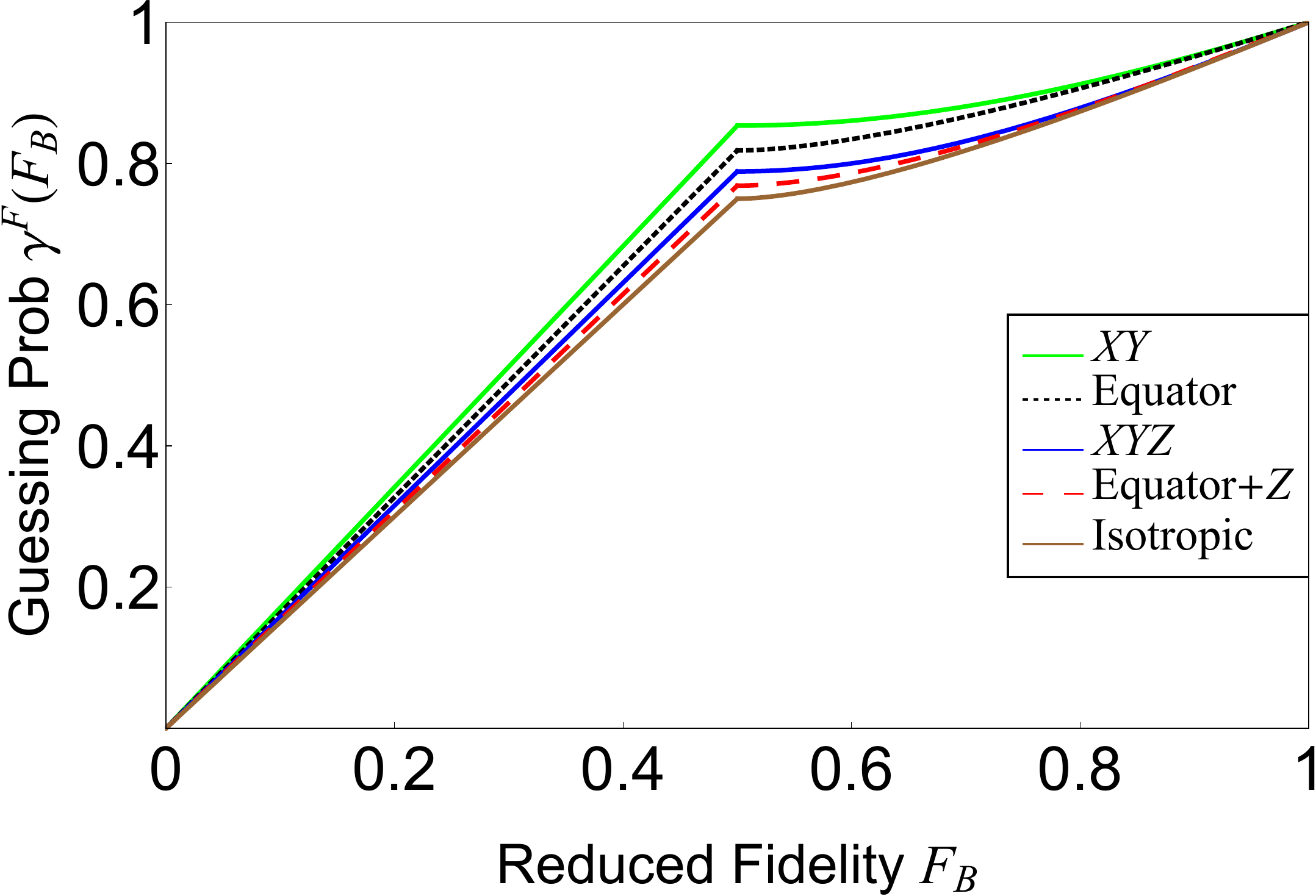}
		\caption{\label{fig:Fidelity-Pure}
			Relation between the  guessing probability $\gamma^\rmF(F_B)$ and the reduced fidelity $F_B$  for various verification protocols of the Bell state.  Here the $XY$ protocol and isotropic protocol are introduced in the main text, while  other protocols are proposed in the Supplementary Material. 
		}
	\end{center}
\end{figure}

\begin{table*}
	\caption{\label{tab:ProtocolGHZDH}
		Concrete protocols for verifying the Bell state in an untrusted quantum network.  Here $\gamma(C)$ ($\hat{\gamma}(C)$) is the maximum guessing probability for pure (mixed) states with concurrence at most $C$, and  $\bm{v}(C=0)$ is an intelligent direction for $C=0$. Entanglement can be certified when the guessing probability surpasses the threshold $\gamma^{\ast}=\gamma(0)=\hat{\gamma}(0)$. The $XY$ protocol and isotropic protocol are the simplest and optimal verification protocols, respectively.  All protocols listed, except for the isotropic protocol, can be generalized to GHZ states.
	}   
	\begin{math}
	\begin{array}{c|c|c|c|c}
	\hline\hline
	\mbox{Protocol}  &\mbox{Threshold } \gamma^{\ast}   & \gamma(C) \mbox{ (pure state)}  & \hat{\gamma}(C) \mbox{ (mixed state)} &  \bm{v}(C=0) \\[0.5ex]
	\hline
	XY  & \frac{1}{2}+\frac{1}{2\sqrt{2}}\approx 0.854     & \frac{1}{2}+\frac{1}{2}\sqrt{\frac{1+C^2}{2}} \vphantom{\bigg|}    &\frac{1}{4}[2+\sqrt{2}+(2-\sqrt{2})C\lsp]  & \frac{1}{\sqrt{2}}(1,1,0)^\rmT      \\[1.1ex]
	XYZ   & \frac{1}{2}+\frac{1}{2\sqrt{3}} \approx 0.789  &\frac{1}{2}+\frac{1}{2}\sqrt{\frac{1+2C^2}{3}}    &\frac{1}{6}[3+\sqrt{3}+(3-\sqrt{3})C\lsp] & \frac{1}{\sqrt{3}}(1,1,1)^\rmT    \\[1ex]      
	\mbox{Isotropic}  & \frac{3}{4}=0.75  &\frac{3}{4}+\frac{C^2\mathrm{arcsinh}\bigl(\frac{\sqrt{1-C^2}}{C}\bigr)}{4\sqrt{1-C^2}} & \frac{3+C}{4} & \mbox{any direction}   \\[1.7ex]
	
	\mbox{Equator} & \frac{1}{2}+\frac{1}{\pi}\approx 0.818           &\frac{1}{2}+\frac{1}{\pi}K(\sqrt{1-C^2}\lsp) &\frac{1}{2\pi}[\pi+2+(\pi-2)C\lsp] &\mbox{any direction in the $xy$-plane} \\[0.8ex]
	\mbox{Polygon(3)} & \frac{5}{6}\approx 0.833           & \frac{4+\sqrt{1+3C^2}}{6} & \frac{5+C}{6} & \mbox{any vertex direction}  \\[0.7ex]
	\mbox{Equator} +Z   & \frac{1}{2}+\frac{1}{\sqrt{4+\pi^2}}\approx 0.769       & - &\frac{1+C}{2}+\frac{1-C}{\sqrt{4+\pi^2}} & \frac{1}{\sqrt{4+\pi^2}}(\pi,0,2)^\rmT \\[1.5ex]
	\mbox{Polygon(3)}+Z & \frac{1}{2}+\frac{1}{\sqrt{13}}\approx 0.777    &  -   & \frac{1}{2}+\frac{1}{\sqrt{13}}+\bigl(\frac{1}{2}-\frac{1}{\sqrt{13}}\bigr)C & \frac{1}{\sqrt{13}}(3,0,2)^\rmT \\[1ex]
	\hline\hline
	\end{array} 
	\end{math}
\end{table*}

\bigskip
\noindent\textbf{Simplest and optimal  verification protocols}\\
Here we propose several concrete verification protocols, including the simplest and optimal protocols. The main results are summarized in \tref{tab:ProtocolGHZDH} and illustrated in Fig.~\ref{fig:Concurrence-Pure};  more technical details can be found in Supplementary Note~9.

In the simplest verification protocol, Alice can perform two projective measurements $\bm{r}_1$ and $\bm{r}_2$ with probabilities $p_1$ and $p_2$, respectively.  Here the maximum guessing probability $\gamma(C,\mu)$ only depends on the angle between $\bm{r}_1$ and $\bm{r}_2$ in addition to the probabilities $p_1$ and $p_2$. Moreover, $\gamma(C,\mu)$ is minimized when $\bm{r}_1 \cdot \bm{r}_2=0$ and $p_1=p_2=1/2$, in which case we have
\begin{equation}
g(C,\mu)=\sqrt{\frac{1+C^2}{2}},\quad  \gamma(C,\mu)=\frac{1}{2}+\frac{1}{2}\sqrt{\frac{1+C^2}{2}},
\end{equation}
and the guessing probability threshold is  $\gamma^{\ast}=(2+\sqrt{2})/4$. 
When $\bm{r}_1=(1,0,0)^\rmT$ and $\bm{r}_2=(0,1,0)^\rmT$ for example, we get the  $XY$ protocol. Previously, \rcite{Pappa2012} proposed an equivalent protocol, but neither derived the exact formula for the guessing probability  nor proved the optimality of the $XY$ protocol among all two-setting protocols.

To determine the optimal  protocol, we need to minimize $g(C,\mu)$ over $\mu$. By \thref{thm:MGP} (cf. \lref{lem:Convex}  in Supplementary Note 2), $g(C,\mu)$ is convex in $\mu$, so $g(C,\mu)$ is minimized when $\mu$ is the uniform distribution on the Bloch sphere, which yields  the  \emph{isotropic protocol} with 
\begin{align}
g(C, \mu )&
=\frac{1}{2}+\frac{C^2\mathrm{arcsinh}(\frac{\sqrt{1-C^2}}{C})}{2\sqrt{1-C^2}}, \\
\gamma(C, \mu ) &=\frac{3}{4}+\frac{C^2\mathrm{arcsinh}(\frac{\sqrt{1-C^2}}{C})}{4\sqrt{1-C^2}},
\end{align}
and the guessing probability threshold  is $\gamma^{\ast}=3/4$.

Protocols based on the Pauli $Z$ measurement and measurements on the $xy$-plane are of special interest to the verification of   GHZ states as we shall see shortly. Prominent examples include the $XYZ$ protocol (cf. \fref{fig:Correlation-ellipsoid}), equator protocol, equator$+Z$ protocol, polygon protocol,  and polygon$+Z$ protocol (see Supplementary Note 9).

\bigskip
\noindent\textbf{Sample efficiency}\\
To construct a practical verification protocol, it is crucial to clarify the sample efficiency. Although this problem has been resolved in standard QSV  \cite{pallister2018,ZhuH2019AdvS,ZhuH2019AdvL}, little is known about the sample efficiency in the DI and SDI scenarios \cite{Dimic2021}. 
Here we clarify the sample efficiency of our verification protocols in  the SDI scenario. 
Consider a quantum device that  is supposed to produce the target state $|\Phi\>\in\caH$, but actually produces the states $\rho_1,\rho_2,\dots,\rho_N$ in $N$ runs. Our task is to verify  whether these states are sufficiently close to the target state on average. Here the reduced fidelity is a natural choice for quantifying the closeness since Alice is ignorant to 
the local unitary transformations acting on Bob's system. To guarantee that the average reduced fidelity of the states  $\rho_1,\rho_2,\dots,\rho_N$  is larger than $1-\epsilon$  with  significance level $\delta$ (confidence level $1-\delta$), the number of tests required is determined in Supplementary Note 6, with the result
\begin{equation}\label{eq:NumberTestSDI}
N=\biggl\lceil\frac{ \ln \delta}{\ln[1-2(1-\gamma^{\ast})\epsilon\lsp]}\biggr\rceil\approx \frac{ \ln \delta^{-1}}{2(1-\gamma^{\ast})\epsilon}. 
\end{equation}
Note that the sample efficiency is determined by
the threshold  $\gamma^{\ast}=\gamma_2^\ast$  defined in \eref{eq:gammaC0}.

The minimum threshold  $\gamma^{\ast}=3/4$ is attained for the  isotropic protocol, in which case  $N\approx(2\ln \delta^{-1})/\epsilon$,
which is comparable to the number $(3\ln \delta^{-1})/(2\epsilon)$ required  in standard QSV \cite{pallister2018,ZhuH2019O}. So the Bell state can be verified in the SDI scenario almost as efficiently as in the standard QSV. Our protocol can achieve the optimal sample complexity because it is tied to a steering inequality whose quantum bound coincides with the algebraic bound. In  contrast, the sample complexity in the DI scenario is quadratically worse in the scaling with  $1/\epsilon$, that is, $N \propto (\ln \delta^{-1})/\epsilon^2$ \cite{Dimic2021} (cf.~Supplementary Note 7).

\begin{figure}
	\begin{center}
		\includegraphics[width=7.7cm]{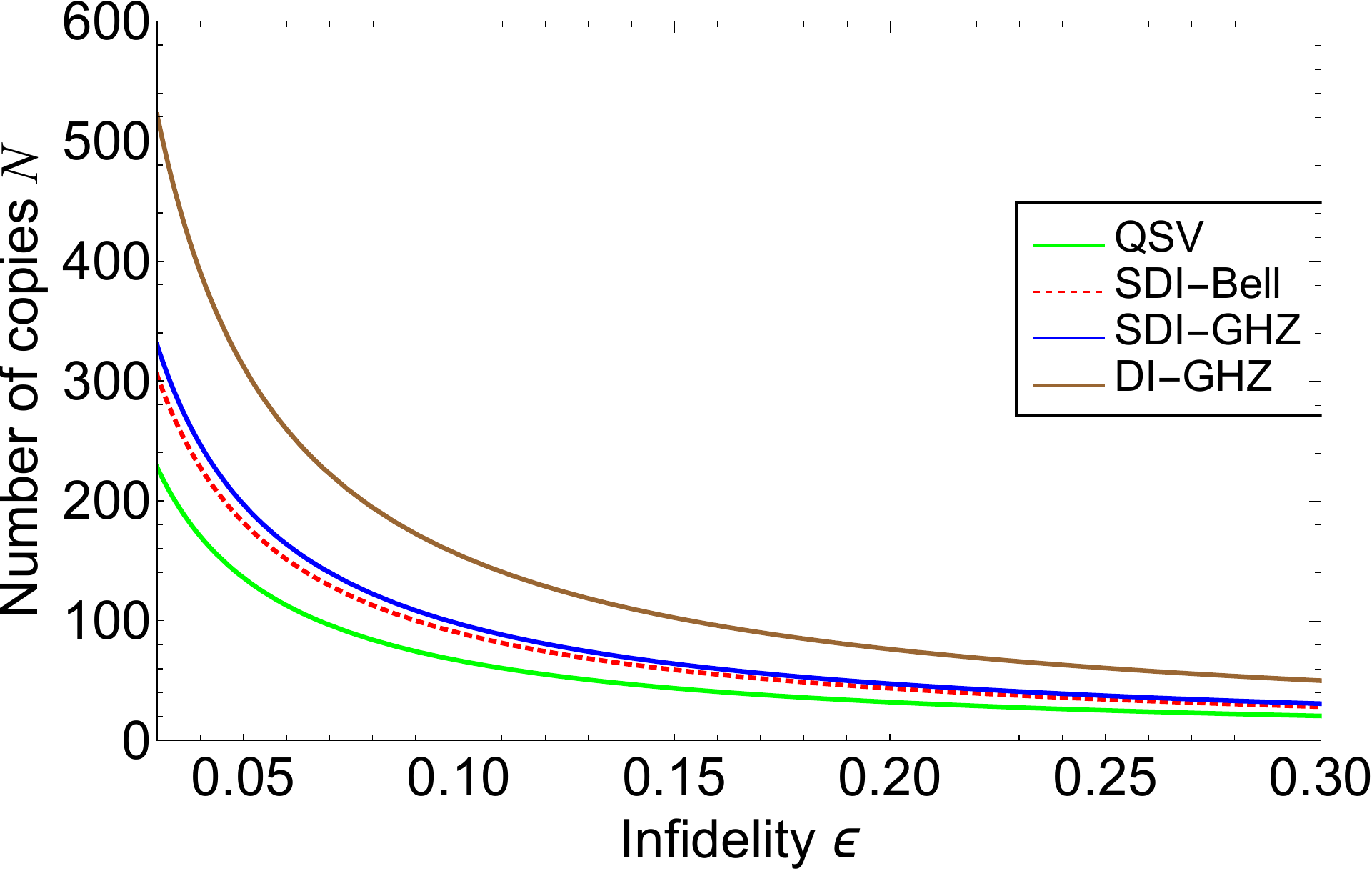}
		\caption{\label{fig:QSVDISDI}
			Sample complexities for 	
			verifying the  Bell state and GHZ states in three different scenarios. In standard QSV, 
			the Bell state and GHZ states can be verified with the same sample complexity \cite{pallister2018,li2020GHZ}. In the SDI scenario, the isotropic protocol is chosen for verifying the Bell state and  the optimized equator$+Z$ protocol is chosen for verifying the GHZ states. In the DI scenario, the Mermin inequality is employed for verifying 
			the three-qubit GHZ state  \cite{Kaniewski2016,Dimic2021}.  Here the significance level is chosen to be $\delta=0.01$. 
		}
	\end{center}
\end{figure}

\bigskip
\noindent\textbf{Verification of the GHZ state}\\
Next, consider   the  GHZ state  $|G^n\>=(|0\>^{\otimes n}+| 1\>^{\otimes n})/\sqrt{2}$ of $n$-qubits with $n\geq 3$. To verify this state, the $n$ parties can randomly perform certain tests based on local projective measurements. In each test, the verifier (one of the parties) asks  each party to  perform a local projective measurement as specified by a unit vector on the Bloch sphere and return the measurement outcome.  If all parties are honest, then
only the target state $|G^n\>$ can pass all tests with certainty, so  the GHZ state can be verified. In the presence of dishonest parties,  let $\scrD$ be the set of dishonest parties, who know which parties are honest or dishonest and who may collude with each other; let $\scrH$ be the set of honest parties (including the verifier), who do not know which other parties are honest or dishonest. The goal is to verify $|G^n\>$ up to local unitary transformations on the joint Hilbert space of $\scrD$ \cite{Pappa2012,McCut2016}.
Assuming $|\scrD|,|\scrH|\geq 1$, then   $|G^n\>$ may be regarded as a Bell state shared between $\scrH$ and $\scrD$. So the verification of the GHZ state is closely tied to the verification of the Bell state. Actually, there is no essential difference when $|\scrH|=1$.

However, a key distinction arises when $|\scrH|\geq 2$ because each member of $\scrH$  can only perform local projective measurements on his/her party. So the potential tests that the verifier can realize are restricted.
Careful analysis in Supplementary Note 10 shows that only two types of tests for verifying the GHZ state $|G^n\>$
can be  constructed from local projective measurements.
In the first type, all parties perform $Z$ measurements, and the test is passed if they obtain the same outcome.  In this way the verifier can   effectively realize the $Z$ measurement on $V_\scrH$, where $V_\scrH$ is the two-dimensional subspace spanned by $\bigotimes_{j\in \scrH}|0\rangle_j$ and $\bigotimes_{j\in \scrH}|1\rangle_j$.

In the second type of tests, party $j$  performs the $X(\phi_j)$ measurement with $\sum_j \phi_j=0\mod 2\pi$, where
$X(\phi_j)=\rme^{-\rmi \phi_j}|0\rangle\langle 1|+\rme^{\rmi \phi_j}|1\rangle\langle 0|$
corresponds to the Bloch vector $(\cos\phi_j,\sin\phi_j,0)^\rmT$, and each $\phi_j$ is decided by the verifier. 
The test is passed if the number of outcomes $-1$ is even. Suppose $\phi_1,\phi_2,\ldots, \phi_{n}$ are chosen independently and uniformly at random from the interval $[0,2\pi)$. Then   $\phi_\scrH:=\sum_{j\in \scrH}\phi_j \mod 2\pi$ is uniformly distributed in $[0,2\pi)$. Given $\phi\in[0,2\pi)$,  the average of  $\bigotimes_{j\in \scrH}  X(\phi_j)$ under the condition $\phi_\scrH =\phi$ reads
\begin{align}
\biggl\<\bigotimes_{j\in \scrH} X(\phi_j)\biggr\>_\phi=\rme^{-\rmi\phi}\bigotimes_{j\in \scrH}(|0\>\<1|)_j  +\rme^{\rmi\phi}\bigotimes_{j\in \scrH}(|1\>\<0|)_j.
\end{align}
In this way, the verifier can   effectively realize the $X(\phi)$ measurement on  $V_\scrH$, where $\phi$ is completely random. A similar result holds when  $\phi_j$ are chosen independently and uniformly at random from the discrete set $\{2k\pi/M\}_{k=0}^{M-1}$ with $M\geq 3$ being a positive integer.

By the above analysis, the verifier can effectively realize projective measurements along the $z$-axis or on the $xy$-plane when represented on the Bloch-sphere of $V_\scrH$, but not other projective measurements (assuming $|\scrH|\geq 2$).  Each verification protocol of the GHZ state corresponds to a probability distribution $\mu$ on the Bloch sphere that is supported on  the equator together with the north and south poles. Moreover, for  all protocols in \tref{tab:ProtocolGHZDH} except for the isotropic protocol (cf. Supplementary Notes 9 and 10), the guessing probabilities are the same as in the verification of the Bell state. 
To be specific, $\gamma(C,\mu)$ and $\gamma^\rmF(F,\mu)$ can be defined as before; \thsref{thm:MGP}, \ref{thm:gammaF}, and \lref{lem:gammaCHighDim} still hold, except that now $C$ refers to the bipartite concurrence between $\scrH$ and $\scrD$. Although variants of the $XY$ protocol and equator protocol were proposed previously \cite{Pappa2012,McCut2016},  such exact formulas for the guessing probabilities are not known in the literature.
To optimize the performance, $\mu$ should be uniform on the equator, which leads to the equator$+Z$ protocol; the optimal probability $p_Z$ for performing the $Z$ measurement depends on $C$ or $F$ as before.

A quantum state $\rho$ is genuinely multipartite entangled (GME) if its fidelity with the GHZ state $\tr(\rho|G^n\>\<G^n|)$ is larger than $1/2$ \cite{Guhne2009}. The  GME can be certified if the guessing probability surpasses the detection threshold $\gamma^\rmF(1/2)=\gamma^\ast$.
This threshold is minimized at the special equator$+Z$ protocol with $p_Z=4/(4+\pi^2)\approx0.288$, in which case we have
\begin{equation}\label{eq:gammastar}
\gamma^\rmF(1/2)=\gamma^\ast=\frac{1}{2}+\frac{1}{\sqrt{4+\pi^2}}\approx 0.769.
\end{equation}
This threshold is only $2.5\%$ higher than the optimal threshold $3/4$ for certifying the entanglement of the Bell state based on the isotropic protocol.

The sample efficiency for verifying the GHZ state can be determined following a similar analysis applied to the Bell state. The formula in \eref{eq:NumberTestSDI} still applies, except that the choice of verification protocols is restricted. 
Now the minimum of $\gamma^\ast$ is achieved at a special equator$+Z$ protocol  [cf.~\eref{eq:gammastar}]. So the GHZ state can be verified in the SDI scenario with almost the same efficiency as in the standard QSV \cite{li2020GHZ}, as illustrated in Fig.~\ref{fig:QSVDISDI}. In the DI scenario, by contrast, it is in general impossible to achieve such a high efficiency unless one can construct a Bell inequality for which the quantum bound coincides with the algebraic bound \cite{Dimic2021}. Notably, 
the three-qubit GHZ state can be verified with such a high efficiency by virtue of the Mermin inequality \cite{Kaniewski2016,Dimic2021}.

It should be pointed out that  all our  protocols for  verifying  GHZ states are applicable even in  the presence of an arbitrary number of  dishonest parties as long as the verifier is honest. Meanwhile, these protocols are useful for detecting GME. For some cryptographic tasks such as anonymous quantum  communication, the security for all honest parties can be guaranteed at the same time with the assistance of a trusted common random source (CRS) \cite{Pappa2012}. In this case, the number of honest parties can affect the security parameter.

\section{DISCUSSION}

We proposed a simple and practical approach for verifying the Bell state in an untrusted quantum network in which one party is not honest. We also established a simple connection between verification protocols of the Bell state and probability distributions on the Bloch sphere together with an intuitive geometric picture. Based on this connection, we derived  simple formulas for the guessing probability as  functions of the concurrence and reduced fidelity.
Meanwhile, we clarified the sample efficiency  of each verification protocol and showed that the sample efficiency is determined by
the threshold in the  guessing probability. Moreover, we constructed the optimal and simplest  protocols for verifying the Bell state, which  are also very useful to detecting entanglement in the untrusted network.

Furthermore, we reduce the verification problem of  GHZ states to the counterpart  of the Bell state, which enables us to construct the optimal protocol for verifying  GHZ states and for detecting GME. Our work shows that both Bell state and GHZ states can be verified in the SDI scenario with the same sample complexity as in standard QSV. By contrast, the sample complexity in the DI scenario is in general quadratically worse. 
This work is  instrumental to
verifying entangled states in untrusted quantum networks, which is crucial to guaranteeing  the proper functioning of quantum networks.  In addition, this work is of intrinsic interest to the foundational studies on quantum steering. In the future, it would be desirable to generalize our results to generic bipartite pure states, stabilizer states, and other  quantum states.

\bigskip

\noindent\textbf{ACKNOWLEDGEMENTS}\\
This work is  supported by  the National Natural Science Foundation of China (Grants No.~11875110 and No.~62101600) and  Shanghai Municipal Science and Technology Major Project (Grant No.~2019SHZDZX01).\\

\twocolumngrid

\newpage

\clearpage

\setcounter{equation}{0}
\setcounter{figure}{0}
\setcounter{lemma}{0}

\renewcommand{\figurename}{Supplementary Figure}
\renewcommand{\theequation}{S\arabic{equation}}
\renewcommand{\thetable}{S\arabic{table}}
\renewcommand{\thetheorem}{S\arabic{theorem}}
\renewcommand{\thelemma}{S\arabic{lemma}} 

\def\eqref#1{\textup{(\ref{#1})}}  

\newtheorem{conjecture}{Conjecture}

\newpage
\section*{Supplementary Material}
In this Supplementary Material we provide rigorous proofs for the Theorems and Lemmas  presented in the main text together with some  auxiliary results. First, we prove several key results on the guessing probability, including \esref{eq:gammaTr}, \eqref{eq:Tr}, and \thref{thm:MGP} in the main text. Then we prove 
\lref{lem:gammaCHighDim} and \thref{thm:Mixed} after summarizing  the basic properties of the guessing probability.
Next, we consider the fidelity as the figure of merit and prove \thref{thm:gammaF}. Next,
we determine the sample efficiencies of our verification protocols and explain the significance of the guessing probability threshold. Next, we compare  the sample efficiency of semi-device-independent (SDI)  quantum state verification (QSV) with standard QSV and device-independent (DI) QSV based on self-testing.  Next, we propose a number of concrete protocols for verifying the Bell state, including protocols based on polygons and platonic solids; the properties of these protocols are discussed in detail. Finally, we
establish an intimate connection between the verification of GHZ states and the verification of the Bell state, which enables us to construct efficient and even optimal protocols for verifying GHZ states. In the course of study, we introduce the concept of compatible measurements, which is of interest beyond the main focus of this work.

\section*{\label{asec:GuessProbApp}SUPPLEMENTARY Note 1: Guessing probability}
Here we prove several key results on the guessing probability, namely, \esref{eq:gammaTr}, \eqref{eq:Tr}, and \thref{thm:MGP} in the main text.
\begin{proof}[Proof of \eref{eq:gammaTr}]
	Recall that any two-qubit state has the form
	\begin{equation}
	\rho=\frac{1}{4}\biggl(\mathbb{I} +\bm{a}\cdot \bm{\sigma} \otimes \mathbb{I}+\mathbb{I}\otimes \bm{b}\cdot \bm{\sigma}+\sum_{j,k}T_{j,k}\sigma_j\otimes \sigma_k\biggr). \label{eq:twoqubitApp}
	\end{equation}
	After Alice performs the projective measurement $\bm{r}\cdot \bm{\sigma}$, the unnormalized reduced states of Bob associated with the two outcomes read
	\begin{equation}
	\rho_{\pm}=\tr_A(\rho P_{\pm})=\frac{1}{4}\biggl(\mathbb{I}\pm \bm{a}\cdot \bm{r}\mathbb{I}+\bm{b}\cdot \bm{\sigma}\pm \sum_{j,k}T_{jk}r_j\sigma_k\biggr),\;
	\end{equation}
	which implies that
	\begin{equation}
	\rho_+-\rho_-=\frac{1}{2}\biggl(\bm{a}\cdot \bm{r}\mathbb{I}+ \sum_{j,k}T_{jk}r_j\sigma_k\biggr).
	\end{equation}
	Therefore,
	\begin{align}
	\gamma(\rho,\bm{r})&=\frac{1}{2}+\frac{1}{2}\|\rho_+-\rho_-\|_1\nonumber\\
	&=\frac{1}{2}+\frac{1}{4}\biggl\|\bm{a}\cdot \bm{r}\mathbb{I}+ \sum_{j,k}T_{jk}r_j\sigma_k\biggr\|_1\nonumber\\
	&=\frac{1}{2}\bigl(1+\max\bigl\{|\bm{a}\cdot \bm{r}|,\|T^\rmT\bm{r}\| \bigr\}\bigr).
	\end{align}
	
	If  $\rho$ is pure, then $\rho_+-\rho_-$ has at most one positive eigenvalue and one negative eigenvalue.
	So $|\bm{a}\cdot \bm{r}|\leq \|T^\rmT\bm{r}\|$ and
	\begin{align}
	\gamma(\rho,\bm{r})&=\frac{1}{2}\bigl(1+\|T^\rmT\bm{r}\|\bigr),
	\end{align}
	which confirms \eref{eq:gammaTr}.
\end{proof}

\begin{proof}[Proof of \eref{eq:Tr}]
	If $\rho$ is a pure state with concurrence $C$, then its correlation matrix $T$ has three singular values $1, C, C$, and so does $T^\rmT$. Accordingly, $TT^\rmT$ has three eigenvalues $1, C^2, C^2$. The semi-major axis $\bm{v}$ of the correlation ellipsoid happens to be an eigenvector of $TT^\rmT$ with eigenvalue 1. Therefore,
	\begin{align}
	&\bigl\|\sqrt{T T^\rmT}\lsp\bm{r}\bigr\|^2=
	\|T^\rmT\bm{r}\|^2=\bm{r}^\rmT TT^\rmT\bm{r}\nonumber\\
	&=
	(\bm{r}\cdot\bm{v})^2+C^2[1-(\bm{r}\cdot\bm{v})^2]=C^2+(1-C^2)(\bm{r}\cdot\bm{v})^2,
	\end{align}
	which implies \eref{eq:Tr}.   
\end{proof}

\begin{proof}[Proof of \thref{thm:MGP}]
	Given any unit vector $\bm{v}$ on the Bloch sphere, there exists a two-qubit pure state $\rho$ with concurrence $C$ and correlation matrix $T$ such that $\bm{v}$ is an eigenvector of $TT^\rmT$ associated with the eigenvalue 1. In other words, the semi-major axis of the correlation ellipsoid of $\rho$ can be chosen to coincide with $\bm{v}$. So  \thref{thm:MGP} follows from \esref{eq:gammaTr} and \eqref{eq:Tr} in the main text. 
\end{proof}

\Thref{thm:MGP} implies that
\begin{align}
\gamma_2(C,\mu)&=\max_\rho \{\gamma(\rho,\mu)|C(\rho)\leq C\}\nonumber\\
&=\max_\rho \{\gamma(\rho,\mu)|C(\rho)= C\},
\end{align}
where both maximizations are taken over two-qubit pure states. This result corroborates the intuition that the maximum guessing probability is nondecreasing with the concurrence.

\section*{\label{asec:BasicProp}SUPPLEMENTARY Note 2: Basic properties of the guessing probability}
Here  we summarize the basic properties of the guessing probability $\gamma_2(C,\mu)$ and the function $g(C,\mu)$.

\begin{lemma}\label{lem:Convex}
	Suppose $0\leq C\leq 1$; then $g(C,\mu)$ is nondecreasing in $C$; in addition, $g(C,\mu)$ is convex in $C$ and $\mu$, respectively, that is,
	\begin{align}
	g(p_1 C_1+p_2C_2, \mu)&\leq p_1 g(C_1, \mu)+ p_2 g(C_2,\mu),\\
	g(C, p_1 \mu_1+p_2\mu_2)&\leq p_1 g(C, \mu_1)+ p_2 g(C, \mu_2),
	\end{align}
	where $p_1,p_2\geq 0$, $p_1+p_2=1$, $0\leq C_1, C_2\leq 1$, and $\mu_1,\mu_2$ are two probability distributions  on the Bloch sphere. Similarly, $\gamma_2(C,\mu)$ is nondecreasing in $C$ and  is convex in $C$ and $\mu$, respectively.
\end{lemma}
\Lref{lem:Convex} follows from \esref{eq:gCmu} and \eqref{eq:gammaCmu} in \thref{thm:MGP}. The monotonicity of $\gamma_2(C,\mu)$ is also clear from its definition.

If $\mu_1$ and $\mu_2$ can be turned into each other by an orthogonal transformation, then $g(C, \mu_1)=g(C, \mu_2)$, so
\begin{equation}
g(C, p_1 \mu_1+p_2\mu_2)\leq g(C, \mu_1).
\end{equation}
Therefore, $g(C,\mu)$ for given $C$ is minimized when $\mu$ is the uniform distribution on the Bloch sphere.

The \emph{verification matrix} of the distribution $\mu$ is defined as
\begin{equation}
\Xi(\mu):= \int\rmd\mu(\bm{r})\bm{r} \bm{r}^\rmT \label{eq:Ximu},
\end{equation}
which is a $3\times 3$ positive-semidefinite  matrix with trace~1, that is, $\tr[\lsp\Xi(\mu)]=1$. The operator norm
(the largest eigenvalue) of  $\Xi(\mu)$ is bounded from below by $1/3$, that is, $\|\Xi(\mu)\|\geq 1/3$. The distribution $\mu$ and the corresponding protocol are \emph{balanced} if $\Xi(\mu)$ is proportional to the identity matrix, in which case the lower bound $\|\Xi(\mu)\|\geq 1/3$ is saturated.

As we shall see shortly, the properties of $g(C, \mu )$ is closely tied to the verification matrix $\Xi(\mu)$. For example, an upper bound  for   $g(C,\mu)$ can be constructed from $\Xi(\mu)$, as shown in  the following lemma.
\begin{lemma}\label{lem:gCUB}
	Suppose $0\leq C\leq 1$; then
	\begin{align}
	g(C, \mu)\leq \sqrt{C^2+(1-C^2)\|\Xi(\mu)\|}. \label{eq:gCUB}
	\end{align}
	The inequality is saturated if there exists a unit vector $\bm{v}$ in  the eigenspace of $\Xi(\mu)$ associated with the largest eigenvalue  and  $|\bm{r}\cdot \bm{v}|$ is a constant in the support of $\mu$ except for a set of measure zero. Any such unit vector, if it exists, represents an intelligent direction. 
\end{lemma}

\begin{proof}[Proof of \lref{lem:gCUB}]
	Let $\bm{v}$ be an arbitrary unit vector in dimension 3. The upper bound in \eref{eq:gCUB} follows from \eref{eq:gCmu} in the main text  and the following equation
	\begin{align}
	&\int\rmd\mu(\bm{r})\sqrt{C^2+(1-C^2)(\bm{r}\cdot\bm{v})^2}\nonumber\\
	&\leq \sqrt{C^2+(1-C^2)\int\rmd\mu(\bm{r})(\bm{r}\cdot\bm{v})^2}\nonumber\\
	&= \sqrt{C^2+(1-C^2)\bm{v}^\rmT\Xi\bm{v}}\nonumber\\
	&\leq  \sqrt{C^2+(1-C^2)\|\Xi\|},
	\end{align} 
	where $\Xi$ is an abbreviation of $\Xi(\mu)$. Here the first inequality is due to the concavity of the square-root function and is
	saturated iff
	$|\bm{r}\cdot \bm{v}|$ is a constant in the support of $\mu$ except for a set of measure zero;  the second inequality is saturated iff $\bm{v}$ is an eigenvector of $\Xi$ with the largest eigenvalue $\|\Xi\|$. So the inequality in \eref{eq:gCUB} is saturated iff there exists a unit vector $\bm{v}$ that belongs to the eigenspace of $\Xi$ associated with the largest eigenvalue  and  $|\bm{r}\cdot \bm{v}|$ is a constant in the support of $\mu$ except for a set of measure zero. In addition, any such unit vector, if it exists, represents an intelligent direction.
\end{proof}

By contrast, a  lower bound for $g(C, \mu )$ can be constructed by choosing  $\bm{v}$ as  an eigenvector of $\Xi(\mu)$ associated with the largest eigenvalue. Alternative lower bounds are presented in the following lemma.
\begin{lemma}\label{lem:gCmuLB}
	Suppose $0\leq C\leq 1$; then
	\begin{align}
	g(C,\mu)&\geq \|\Xi\|+(1-\|\Xi\|)C\geq \frac{1+2C}{3}\geq C, \label{eq:gCmuLB}\\
	\gamma_2(C,\mu)&\geq \frac{1+\|\Xi\|+(1-\|\Xi\|)C}{2}\geq \frac{2+C}{3}\geq C. \label{eq:gammaCmuLB}
	\end{align}
\end{lemma}
\begin{proof}
	\Eref{eq:gCmuLB} can be derived as follows,
	\begin{align}
	g(C,\mu)&=\max_{\bm{v}} \int\rmd\mu(\bm{r})\sqrt{C^2+(1-C^2)(\bm{r}\cdot\bm{v})^2}\nonumber\\
	&\geq \max_{\bm{v}} \bigl[(1-\bm{v}^\rmT\Xi\bm{v})C+\bm{v}^\rmT\Xi\bm{v}\bigr]  \nonumber\\
	& =(1-\|\Xi\|)C+\|\Xi\|\geq \frac{1+2C}{3}\geq C.
	\end{align}
	Here the first inequality follows from the concavity of  the square-root function and the constraint $0\leq (\bm{r}\cdot\bm{v})^2\leq 1$;   the second inequality follows from the facts $\|\Xi\|\geq 1/3$ and $0\leq C\leq 1$.
	
	\Eref{eq:gammaCmuLB} follows from \eref{eq:gCmuLB} and the equality $\gamma_2(C,\mu)=[1+g(C,\mu)]/2$ in \eref{eq:gammaCmu} in the main text.
\end{proof}

The bound in \eref{eq:gCmuLB} is nearly tight
when $C\rightarrow 1$. To see this point,
note that
\begin{align}
&\int\rmd\mu(\bm{r})\sqrt{C^2+(1-C^2)(\bm{r}\cdot\bm{v})^2}\nonumber\\
&\approx \int\rmd\mu(\bm{r}) \Bigl\{1-\frac{1}{2}[1-(\bm{r}\cdot\bm{v})^2](1-C^2)\Bigr\}\nonumber\\
&=1-\frac{1}{2}(1-\bm{v}^\rmT\Xi\bm{v})(1-C^2)\approx 1-(1-\bm{v}^\rmT\Xi\bm{v})(1-C)\nonumber\\ &
=[1-\bm{v}^\rmT\Xi\bm{v}]C+\bm{v}^\rmT\Xi\bm{v},
\end{align}
where we have kept the first-order approximation in $1-C$. This equation implies that
\begin{equation}\label{eq:gClim1}
g(C,\mu)\approx\|\Xi\|+(1-\|\Xi\|)C,
\end{equation}
so the bound in \eref{eq:gCmuLB} is nearly tight
when $C\rightarrow 1$. In the limit $C\rightarrow 1$, \eref{eq:gClim1} implies that  $g(C,\mu)$ is minimized when $\|\Xi\|=1/3$, that is, $\Xi=\mathbb{I}/3$, so that the verification protocol is balanced.

\begin{lemma}\label{lem:gCslope}
	Suppose $0\leq C_1\leq C_2\leq 1$; then
	\begin{equation}
	g(C_2,\mu)-g(C_1,\mu)\leq \frac{2}{3}(C_2-C_1). \label{eq:gCslope}
	\end{equation}
\end{lemma}
\begin{proof}
	When $C_2=1$, we have $g(C_2)=1$, so \eref{eq:gCslope} reduces to the  equation
	\begin{equation}
	1-g(C_1,\mu)\leq \frac{2}{3}(1-C_1) \label{eq:SlopeProof}
	\end{equation}
	and so  follows from \eref{eq:gCmuLB}.
	
	When $0\leq C_1\leq C_2<1$, \eref{eq:gCslope} can be derived as follows
	\begin{align}
	&g(C_2,\mu)-g(C_1,\mu)\leq \frac{C_2-C_1}{1-C_1}[g(C=1,\mu)-g(C_1,\mu)]\nonumber\\
	&=\frac{C_2-C_1}{1-C_1}[1-g(C_1,\mu)]\leq \frac{C_2-C_1}{1-C_1}\times\frac{2}{3}(1-C_1)\nonumber\\
	&=\frac{2}{3}(C_2-C_1).
	\end{align}
	Here the first inequality follows from the facts that $g(C,\mu)$ is nondecreasing and convex in $C$; the second inequality follows from \eref{eq:SlopeProof}.
\end{proof}

\begin{lemma}\label{lem:gCmup}
	Suppose $0\leq C\leq p\leq 1$ and $p>0$; then
	\begin{gather}
	p g\Bigl(\frac{C}{p},\mu\Bigr)\leq g(C,\mu), \label{eq:gCmup}\\
	p \gamma_2\Bigl(\frac{C}{p},\mu\Bigr)\leq \gamma_2(C,\mu). \label{eq:gammaCmup}
	\end{gather}
\end{lemma}

\begin{proof}
	According to \lref{lem:gCslope}, we have
	\begin{align}
	g\Bigl(\frac{C}{p},\mu\Bigr)\leq g(C,\mu)+\frac{2}{3}\Bigl(\frac{C}{p}-C\Bigr),
	\end{align}
	so 
	\begin{align}
	&pg\Bigl(\frac{C}{p},\mu\Bigr)-g(C,\mu)\leq-(1-p)g(C,\mu)+\frac{2}{3}(1-p)C\nonumber\\
	&=-(1-p)\Bigl[g(C,\mu)-\frac{2}{3}C\Bigr]\leq -\frac{1-p}{3},
	\end{align}
	which implies \eref{eq:gCmup}. Here the last inequality follows from \lref{lem:gCmuLB}.
	
	\Eref{eq:gammaCmup} follows from \eref{eq:gCmup} and the equality $\gamma_2(C,\mu)=[1+g(C,\mu)]/2$
	in \eref{eq:gammaCmu} in the main text.
\end{proof}

\begin{lemma}\label{lem:g*(mu)=2max}
	Suppose the strategy $\mu$ satisfies
	\begin{align}\label{eq:CenterSymProt}
	\int\rmd\mu(\bm{r})\bm{r}=0.
	\end{align}
	Then $g^\ast(\mu)$ defined in \eref{eq:gC0} in the main text can be expressed as
	\begin{align}\label{eq:g*(mu)=2max}
	g^\ast(\mu)=2\,\max_{R}\,\left|\int_{R}\rmd\mu(\bm{r}) \bm{r}\right|=2\max_{\bm{v}} \left|\int_{\bm{r}\cdot\bm{v}\geq 0}\rmd\mu(\bm{r}) \bm{r}\right|,
	\end{align}
	where the first maximization is over all measurable subsets of  the Bloch sphere, and the second maximization is over all unit vectors on the Bloch sphere. In addition, $\bm{v}$ is an intelligent direction at $C=0$ iff it satisfies the following two conditions,   
	\begin{align}\label{eq:vparamax}
	\bm{v} \,\mathbin{\!/\mkern-5mu/\!}\,  \int_{\bm{r}\cdot\bm{v}\geq 0}\rmd\mu(\bm{r}) \bm{r},\quad 2\left|\int_{\bm{r}\cdot\bm{v}\geq 0}\rmd\mu(\bm{r}) \bm{r}\right|=g^\ast( \mu ).
	\end{align}
\end{lemma}

\begin{proof}
	Let $R$ be any measurable subset of the Bloch sphere and 
	$\overline{R}$ its  complement; then 
	we have  
	\begin{align}
	\int_{R}\rmd\mu(\bm{r}) \bm{r} = -\int_{\overline{R}}\rmd\mu(\bm{r}) \bm{r}
	\end{align}
	by the  assumption \eref{eq:CenterSymProt}. 
	Therefore,
	\begin{align}
	&\int\rmd\mu(\bm{r}) |\bm{r}\cdot\bm{v}|
	=\int_{\bm{r}\cdot\bm{v}\geq0}\rmd\mu(\bm{r}) \bm{r}\cdot\bm{v} -
	\int_{\bm{r}\cdot\bm{v}<0}\rmd\mu(\bm{r}) \bm{r}\cdot\bm{v}  \nonumber\\
	&=2\bm{v}\cdot\int_{\bm{r}\cdot\bm{v}\geq 0}\rmd\mu(\bm{r}) \bm{r}\geq  2\bm{v}\cdot\int_{R}\rmd\mu(\bm{r}) \bm{r}; \label{eq:int|rv|dmu0}
	\end{align} 
	here the inequality is saturated if $R$ is the region determined by the inequality $\bm{r}\cdot\bm{v}\geq 0$. 
	It follows  that
	\begin{align}\label{eq:int|rv|dmu}
	\int\rmd\mu(\bm{r}) |\bm{r}\cdot\bm{v}|
	&=2\max_{R}\left[\bm{v}\cdot\int_{R}\rmd\mu(\bm{r}) \bm{r}\right].
	\end{align}
	
	By plugging \eref{eq:int|rv|dmu} into \eref{eq:gC0}  in the main text we obtain
	\begin{align}
	&g^\ast( \mu)=\max_{\bm{v}}\int\rmd\mu(\bm{r}) |\bm{r}\cdot\bm{v}|       =2\,\max_{\bm{v},R}\left[\bm{v}\cdot\int_{R}\rmd\mu(\bm{r}) \bm{r}\right] \nonumber\\
	&=2\,\max_{R}\,\left|\int_{R}\rmd\mu(\bm{r}) \bm{r}\right|,
	\end{align} 
	which confirms the first equality in  \eref{eq:g*(mu)=2max}. 
	The second equality in \eref{eq:g*(mu)=2max} follows from the equation below,
	\begin{align}
	&2\max_{\bm{v}} \left|\int_{\bm{r}\cdot\bm{v}\geq 0}\rmd\mu(\bm{r}) \bm{r}\right|\leq 2\,\max_{R}\,\left|\int_{R}\rmd\mu(\bm{r}) \bm{r}\right|=g^\ast( \mu )\nonumber\\
	&=\max_{\bm{v}}\int\rmd\mu(\bm{r}) |\bm{r}\cdot\bm{v}|=2\max_{\bm{v}}\bm{v}\cdot\int_{\bm{r}\cdot\bm{v}\geq 0}\rmd\mu(\bm{r}) \bm{r}\nonumber\\
	&\leq 2\max_{\bm{v}} \left|\int_{\bm{r}\cdot\bm{v}\geq 0}\rmd\mu(\bm{r}) \bm{r}\right|.
	\end{align}
	
	According to   \eqref{eq:int|rv|dmu0},  $\bm{v}$ is an intelligent direction iff 
	\begin{equation}
	2\bm{v}\cdot\int_{\bm{r}\cdot\bm{v}\geq 0}\rmd\mu(\bm{r}) \bm{r}=g^\ast(\mu).
	\end{equation}
	In view of \eref{eq:g*(mu)=2max}, this condition holds iff the conditions in \eref{eq:vparamax} hold. 
\end{proof}

\section*{\label{asec:HigmDim}SUPPLEMENTARY Note 3: Proof of \lref{lem:gammaCHighDim}}
Here we prove \lref{lem:gammaCHighDim} presented in the main text, which shows that Bob cannot increase the guessing probability by preparing  a higher-dimensional state $\rho$ whose local support for Alice is not contained in the local support of the target Bell state.

Let $\rho_A=\tr_B(\rho)$ and let
\begin{equation}
P_A=|0\>\<0|+|1\>\<1| \label{eq:PA}
\end{equation}
be the projector onto the local support of the target Bell state. Before proving \lref{lem:gammaCHighDim}, it is instructive to give formal mathematical definitions of $\gamma_2(C,\mu)$ and $\gamma(C,\mu)$:
\begin{align}
\gamma_2(C,\mu)&=\max_{\rho}\{\gamma(\rho,\mu)| C(\rho)\leq C, \rho^2=\rho, P_A\rho_A=\rho_A\},\\
\gamma(C,\mu)&=\max_{\rho}\{\gamma(\rho,\mu)| C(\rho)\leq C, \rho^2=\rho\}.
\end{align}
Here the constraint $\rho^2=\rho$ means $\rho$ is a pure state, while the constraint $P_A\rho_A=\rho_A$ means $\rho_A$ is supported in the support of $P_A$ (the local support of the target Bell state). By definition we have
\begin{equation}
\gamma(C,\mu)\geq \gamma_2(C,\mu).
\end{equation}
To prove the equality $\gamma(C,\mu)= \gamma_2(C,\mu)$, it suffices to derive the opposite inequality.

\begin{proof}[Proof of \lref{lem:gammaCHighDim}]
	If Alice performs the projective measurement $\{P_A,\mathbb{I}-P_A\}$ on the state $\rho$, then the probability of obtaining the first outcome is
	$q=\tr[(P_A\otimes \mathbb{I}_B)\rho]$. When $q>0$, after obtaining this outcome, the state $\rho$ turns into
	\begin{equation}
	\rho'=\frac{1}{q}(P_A\otimes \mathbb{I}_B)\rho(P_A\otimes \mathbb{I}_B).
	\end{equation}
	Note that $\tr_B(\rho')$ is supported in the support of $P_A$. In addition, the concurrence of $\rho'$ obeys the upper bound
	\begin{align}\label{eq:Crhop}
	C(\rho')\leq \min\{C(\rho)/q,1\}.
	\end{align}
	
	When $q>0$ and $C(\rho)\leq q$, the guessing probability of Bob satisfies
	\begin{align}
	\gamma(\rho,\mu)&=q\gamma(\rho',\mu)\leq
	q\gamma_2(C(\rho'),\mu)\leq q\gamma_2(C(\rho)/q,\mu)\nonumber\\
	&\leq \gamma_2(C(\rho),\mu). \label{eq:gammarhoProof1}
	\end{align}
	Here the first inequality follows from the definition of $\gamma_2(C,\mu)$; the second inequality follows from the inequality $C(\rho')\leq C(\rho)/q$ and the monotonicity of $\gamma_2(C,\mu)$ in $C$ (cf. \lref{lem:Convex}); the  third inequality follows from \eref{eq:gammaCmup} in \lref{lem:gCmup}.
	
	When $0<q\leq C(\rho)\leq  1$, the guessing probability of Bob satisfies
	\begin{align}
	\gamma(\rho,\mu)&=q\gamma(\rho',\mu)\leq q\leq C(\rho)\leq \gamma_2(C(\rho),\mu), \label{eq:gammarhoProof2}
	\end{align}
	where the first inequality follows from the fact that $\gamma_2(C,\mu)\leq 1$ and the last inequality follows from \eref{eq:gammaCmuLB} in \lref{lem:gCmuLB}. 
	
	\Esref{eq:gammarhoProof1} and \eqref{eq:gammarhoProof2} together imply that
	\begin{align}
	\gamma(\rho,\mu)\leq \gamma_2(C(\rho),\mu).  \label{eq:gammarhoProof3}
	\end{align}
	Note that this equation holds even if $q=0$ in which case $\gamma(\rho,\mu)=0$. As a corollary of \eref{eq:gammarhoProof3}, we can deduce
	\begin{align}
	&\gamma(C,\mu)=\max_{\rho}\{\gamma(\rho,\mu)| C(\rho)\leq C, \rho^2=\rho\}\nonumber\\
	&\leq \max_{\rho}\{\gamma_2(C(\rho),\mu)| C(\rho)\leq C\}=\gamma_2(C,\mu). \label{eq:gammarhoproof3}
	\end{align}
	Here the last equality follows from the monotonicity of $\gamma_2(C,\mu)$ in $C$ (cf. \lref{lem:Convex}).
	Since $\gamma(C,\mu)\geq \gamma_2(C,\mu)$ by definition, \eref{eq:gammarhoproof3} implies that
	$\gamma(C,\mu)=\gamma_2(C,\mu)$ for $0\leq C\leq 1$, which confirms \lref{lem:gammaCHighDim}.
\end{proof}

\section*{\label{asec:Mix}SUPPLEMENTARY Note 4: Proof of \thref{thm:Mixed}}
Here we prove \thref{thm:Mixed} presented in the main text, which determines the maximum guessing probability for mixed states.

\begin{proof}[Proof of \thref{thm:Mixed}]
	The third equality in \eref{eq:gammahat}  in the theorem follows from \eref{eq:gammaC0}  in the main text and the equality $\gamma^\ast(\mu)=\gamma_2^\ast(\mu)$.

	The second equality in \eref{eq:gammahat} is equivalent to the following equality
	\begin{align}
	\hat{\gamma}_2(C,\mu)=
	(1-C)\gamma_2^\ast(\mu) +C  \label{eq:gammahat2}
	\end{align}
	given that $\gamma^\ast(\mu)=\gamma_2^\ast(\mu)$. 
	Suppose $\rho$ is the state prepared by Bob and $\rho_A=\tr_B(\rho)$ is supported in the support of $P_A$ defined in \eref{eq:PA}. Let $\rho=\sum_j q_j \rho_j$ be any convex decomposition of $\rho$ into pure states such that $C(\rho)=\sum_j q_j C_j$, where $C_j=C(\rho_j)$ is the concurrence of $\rho_j$.  Note that $\tr_B(\rho_j)$ is supported in the support of $P_A$ and
	$C_j\leq 1 $ for each $j$. Then the guessing probability of Bob satisfies
	\begin{align}
	\hat{\gamma}(\rho,\mu)&\leq \sum_j q_j \gamma(\rho_j,\mu)\leq \sum_j  q_j \gamma_2(C_j,\mu )\nonumber\\
	&\leq [1-C(\rho)]\gamma_2(0,\mu) +C(\rho)\gamma_2(1,\mu)\nonumber\\
	&= [1-C(\rho)]\gamma_2^\ast(\mu) +C(\rho).
	\label{eq:gammarhoMixProof1}
	\end{align}
	Here the third inequality follows from the convexity of $\gamma_2(C,\mu)$ in $C$ (cf. \lref{lem:Convex}); the last equality follows from the facts that $\gamma_2(1,\mu)=1$ and $\gamma_2^\ast(\mu)=\gamma_2(0,\mu)$.
	As an implication of \eref{eq:gammarhoMixProof1}, we have
	\begin{align}
	\hat{\gamma}_2(C,\mu)&=\max_{\rho}\{\gamma(\rho,\mu)| C(\rho)\leq C, P_A\rho_A=\rho_A\}\nonumber\\
	&\leq \max_{\rho}\{[1-C(\rho)]\gamma_2^\ast(\mu) +C(\rho)| C(\rho)\leq C\}\nonumber\\
	&=(1-C)\gamma_2^\ast(\mu) +C. \label{eq:gammarhoMixProof2}
	\end{align}
	
	To prove the second equality in \eref{eq:gammahat}, it remains to prove the opposite inequality to  \eref{eq:gammarhoMixProof2}, that is,
	\begin{equation}
	\hat{\gamma}_2(C,\mu)\geq(1-C)\gamma_2^\ast(\mu) +C. \label{eq:gammarhoMixProof3}
	\end{equation}
	Let $\rho_1$ be a pure product state such that
	$\tr_B(\rho_1)$ is supported in the support of $P_A$ and that
	$\gamma(\rho_1,\mu)=\gamma_2^\ast(\mu)$. Let $\rho_2$ be a Bell state such that $\tr_B(\rho_2)$ is supported in the support of $P_A$ and that the support of  $\tr_A\rho_2$ is orthogonal to that of  $\tr_A\rho_1$. Then  $\gamma(\rho_2,\mu)=1$ and $C(\rho_2)=1$. Let $\rho=(1-p)\rho_1+p\rho_2$ with $0\leq p\leq 1$; then
	\begin{align}
	C(\rho)&=(1-p)C(\rho_1)+p C(\rho_2)=p,\\
	\gamma(\rho,\mu)&=(1-p)\gamma(\rho_1,\mu)+p\gamma(\rho_2,\mu)\nonumber\\
	&=(1-p)\gamma_2^\ast(\mu)+p.
	\end{align}
	By setting $p=C$, we can deduce \eref{eq:gammarhoMixProof3}, which implies 
	\eref{eq:gammahat2} and the second equality in \eref{eq:gammahat}.

	Finally, let us prove the first equality in \eref{eq:gammahat}. Since $\hat{\gamma}(C,\mu)\geq \hat{\gamma}_2(C,\mu)$ by definition, it remains to prove the opposite inequality.
	
	Suppose $\rho$ is the state prepared by Bob; here the support of $\rho_A=\tr_B(\rho)$ is not restricted in contrast to the above discussions.
	Let $q=\tr[(P_A\otimes \mathbb{I}_B)\rho]$ and
	\begin{equation}
	\rho'=\frac{1}{q}(P_A\otimes \mathbb{I}_B)\rho(P_A\otimes \mathbb{I}_B),\quad q>0;
	\end{equation}
	then $C(\rho')\leq \min\{C(\rho)/q,1\}$.
	When $q>0$ and $C(\rho)\leq q$, the guessing probability of Bob satisfies
	\begin{align}
	\gamma(\rho,\mu)&=q\gamma(\rho',\mu)\leq
	q\hat{\gamma}_2(C(\rho'),\mu)\leq q\hat{\gamma}_2(C(\rho)/q,\mu)\nonumber\\
	&\leq \hat{\gamma}_2(C(\rho),\mu). \label{eq:gammarhoMixProof5}
	\end{align}
	Here the first inequality follows from the definition of $\hat{\gamma}_2(C,\mu)$; the second inequality follows from the inequality $C(\rho')\leq C(\rho)/q$ and the monotonicity of $\hat{\gamma}_2(C,\mu)$ in $C$ according to \eref{eq:gammahat2};  the  third inequality also follows from \eref{eq:gammahat2}.
	
	When $0<q\leq C(\rho)\leq  1$, the guessing probability of Bob satisfies
	\begin{align}
	\gamma(\rho,\mu)&=q\gamma(\rho',\mu)\leq q\leq C(\rho)\leq \hat{\gamma}_2(C(\rho),\mu), \label{eq:gammarhoMixProof6}
	\end{align}
	where the last inequality follows from \eref{eq:gammahat2} again.
	
	\Esref{eq:gammarhoMixProof5} and \eqref{eq:gammarhoMixProof6} together imply that
	\begin{align}
	\gamma(\rho,\mu)\leq \hat{\gamma}_2(C(\rho),\mu). \label{eq:gammarhoMixProof7}
	\end{align}  
	This equation holds even when $q=0$, in which case we have $\gamma(\rho,\mu)=0$. As a corollary of \eref{eq:gammarhoMixProof7} we can deduce 
	\begin{align}
	&\hat{\gamma}(C,\mu)=\max_{\rho}\{\gamma(\rho,\mu)| C(\rho)\leq C\}\nonumber\\
	&\leq \max_{\rho}\{\hat{\gamma}_2(C(\rho),\mu)| C(\rho)\leq C\}=\hat{\gamma}_2(C,\mu).
	\end{align}
	Here  the last equality follows from the monotonicity of $\hat{\gamma}_2(C,\mu)$ in $C$ by \eref{eq:gammahat2}. Given the opposite inequality by definition, we  conclude that $\hat{\gamma}(C,\mu)=\hat{\gamma}_2(C,\mu)$, which completes the proof of \eref{eq:gammahat} and \thref{thm:Mixed}.
\end{proof}

To complement \fref{fig:Concurrence-Pure} in the main text, the relation between the guessing probability and concurrence in the mixed-state scenario is illustrated in Supplementary Figure 1.

\begin{figure}
	\begin{center}
		\includegraphics[width=8cm]{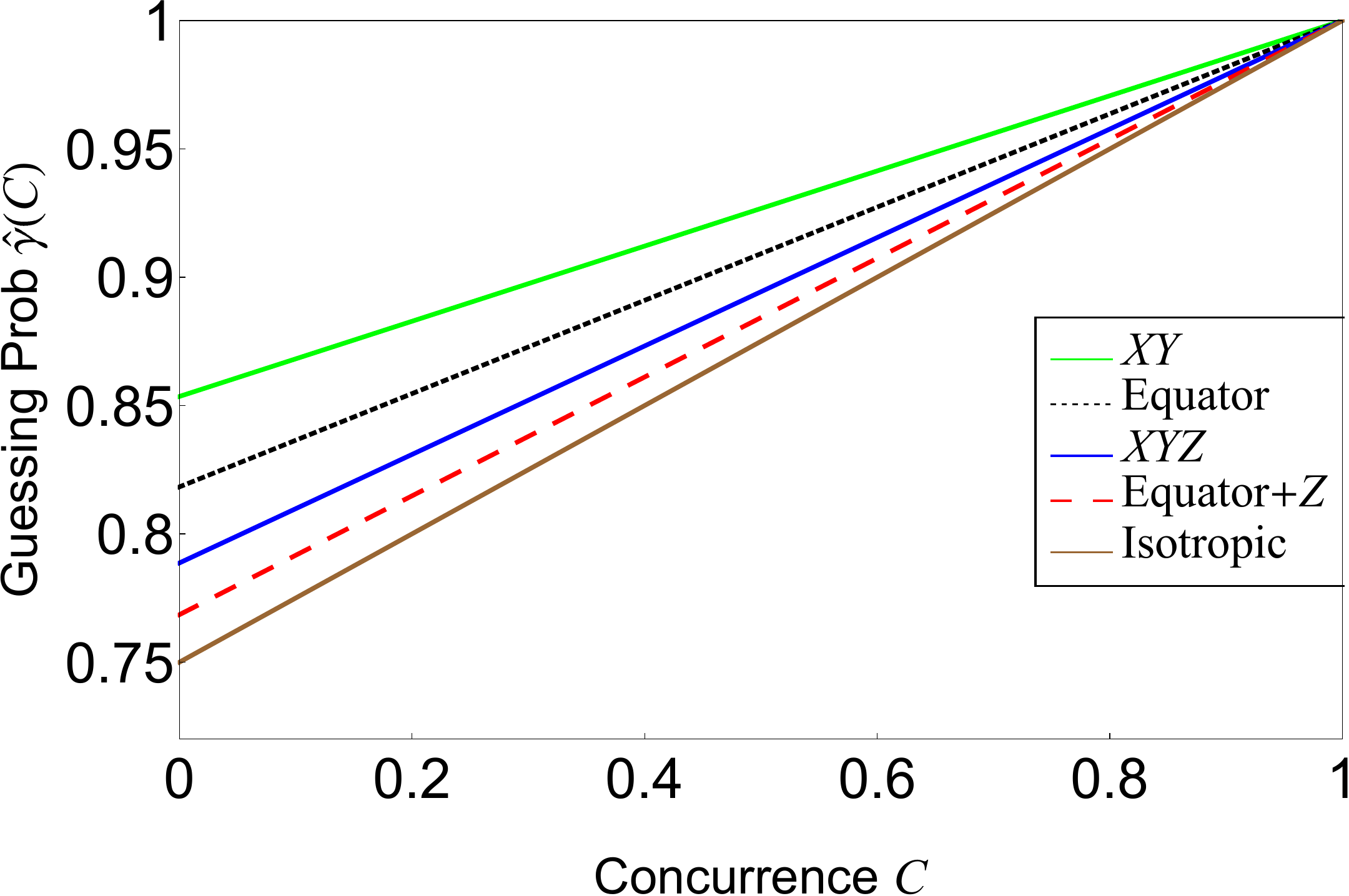}
		\caption{\label{fig:Concurrence-Mixed}
			The  guessing probability $\hat{\gamma}(C)=\hat{\gamma}_2(C)$ as a function of the concurrence $C$ in the mixed-state scenario and for various verification protocols of the Bell state. Here the $XY$ protocol and isotropic protocol are introduced in the main text, while  other protocols are proposed in Supplementary Note 9.    
		}
	\end{center}
\end{figure}

\section*{SUPPLEMENTARY Note 5: Fidelity as the figure of merit}
Here we discuss in more details the properties of the guessing probabilities $\gamma_2^\rmF(F,\mu)$ and $\gamma^\rmF(F,\mu)$ and then prove \thref{thm:gammaF} presented in the main text.

To start with we give formal mathematical definitions of   $\gamma_2^\rmF(F,\mu)$ and $\gamma^\rmF(F,\mu)$:
\begin{align}
\gamma_2^\rmF(F,\mu)&=\max_\rho \{\gamma(\rho,\mu)|F_B(\rho)\leq F, \rho^2=\rho, P_A\rho_A=\rho_A\}, \\
\gamma^\rmF(F,\mu)&=\max_\rho \{\gamma(\rho,\mu)|F_B(\rho)\leq F, \rho^2=\rho\},
\end{align}
where $\rho_A=\tr_B(\rho)$ and $P_A$ is the projector onto the local support of the target Bell state as defined in \eref{eq:PA}. Note that $F_B(\rho)\geq 1/2$ for any two-qubit pure state $\rho$ that satisfies $P_A\rho=\rho$, so $\gamma_2^\rmF(F,\mu)$ is defined only when $1/2\leq F\leq 1$, while $\gamma^\rmF(F,\mu)$ is defined for $0\leq F\leq 1$.

\subsection*{Auxiliary results}

\begin{lemma}\label{lem:gamma2F}
	Suppose $1/2\leq F\leq 1$; then
	\begin{align}
	&\gamma_2^\rmF(F,\mu)=\gamma_2(2F-1,\mu)=\frac{1}{2}+\frac{1}{2}g(2F-1,\mu)  \label{eq:gammaFtilde1}\\
	&=\frac{1}{2}+\frac{1}{2}\max_{\bm{v}}\int\rmd\mu(\bm{r}) \sqrt{(2F-1)^2+4F(1-F)(\bm{r}\cdot\bm{v})^2}.\label{eq:gammaFtilde2}
	\end{align}
	\Lref{lem:gamma2F} in particular implies that
	\begin{align}
	\gamma_2^\rmF(1/2,\mu)&=\gamma_2(0,\mu)=\gamma^\ast(\mu)=\frac{1}{2}+\frac{1}{2}g^\ast(\mu),\\
	\gamma_2^\rmF(1,\mu)&=\gamma_2(1,\mu)=\frac{1}{2}+\frac{1}{2}g(1,\mu)=1.
	\end{align}
\end{lemma}
\begin{proof}
	If $\rho$ is a pure state that satisfies the condition $P_A\rho_A=\rho_A$, then $\rho$ is a two-qubit state and has the form in \eref{eq:twoqubitApp}. In addition, we have
	\begin{equation}
	F_B(\rho)=\frac{1+\|T\|_1}{4}=\frac{1+C}{2}\geq \frac{1}{2},
	\end{equation}
	where $T$ is the correlation matrix appearing in \eref{eq:twoqubitApp} and $\|T\|_1=\tr\sqrt{TT^\rmT}$. This equation implies \eref{eq:gammaFtilde1} in \lref{lem:gamma2F}. \Eref{eq:gammaFtilde2} follows from \eref{eq:gammaFtilde1} above and \thref{thm:MGP} in the main text.
\end{proof}

\begin{lemma}\label{lem:gammaFconvex}
	The guessing probability  $\gamma_2^\rmF(F,\mu)$ is nondecreasing in $F$ for $1/2\leq F\leq 1$; in addition, $\gamma_2^\rmF(F,\mu)$ is convex in $F$ and $\mu$, respectively.
\end{lemma}

\begin{lemma}\label{lem:gammaFmuLB}
	Suppose $1/2\leq F\leq 1$. Then
	\begin{align}
	\gamma_2^\rmF(F,\mu)&\geq \|\Xi\|+(1-\|\Xi\|)F\geq \frac{1+2F}{3}\geq F, \label{eq:gammaFCmuLB}
	\end{align}
	where $\Xi=\Xi(\mu)$ is the verification matrix defined in \eref{eq:Ximu}.
\end{lemma}
\Lref{lem:gammaFconvex} follows from \lsref{lem:Convex} and \ref{lem:gamma2F}.
\Lref{lem:gammaFmuLB} follows from \lsref{lem:gCmuLB} and \ref{lem:gamma2F}

\begin{lemma}\label{lem:qgammaMono}
	Suppose $0< F\leq 1$. Then $q\gamma_2^\rmF(F/q,\mu)$ is nondecreasing in $q$ for $F\leq q\leq \min\{1,2F\}$.
\end{lemma}
\begin{proof}[Proof of \lref{lem:qgammaMono}]
	When  $F=1$,  \lref{lem:qgammaMono} holds trivially since $q$ can only take on the value 1. To prove \lref{lem:qgammaMono} when $0<F<1$, it suffices to prove the following inequality
	\begin{align}
	&q_1\gamma_2^\rmF\Bigl(\frac{F}{q_1},\mu\Bigr)\leq  q_2\gamma_2^\rmF\biggl(\frac{F}{q_2},\mu\biggr)\quad \label{eq:qgammaMono}
	\end{align}
	for $ F\leq q_1< q_2\leq \min\{1,2F\}$.

	Let
	\begin{equation}
	x=\frac{q_2(q_1-F)}{q_1(q_2-F)};
	\end{equation}
	then $0\leq x<1$ and
	\begin{align}
	\frac{F}{q_1}=x\frac{F}{q_2}+1-x,
	\end{align}
	which implies that
	\begin{align}
	\gamma_2^\rmF\biggl(\frac{F}{q_1},\mu\biggr)\leq x \gamma_2^\rmF\biggl(\frac{F}{q_2},\mu\biggr)+1-x,
	\end{align}
	given that $\gamma_2^\rmF(F,\mu)$ is convex in $F$ and $\gamma_2^\rmF(F=1,\mu)=1$
	As a consequence, we have
	\begin{align}
	&q_1\gamma_2^\rmF\Bigl(\frac{F}{q_1},\mu\Bigr)\leq q_1x \gamma_2^\rmF\Bigl(\frac{F}{q_2},\mu\biggr)+q_1(1-x)\nonumber\\
	&=\frac{(q_1-F)q_2}{(q_2-F)}\gamma_2^\rmF\biggl(\frac{F}{q_2},\mu\biggr)+\frac{(q_2-q_1)q_2}{(q_2-F)}\frac{F}{q_2}\nonumber\\
	&\leq\frac{(q_1-F)q_2}{(q_2-F)}\gamma_2^\rmF\biggl(\frac{F}{q_2},\mu\biggr)+\frac{(q_2-q_1)q_2}{(q_2-F)}\gamma_2^\rmF\biggl(\frac{F}{q_2},\mu\biggr)\nonumber\\ &=q_2\gamma_2^\rmF\biggl(\frac{F}{q_2},\mu\biggr),
	\end{align}
	which confirms \eref{eq:qgammaMono} and  implies \lref{lem:qgammaMono}. Here
	the second inequality follows from \eref{eq:gammaFCmuLB}.
\end{proof}

\begin{lemma}\label{lem:gammarhoFUB}
	Suppose $\rho$ is a bipartite state shared between Alice and Bob. Then
	\begin{align}
	\!\!\gamma(\rho,\mu)\leq
	\begin{cases}
	2F_B(\rho)\gamma_2^\rmF(1/2,\mu)  & 0\leq F_B(\rho)\leq  1/2,\\[0.5ex]
	\gamma_2^\rmF(F_B(\rho),\mu) & 1/2\leq F_B(\rho)\leq 1.
	\end{cases}   \label{eq:gammarhoFUB}
	\end{align}
\end{lemma}
\begin{proof}[Proof of \lref{lem:gammarhoFUB}]
	Let $\rho_A=\tr_B(\rho)$,  $q=\tr(P_A \rho_A)$,  and
	\begin{equation}
	\rho'=\frac{1}{q}(P_A\otimes \mathbb{I}_B)\rho(P_A\otimes \mathbb{I}_B),\quad q>0.
	\end{equation}
	If  $q=0$, then
	\begin{align}
	F_B(\rho)&=0,\quad \gamma(\rho,\mu)=0.
	\end{align}
	If $0<q\leq 1$, then $\tr_B(\rho')$ is supported in the support of $P_A$. In addition, we have
	\begin{align}
	F_B(\rho)&=qF_B(\rho')\geq \frac{q}{2},\quad \gamma(\rho,\mu)=q\gamma(\rho',\mu), \label{eq:FBproof1}
	\end{align}
	where the inequality is due to the fact that  $F_B(\rho')\geq 1/2$ since $\rho'$ is a two-qubit pure state with the same local support for Alice as the target Bell state.
	
	If $F_B(\rho)=0$, then $q=0$ and $\gamma(\rho,\mu)=0$, so \eref{eq:gammarhoFUB} holds.

	If $F_B(\rho)>0$, then $q>0$, and \eref{eq:FBproof1} implies that
	\begin{align}
	&\gamma(\rho,\mu)=q\gamma(\rho',\mu)\leq q\gamma_2^\rmF(F_B(\rho'),\mu)=q\gamma_2^\rmF(F_B(\rho)/q,\mu)\nonumber\\
	&\leq
	\begin{cases}
	2F_B(\rho)\gamma_2^\rmF(1/2,\mu)  & 0\leq F_B(\rho)\leq  1/2,\\[0.5ex]
	\gamma_2^\rmF(F_B(\rho),\mu) & 1/2\leq F_B(\rho)\leq 1.
	\end{cases}
	\end{align}
	so \eref{eq:gammarhoFUB} still holds. Here the first inequality follows from the definition of
	the guessing probability $\gamma_2^\rmF(F,\mu)$, and the last inequality follows from
	\lref{lem:qgammaMono} together with the constraint $F_B(\rho)\leq q\leq \min\{1,2F_B(\rho)\}$, given that $F_B(\rho')\geq 1/2$.
\end{proof}

\subsection*{\label{asec:thm:gammaFproof}Proof of \thref{thm:gammaF}}
\begin{proof}[Proof of \thref{thm:gammaF}]
	First let us consider \eref{eq:gamma2F} in the theorem. 	
	The equality $\gamma_2^\rmF(F,\mu)=\gamma_2(2F-1,\mu)$ follows from the fact that $F_B(\rho)=[1+C(\rho)]/2$ for any two-qubit pure state that satisfies $P_A\rho=\rho$ (cf. \lref{lem:gamma2F}). The upper bound  in \eref{eq:gamma2F} follows from \eref{eq:gammaCmuLUB} in the main text. Note that this is the   best linear upper bound for $\gamma_2^\rmF(F,\mu)$.

	To prove \eref{eq:gammaF} in \thref{thm:gammaF}, note that
	\begin{align}
	\gamma^\rmF(F,\mu)&=\max_\rho \{\gamma(\rho,\mu)|F_B(\rho)\leq F, \rho^2=\rho\}\nonumber\\
	&\leq \begin{cases}
	2F\gamma_2^\rmF(1/2,\mu)  & 0\leq F\leq  1/2,\\
	\gamma_2^\rmF(F,\mu) & 1/2\leq F\leq 1,
	\end{cases}\label{eq:gammaFproof1}
	\end{align}
	where the inequality follows from \lref{lem:gammarhoFUB} and  the monotonicity of $\gamma_2^\rmF(F,\mu)$ in $F$ as shown in
	\lref{lem:gammaFconvex}.

	When $1/2\leq F\leq 1$, we have $\gamma^\rmF(F,\mu)\geq \gamma_2^\rmF(F,\mu)$ by definition, so \eref{eq:gammaFproof1} implies that
	\begin{align}
	&\gamma^\rmF(F,\mu)=
	\gamma_2^\rmF(F,\mu)=\gamma_2(2F-1,\mu)\nonumber\\ &=\frac{1}{2}+\frac{1}{2}\max_{\bm{v}}\int\rmd\mu(\bm{r}) \sqrt{(2F-1)^2+4F(1-F)(\bm{r}\cdot\bm{v})^2},
	\end{align}
	which confirms \eref{eq:gammaF}. Here the second
	equality follows from \eref{eq:gamma2F} in \thref{thm:gammaF} and the third equality follows from \thref{thm:MGP}.

	When $0\leq F\leq 1/2$,
	let $\rho'=|\Psi' \rangle\langle \Psi'|$ be a pure product state
	that satisfies $P_A\rho'=\rho'$, $F_B(\rho')=1/2$,  and $\gamma(\rho',\mu)=\gamma_2^\rmF(1/2,\mu)$.
	Let $\rho=|\Psi\rangle\langle \Psi|$ with
	\begin{equation}
	|\Psi\rangle=\sqrt{2F}|\Psi'\rangle +\sqrt{1-2F}|22\rangle;
	\end{equation}
	then we have
	\begin{align}
	F_B(\rho)&=2FF_B(\rho')=F, \\
	\gamma(\rho,\mu)&=2F\gamma(\rho',\mu)=2F\gamma_2^\rmF(1/2,\mu),
	\end{align}
	which imply that $\gamma^\rmF(F,\mu)\geq 2F\gamma_2^\rmF(1/2,\mu)$. In conjunction with \eref{eq:gammaFproof1}, we conclude that
	\begin{align}
	&\gamma^\rmF(F,\mu)= 2F\gamma_2^\rmF(1/2,\mu)=2F\gamma_2(0,\mu)=2\gamma_2^\ast(\mu)F\nonumber\\
	&=2\gamma^\ast(\mu) F. 
	\end{align}
	Here the second equality follows from
	\eref{eq:gamma2F} in \thref{thm:gammaF}; the third equality follows from the definition in \eref{eq:gammaC0} in the main text, and the last equality follows from \lref{lem:gammaCHighDim}. 
	By virtue of  \thref{thm:MGP} and \esref{eq:gC0} and \eqref{eq:gammaC0} in the main text, we can also derive a more explicit expression for $\gamma^\rmF(F,\mu)$,
	\begin{align}
	&\gamma^\rmF(F,\mu)=F+Fg^\ast(\mu)
	=F+F\max_{\bm{v}}\int\rmd\mu(\bm{r}) |\bm{r}\cdot\bm{v}|.
	\end{align}
	
	Finally, we can prove \eref{eq:gammaFLUB} in \thref{thm:gammaF}.  When $1/2\leq F\leq 1$,  \eref{eq:gammaFLUB} follows from \esref{eq:gamma2F} and \eqref{eq:gammaF} in \thref{thm:gammaF}, and it offers the best linear upper bound for the guessing probability $\gamma^\rmF(F,\mu)$. 
	When $0\leq F\leq 1/2$, \eref{eq:gammaFLUB} follows from \eref{eq:gammaF} and the simple inequality $\gamma^\ast(\mu)\geq 1/2$, which is clear from \eref{eq:gammaC0} in the main text. This observation completes the proof of \thref{thm:gammaF}.
\end{proof}

\section*{\label{asec:SDIeff}SUPPLEMENTARY Note 6: Sample efficiency of SDI QSV}
In this section we clarify the sample efficiency of SDI QSV over an untrusted quantum network, in which some parties, but not all,   are honest.
For simplicity, let us first consider the
verification of the two-qubit Bell state $|\Phi\rangle=(|00\>+|11\>)/\sqrt{2}$. To quantify the closeness between the actual state $\rho$  and the target state $|\Phi\rangle$, here we 
adopt the  reduced fidelity
\begin{align}
F_B(\rho)=\max_{U_B} \hspace*{0.05cm} \<\Phi|(\mathbb{I}_A\otimes U_B)\rho (\mathbb{I}_A\otimes U_B)^\dag |\Phi\>,
\end{align}
as defined in \eref{def:reducedfidelity} in the main text,  
where the maximization is taken over all local unitary transformations on $\caH_B$.

A key to determining the sample efficiency is the relation between the guessing probability and the reduced fidelity as presented in  Theorem~3 in the main text. The linear upper bound for $\gamma^\rmF(F,\mu)$ in  \eref{eq:gammaFLUB} is particularly important and is reproduced below,
\begin{equation}\label{eq:prob-fidelity}
\gamma^\rmF(F,\mu) \leq 1-2(1-\gamma^\ast)(1-F). 
\end{equation}
Given a verification protocol specified by the distribution $\mu$ and any state $\rho$ with $F_B(\rho)\leq 1-\epsilon$, then  the probability that $\rho$ can pass one  test  satisfies
\begin{align} \label{eq:steeringprob}
p \leq 1-2(1-\gamma^{\ast})\epsilon.
\end{align}

Now suppose  the states $\rho_1,\rho_2,\dots,\rho_N$ prepared in $N$ runs are  independent  of each other. Let $\epsilon_j=1-F_B(\rho_j)$;
then the probability that these states can pass all $N$ tests satisfies the following inequalities,
\begin{equation}\label{eq:ProbPassNtestSDI}
\prod_{j=1}^N p_j \leq \prod_{j=1}^N[ 1-2(1-\gamma^{\ast})\epsilon_j]\leq[1- 2(1-\gamma^{\ast})\bar{\epsilon}\lsp]^N,
\end{equation}
where $\bar{\epsilon}=\sum_j\epsilon_j/N$ is the average infidelity.
In order to ensure the condition $\bar{\epsilon}<\epsilon$ with significance level $\delta$ (confidence level $1-\delta$), that is, to ensure the condition
$\prod_{j}p_j\leq\delta$ when $\bar{\epsilon}\geq\epsilon$,
it suffices to perform 
\begin{equation}\label{aeq:NumberTestSDI}
N = \biggl\lceil\frac{ \ln \delta}{\ln[1-2(1-\gamma^{\ast})\epsilon]}\biggr\rceil\approx \frac{ \ln \delta^{-1}}{2(1-\gamma^{\ast})\epsilon}
\end{equation}
tests, which scale as $N=O(1/(1-\gamma^{\ast})\epsilon)$.
This equation demonstrates the importance of the threshold $\gamma^{\ast}$ in determining the verification efficiency. A small threshold $\gamma^{\ast}$ means a high efficiency.

\begin{figure}
	\begin{center}
		\includegraphics[width=8cm]{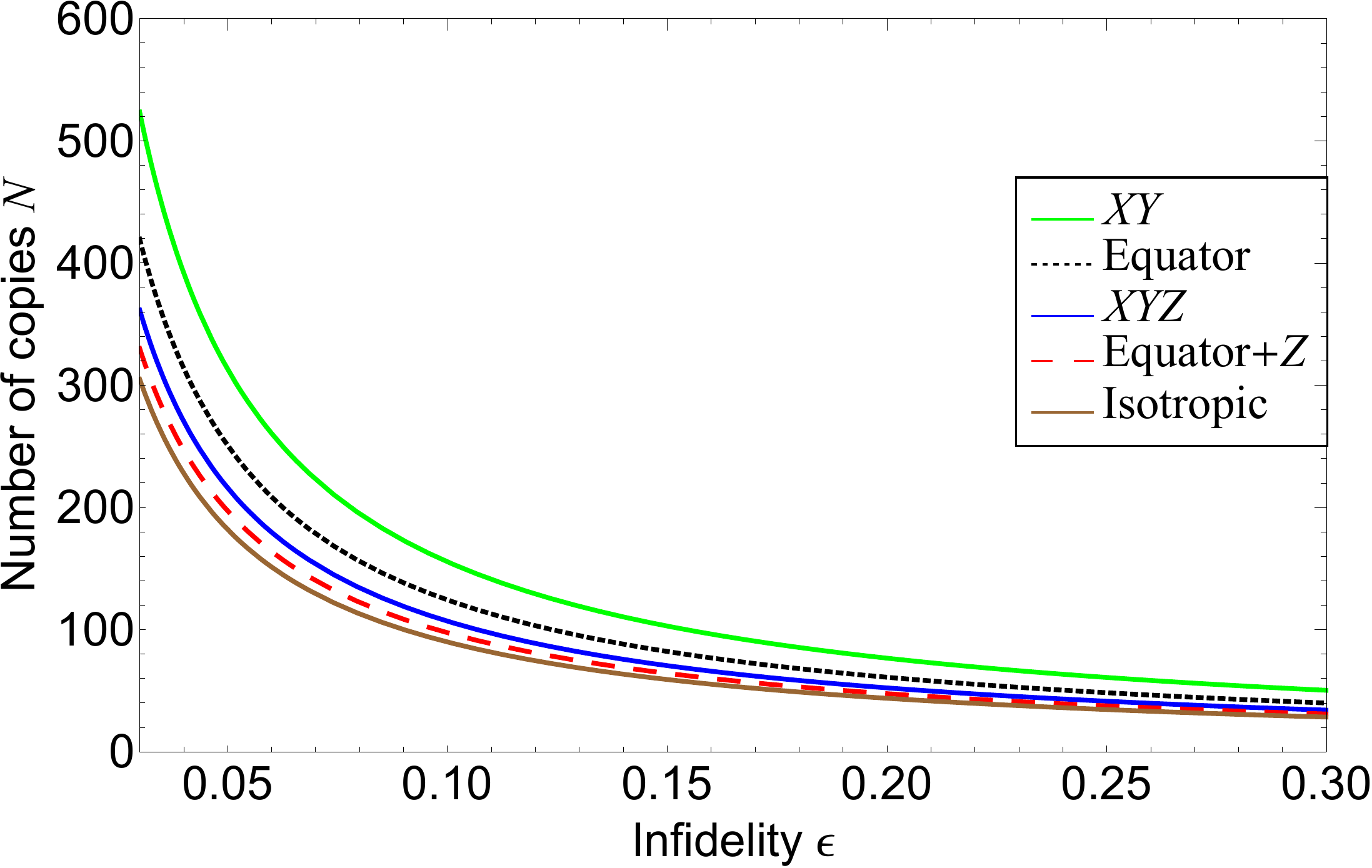}
		\caption{\label{fig:SampleCompare}
			Number of copies needed to verify  the Bell state in the SDI scenario 
			for various verification protocols. Here  the $XY$ protocol and isotropic protocol are introduced in the main text, while  other protocols are proposed in Supplementary Note 9. 
			The  significance level is chosen to be $\delta=0.01$.  
		}
	\end{center}
\end{figure}

The minimum threshold  $\gamma^{\ast}=3/4$ is achieved for the  isotropic protocol (cf. Supplementary Note 9 or \tref{tab:ProtocolGHZDH} in the main text), in which case we have
\begin{equation}
N \approx \frac{2 \ln \delta^{-1}}{\epsilon} .
\end{equation}
Surprisingly, this number is comparable to the number $(3\ln \delta^{-1})/(2\epsilon)$ required  in standard QSV \cite{pallister2018SM,ZhuH2019OSM}. Therefore, SDI verification of the Bell state is almost as efficient as standard QSV. The sample efficiencies of the isotropic protocol and  several other protocols in the SDI scenario are illustrated in Supplementary Figure 2.

The above results can be generalized to the verification of the GHZ state following a similar analysis as presented in the main text. Now the threshold is minimized 
at an  equator$+Z$ protocol with $\gamma^{\ast}\approx 0.769$ (cf. \tref{tab:ProtocolGHZDH}). Accordingly, the   number of tests required to achieve infidelity $\epsilon$ and significance level $\delta$ reads
\begin{equation}
N \approx \frac{2.16 \ln \delta^{-1}}{\epsilon}.
\end{equation}
As in the verification of the Bell state, this number is comparable to the number $(3\ln \delta^{-1})/(2\epsilon)$ required  in standard QSV \cite{li2020GHZSM}.

To understand the efficiency of our verification protocols proposed for the the SDI scenario, note that each  protocol is tied to a steering inequality whose quantum bound and algebraic bound are both equal to 1. This bound corresponds to the maximal probability of passing each  test on  average and can be attained by a quantum strategy. In addition, the guessing probability threshold is the maximal guessing probability that can be achieved by a classical strategy. As we shall see in Supplementary Note 7, such a high efficiency  cannot be achieved in the DI scenario in general unless there exists a suitable Bell inequality for which the quantum bound coincides with the algebraic bound.

\subsection*{Comparison with \rscite{Supic2016SM,Gheorghiu2017SM}}

Here we compare our results with previous results in \rscite{Supic2016SM,Gheorghiu2017SM}, which studied the self-testing of the Bell state in the one-sided DI scenario,  which corresponds to the SDI scenario considered in this work. Both works focused on the steering inequality \cite{Supic2016SM,Gheorghiu2017SM}
\begin{equation}
S=\langle\psi\vert Z\otimes Z'\vert\psi\rangle+\langle\psi\vert X\otimes X'\vert\psi\rangle\leq\sqrt{2},
\end{equation}
where  $Z$, $X$ are trusted Pauli measurements on Alice's side and $Z'$,  $X'$ are untrusted measurements on Bob's side. Both quantum bound and algebraic bound of this inequality are $S_Q=2$, which can be achieved by the quantum strategy with ideal measurements on the Bell state $|\Phi\rangle=(|00\>+|11\>)/\sqrt{2}$.  Verification protocol based on this steering inequality is essentially equivalent to the $XY$ protocol considered in this work, which can be regarded as a steering inequality with quantum bound 1 and classical bound $\gamma^\ast=\frac{1}{2}+\frac{1}{2\sqrt{2}}$.

Both works \rscite{Supic2016SM,Gheorghiu2017SM} provided robust self-testing statements in the one-sided DI  scenario. For the steering correlation $S\geq 2-\varepsilon$, \rcite{Supic2016SM} obtained an analytic bound of $13\sqrt{\varepsilon}$ and a numeric bound of $\sqrt{\varepsilon}$ on the trace distance with the target state, while \rcite{Gheorghiu2017SM} obtained an analytic bound of $(3\sqrt{2}+8)\sqrt{\varepsilon} + \varepsilon/2$.
For comparison, here we analyze the performance of the $XY$ protocol with respect to the trace distance. Since the guessing probability in the $XY$ protocol can be rewritten as $\gamma=\frac{1}{2}+\frac{S}{2S_{Q}}$, combining \eref{eq:prob-fidelity},  the reduced infidelity can be bounded from above as    $1-F \leq \varepsilon/[2(2-\sqrt{2})]$, so the trace distance can be upper bounded by $\sqrt{\varepsilon}/\sqrt{2(2-\sqrt{2})}\approx0.924\sqrt{\varepsilon}$, which is tighter than previous results.

\rcite{Supic2016SM} also studied the verification of the  tripartite GHZ state in both 1-trusted setting and 2-trusted setting, which correspond to different types of multipartite EPR-steering. Numerical results are obtained for both settings based on steering inequalities derived from the Mermin inequality.  The scenario considered in this work is more similar to the 2-trusted setting since we assume that each party can only perform local projective measurements.
Our result shows that the robustness in the SDI scenario is higher than that in the DI scenario, which is consistent with the result in \rcite{Supic2016SM}.

\subsection*{Comparison with \rscite{Pappa2012,McCut2016}}
Next we compare our results with previous results in \rscite{Pappa2012SM,McCut2016SM}, which studied the verification of GHZ states in the SDI scenario as considered in this work. 
The $XY$ protocol considered in this work is equivalent to the protocol proposed in \rcite{Pappa2012SM}, and the equator protocol  is equivalent to the $\theta$-protocol proposed in \rcite{McCut2016SM}. The authors of \rscite{Pappa2012SM,McCut2016SM} derived the guessing probability thresholds  $\gamma^{\ast}=\frac{1}{2}+\frac{1}{2\sqrt{2}}\approx 0.854$ for the $XY$ protocol and $\gamma^{\ast}=\frac{1}{2}+\frac{1}{\pi}\approx 0.818$
for the $\theta$-protocol. In addition, they derived the relation
\begin{equation}\label{eq:guessProbPre}
p \leq 1-\frac{\epsilon}{4}
\end{equation}
between the guessing probability and the reduced infidelity for both protocols. However, this  relation is suboptimal compared with our result in \eref{eq:steeringprob}, which offers the best linear upper bound for the guessing probability. In addition, 
\rscite{Pappa2012SM,McCut2016SM} did not consider the sample efficiency. 
If  \eref{eq:guessProbPre} were combined with our analysis above, then the number of tests required to achieve infidelity $\epsilon$ and significance level $\delta$ would be $N \approx 4 (\ln \delta^{-1})/\epsilon$. For comparison, by virtue of \esref{eq:steeringprob} and \eqref{aeq:NumberTestSDI},  we can derive $N \approx 3.41 (\ln \delta^{-1})/\epsilon$ for the $XY$ protocol and $N \approx 2.75 (\ln \delta^{-1})/\epsilon$ for the equator protocol.

In addition to the efficiency advantage mentioned above, the main merit of our work is to propose a simple approach for determining all potential tests of the GHZ state that are based on local projective measurements. In particular, we prove that only two types of tests  can be constructed from local projective measurements (cf. Supplementary Note 10). 
By virtue of  this result, we determine the optimal protocol for verifying the GHZ state and for certifying GME. Such optimality results are difficult to establish and are thus quite rare in the DI and SDI scenarios.  Meanwhile,  our $XYZ$ protocol, which requires only three measurement settings,  is more efficient than all protocols known in the literature even based on infinite measurement settings. The protocols in  \rscite{Pappa2012SM,McCut2016SM} are suboptimal because they   employ only the second type of tests and their theoretical analyses are suboptimal.

\section*{\label{app:SDIvsQSVvsDI}SUPPLEMENTARY Note 7: Comparison with standard QSV and device-independent QSV}
To put our work in perspective, in this section we compare SDI QSV with standard QSV \cite{pallister2018SM,ZhuH2019AdvSSM,ZhuH2019AdvLSM} and DI QSV \cite{Dimic2021SM} based on self-testing \cite{Mayers2004SM,Supic2020SM}. To start with, it is instructive to clarify the assumptions underlying these frameworks. 
In standard QSV, all parties that implement the verification protocol  are trustworthy, although the source of the quantum system is not necessarily trustworthy. 
In DI QSV, the source is not trustworthy and is usually  treated as a black-box; meanwhile, all the parties that implement the verification protocol (or the underlying measurement devices) are not trustworthy. 
In this regard, SDI QSV lies between standard QSV and DI QSV in that some parties are honest, but some others may be dishonest. In a word, the assumptions are strongest in standard QSV and weakest in DI QSV. Accordingly, it is easier to realize standard QSV in experiments than SDI QSV, which in turn is easier than DI QSV.

The verification problems in the three  scenarios can all be summarized as follows. A quantum device  is supposed to produce the target state $|\Psi\>\in\caH$, but actually produces the states $\rho_1,\rho_2,\dots,\rho_N$ in $N$ runs. For simplicity we assume that the states prepared in different runs are independent of each other. Our task is to verify  whether these states are sufficiently close to the target state on average. It should be pointed out that the quantification of closeness depends on the specific scenario under consideration.

\subsection*{Standard quantum state verification}\label{asec:qsv}

In standard QSV \cite{pallister2018SM,ZhuH2019AdvSSM,ZhuH2019AdvLSM}, all parties implementing the verification protocol are trustworthy; in other words, the measurement devices are trustworthy. A verification protocol for a given target state $|\Psi\rangle$  is composed of a number of binary tests represented by two-outcome  measurements  $\{E_l,\mathbb{I}-E_l\}$. 
Here the test operator  $E_l$ corresponds to passing the test and satisfies the condition  $E_l|\Psi\>=|\Psi\>$, so that the target state $|\Psi\>$ can  always pass the test.
Suppose the test $\{E_l,\mathbb{I}-E_l\}$ is performed with probability $p_l$; then the verification operator
is given by $\Omega=\sum_{l} p_l E_l$.
If  $\<\Psi|\rho_j|\Psi\>\leq1-\epsilon$, then the average probability that $\rho_j$ can pass each test satisfies \cite{pallister2018SM,ZhuH2019AdvSSM}
\begin{equation}\label{eq:ProbPass1test}
\tr(\Omega \rho_j)\leq 1- [1-\beta(\Omega)]\epsilon=1- \nu(\Omega)\epsilon,
\end{equation}
where $\beta(\Omega)$ denotes the second largest eigenvalue of $\Omega$, and $\nu(\Omega):=1-\beta(\Omega)$ is the spectral gap from the maximum eigenvalue. 

To guarantee that the average infidelity satisfies $\bar{\epsilon}<\epsilon$ with significance level $\delta$, it suffices to perform \cite{pallister2018SM,ZhuH2019AdvSSM}
\begin{equation}\label{eq:NumberTest}
N = \biggl\lceil\frac{ \ln \delta}{\ln[1-\nu(\Omega)\epsilon\lsp]}\biggr\rceil\approx \frac{ \ln \delta^{-1}}{\nu(\Omega)\epsilon}
\end{equation}
tests, which scale as $N=O(1/\nu (\Omega)\epsilon)$. 
To minimize the number of tests, we need to maximize the value of the spectral gap $\nu(\Omega)$ over all protocols based on LOCC, which is usually extremely difficult if not impossible. So far sample-optimal protocols have been found only for limited classes of states. Fortunately, this problem has been resolved for the Bell state \cite{HayaMT06SM,pallister2018SM,ZhuH2019OSM} and GHZ states recently \cite{li2020GHZSM}. In both cases, the maximum spectral gaps read $\nu(\Omega)=2/3$, and the numbers of required tests  are given by
\begin{equation}N \approx \frac{3 \ln \delta^{-1}}{2\epsilon}.
\end{equation}

\subsection*{\label{asec:DIQSV}Device-independent quantum state verification}

Device-independent QSV \cite{Dimic2021SM} can be viewed as a verification procedure in the  black-box scenario in  which measurement devices cannot be trusted. A key for constructing DI verification protocols is
self-testing \cite{Mayers2004SM,Supic2020SM}, by which quantum states can be certified up to local isometries using only the nonlocal statistics. For example, the maximum quantum bound $2\sqrt{2}$ of the Clauser-Horne-Shimony-Holt (CHSH) inequality \cite{CHSH1969SM} certifies the singlet \cite{McKague2012SM}.   The maximum quantum bound $4$ of the Mermin inequality \cite{Mermin1990SM} certifies the tripartite GHZ state \cite{Kaniewski2016SM}. 
Robustness is the main focus of most studies on 
self-testing when the correlations achieved deviate from the ideal ones.

On the other hand, to construct a practical DI verification protocol, a self-testing protocol is not enough by itself, unless practical issues like sample efficiency and confidence level can be clarified. Unfortunately, the sample complexity of DI QSV has received little attention until the recent work of Dimi\'c et al. \cite{Dimic2021SM} although there are numerous works on self-testing (see \rcite{Supic2020SM} for a review). Following \rscite{Kaniewski2016SM,Dimic2021SM},  here we discuss the sample efficiency of DI verification of the tripartite GHZ state based on self-testing results tied to the Mermin inequality. 

In self-testing, the \emph{extractability} \cite{Kaniewski2016SM} is used to quantify the closeness between the actual state $\rho$ and the target state $\ket{\Psi}$; it is defined as 
\begin{equation}
\Theta(\rho, \ket{\Psi}) := \max_{ \Lambda } F \big( \Lambda (\rho), \ket{\Psi} \big),
\end{equation}
where the maximization is taken over all local isometries. 
Suppose the Bell inequality $\mathcal{B}$ is used to self-test the target state $\ket{\Psi}$. Denote by $\beta_{C}$ and $\beta_{Q}$ the maximal classical bound and quantum bound of $\mathcal{B}$, respectively. To determine the sample complexity,  we need to clarify the relation between the
extractability and the Bell violation. The extractability-violation trade-off is characterized by the function $\mathscr{Q}_{\Psi, \mathcal{B}} : [\beta_{C}, \beta_{Q}] \to [0, 1]$ defined as \cite{Kaniewski2016SM}:
\begin{equation}
\mathscr{Q}_{\Psi, \mathcal{B}}(\beta) := \inf_{\rho_{AB} \in \mathcal{S}_{\mathcal{B}}(\beta)} \Theta(\rho_{AB}, \Psi),
\end{equation}
where $\mathcal{S}_{\mathcal{B}}(\beta)$ denotes the set of quantum states  that can achieve the violation $\beta$ for the Bell inequality $\mathcal{B}$. 
This function sets a lower bound on the extractability given the observed violation $\beta$.

To be concrete, let us consider self-testing of the tripartite GHZ state using the Mermin inequality. Recall that the Mermin operator reads
\begin{equation}\label{eq:Mermin}
\mathcal{B}_{\textnormal{Mermin}} = \sum_{j,k \in \{0, 1\}} (-1)^{jk} A_{j} \otimes B_{k} \otimes C_{j \oplus k},
\end{equation}
where $A_j$ are binary observables (with eigenvalues $\pm 1$) for Alice  (the observables for Bob and Charlie have the same structure). The quantum bound of the Mermin inequality is $\beta_Q=4$, which coincides with the algebraic bound. The extractability-violation function was determined in  \rcite{Kaniewski2016SM}, with the result
\begin{equation}
\label{eq:chsh}
\mathscr{Q}_{G^3, \mathcal{B}_{\textnormal{Mermin}}}(\beta) \geq \frac{1}{2} + \frac{1}{2} \cdot \frac{\beta - \beta^{*}}{\beta_{Q} - \beta^{*}},
\end{equation}
where  $\beta_{Q}=4$ is the quantum bound and $\beta^{*} = 2\sqrt{2}$ is the threshold violation \cite{Kaniewski2016SM}. Notably, for any state $\rho$ with $\Theta(\rho, \Psi)\leq 1-\epsilon$, the Bell violation $\beta$ it can achieve satisfies
\begin{equation}
\label{eq:Merminfidelity}
\frac{1}{2} + \frac{1}{2} \cdot \frac{\beta - \beta^{*}}{\beta_{Q} - \beta^{*}} \leq 1-\epsilon .
\end{equation}

Since the quantum bound coincides with the algebraic bound for the Mermin inequality, the self-testing procedure determined by \eref{eq:Mermin} can be turned into a DI verification protocol with four tests chosen with probability $1/4$ each. Note that the target GHZ state
can always pass each test. It remains  to establish the relation between the guessing probability and the  extractability. For a generic state $\rho$, the average probability of passing each test reads $p=\frac{1}{2}+\frac{\beta}{2\beta_{Q}}$, where $\beta$ is the Bell violation. If $\Theta(\rho, \Psi)\leq 1-\epsilon$, then \eref{eq:Merminfidelity} implies that
\begin{equation}
\label{eq:Merminprob}
p \leq 1-	\left(1-\frac{\beta^{\ast}}{\beta_Q}\right)\epsilon=1-	\frac{2-\sqrt{2}}{2}\epsilon.
\end{equation}

Now we are ready to estimate the number of copies needed to achieve a given extractability in DI verification of the tripartite GHZ state. Suppose  the states $\rho_1,\rho_2,\dots,\rho_N$ prepared in $N$ runs are  independent  of each other. 
Then the probability that these states can pass all $N$ tests is upper bounded by
\begin{equation}\label{eq:ProbPassNtestDI}
\biggl(1-\frac{2-\sqrt{2}}{2}\bar{\epsilon}\biggr)^N,
\end{equation}
where $1-\bar{\epsilon}$ denotes the average extractability.
In order to insure the condition $\bar{\epsilon}<\epsilon$ with significance level $\delta$,
it suffices to perform
\begin{equation}\label{eq:NumberTestMermin}
N = \Biggl\lceil\frac{ \ln \delta}{\ln\bigl[1-\frac{2-\sqrt{2}}{2}\epsilon\lsp\bigr]}\Biggr\rceil\approx \frac{2 \ln \delta^{-1}}{(2-\sqrt{2}\lsp)\epsilon}\approx \frac{ 3.41\ln \delta^{-1}}{\epsilon}
\end{equation}
tests. Surprisingly, the scaling behaviors of $N$ in $\epsilon$ and $\delta$ for DI QSV are the same as the counterparts for standard QSV and SDI QSV considered in this work; moreover, the constant coefficients  for the three scenarios are quite close to each other.

It should be pointed out  that the above analysis  is applicable only when the quantum bound of the underlying Bell inequality coincides with the algebraic bound. If this condition does not hold, then the situation gets  more  complicated \cite{Dimic2021SM}, and the optimal scaling behaviors shown in \eref{eq:NumberTestMermin} cannot be guaranteed. Now the number of tests required reads
\begin{equation}\label{eq:Non-algebraic}
N=O\left(\frac{\ln{\delta^{-1}}}{c^2 {\epsilon}^2}\right),
\end{equation}
where $c$ is a constant characterizing the linear dependence between
the extractability and the Bell violation. 
For example, \eref{eq:Non-algebraic} holds in  DI verification of the Bell state based on the CHSH inequality.
Here the sample efficiency is suboptimal  compared with the counterparts in  standard QSV and SDI QSV (cf. Supplementary Note~6).

\section*{SUPPLEMENTARY Note 8: Verification protocols based on discrete distributions}
Recall that each verification protocol of the Bell state is specified by a probability distribution on the Bloch sphere.  In practice, it is usually more convenient to choose a discrete distribution as specified by a weighted set $\{\bm{r}_j, p_j\}_j$, where $S=\{\bm{r}_j\}_j$ is  a set of unit vectors on the Bloch sphere, and $\{p_j\}_j$ is a probability distribution. This weighted set
means the projective measurement $\bm{r}_j$ is performed with probability $p_j$.  When all the probabilities $p_j$ are equal, the weighted set $\{\bm{r}_j, p_j\}_j$ is abbreviated as $S=\{\bm{r}_j\}_j$ to simplify the notation. All results presented in the main text and in this Supplementary Material  hold for discrete distributions as well as for continuous distributions.  For example, \thref{thm:MGP} implies that
\begin{equation}\label{eq:gCrp}
g(C, \{\bm{r}_j, p_j\}_j)=\max_{\bm{v}} \sum_j  p_j\sqrt{C^2+(1-C^2)(\bm{r}_j\cdot\bm{v})^2}.
\end{equation}
When $C=0$, \eref{eq:gCrp} reduces to 
\begin{equation}\label{eq:gC0rp}
g^\ast( \{\bm{r}_j, p_j\}_j)=\max_{\bm{v}} \sum_j  p_j|\bm{r}_j\cdot\bm{v}|.
\end{equation}
Incidentally, the verification matrix defined in \eref{eq:Ximu} now reduces to
\begin{equation}
\Xi(\mu)=\Xi(\{\bm{r}_j, p_j\}_j)=\sum_j p_j \bm{r}_j \bm{r}_j^\rmT.
\end{equation}
When all $p_j$ are equal, \eref{eq:gC0rp} further reduces to 
\begin{equation}\label{eq:gC0S}
g^\ast(S)=\frac{1}{|S|}\max_{\bm{v}} \sum_{\bm{r}\in S} |\bm{r}\cdot\bm{v}|,
\end{equation}
where $|S|$ denotes the cardinality of $S$. 

\begin{lemma}\label{lem:Transitive}
	Suppose the symmetry group of $\{\bm{r}_j\}_j$ acts transitively. Then $g(C, \{\bm{r}_j, p_j\}_j)$ for a given set  $\{\bm{r}_j\}_j$ is minimized when all $p_j$ are equal.
\end{lemma}
\Lref{lem:Transitive} follows from \lref{lem:Convex}.  It applies in particular when $\{\bm{r}_j \}_j$ forms a platonic solid or a regular polygon. 

Let $S$ be  a set    of  unit vectors on the Bloch sphere and let $\bm{v}$ be a unit vector on the Bloch sphere.  The set $S$ is center symmetric if $\bm{r}\in S$  means $-\bm{r}\in S$ and vice versa.
Define
\begin{align}
\bm{\eta}(S)&:=\sum_{\bm{r}\in S} \bm{r}, \label{eq:etaS}\\
S_{\bm{v}}&:=\{\bm{r}\in S \,|\,\bm{r}\cdot \bm{v}\geq 0\}. \label{eq:Sv}
\end{align}

\begin{lemma}\label{lem:CenterSym}
	Let  $S=\{\bm{r}_j\}_{j=1}^M$ be a set of $M$ unit vectors on the Bloch sphere that is center symmetric. Then
	\begin{align}
	g^\ast(S)
	&=\frac{1}{M}\max_{\bm{v}}\sum_{j=1}^M|\bm{r}_j\cdot\bm{v}| \label{eq:g*(mu)CSv1} \\
	&=\frac{2}{M}\max_{\bm{v}}|\bm{\eta}(S_{\bm{v}} )| \label{eq:g*(mu)CSv2}\\
	&= \frac{2}{M}\,\max_{S'\subseteq S}\,\left|\bm{\eta}(S') \right|. \label{eq:g*(mu)CSS}
	\end{align}
	If the maximum in \eref{eq:g*(mu)CSS} is attained at the set $S'$, then $S'$  contains exactly $M/2$ vectors; in addition,
	\begin{equation}\label{eq:etaSr}
	\bm{\eta}(S')\cdot \bm{r}\geq \frac{1}{2} \;\; \forall \bm{r}\in S',\quad  \bm{\eta}(S')\cdot \bm{r}\leq -\frac{1}{2} \;\;\forall \bm{r}\in \overline{S'},
	\end{equation}
	where $\overline{S'}$ is the complement of $S'$ in $S$. If the maximum in \eref{eq:g*(mu)CSv1} is attained at the unit vector $\bm{v}$, so that $\bm{v}$ is an intelligent direction at $C=0$, then 
	\begin{equation}
	\bm{v} \,\mathbin{\!/\mkern-5mu/\!}\, \bm{\eta}(S_{\bm{v}}), \quad \bm{v}\cdot \bm{r}_j\neq 0\;\;\forall j. 
	\end{equation}
\end{lemma}
\Lref{lem:CenterSym} also implies that an optimal set $S'$ that maximizes $|\bm{\eta}(S')|$ contains exactly one vector in the set $\{\bm{r},-\bm{r}\}$ for each $\bm{r}\in S$.

\begin{proof}
	\Esref{eq:g*(mu)CSv1}-\eqref{eq:g*(mu)CSS} follow from \eref{eq:gC0S} and \lref{lem:g*(mu)=2max}; note that the condition in \eref{eq:CenterSymProt} is guaranteed by the assumption that $\{\bm{r}_j\}_{j=1}^M$ is center symmetric. 	
	
	Suppose the maximum in \eref{eq:g*(mu)CSS} is attained at the set $S'$. Then 
	\begin{align}
	|\bm{\eta}(S')-\bm{r}|^2&\leq |\bm{\eta}(S')|^2\quad \forall \bm{r}\in  S', \\
	|\bm{\eta}(S')+\bm{r}|^2&\leq |\bm{\eta}(S')|^2\quad \forall \bm{r}\in \overline{S'},
	\end{align}
	which imply  \eref{eq:etaSr}. 
	As a corollary, $S'$ contains exactly one vector in the set $\{\bm{r},-\bm{r}\}$ for each $\bm{r}\in S$. In particular, $S'$ has cardinality $M/2$.
	
	Next, suppose the maximum in \eref{eq:g*(mu)CSv1} is attained at the unit vector $\bm{v}$.  Then $\bm{v} \,\mathbin{\!/\mkern-5mu/\!}\, \bm{\eta}(S_{\bm{v}})$ by \lref{lem:g*(mu)=2max}, and the maximum in \eref{eq:g*(mu)CSS} is attained at the set $S_{\bm{v}}$. In addition, $\bm{\eta}(S_{\bm{v}})\cdot \bm{r}_j\neq 0$ for all $j$ according to \eref{eq:etaSr}, so $\bm{v}\cdot \bm{r}_j\neq 0$ for all $j$. 
\end{proof}

\section*{\label{app:discrete}SUPPLEMENTARY Note 9: Concrete  verification protocols}

Here we study various concrete protocols for verifying the Bell state in the presence of a dishonest party. In addition to the verification protocols listed in \tref{tab:ProtocolGHZDH}, we also consider protocols based on platonic solids and arbitrary regular polygons (with extension by including the $Z$ measurement). Given a verification protocol specified by a distribution $\mu$ on the Bloch sphere, our main goal is to determine $g(C,\mu)$, $\gamma(C,\mu)=\gamma_2(C,\mu)$,  $\gamma^\ast(\mu)=\gamma_2^\ast(\mu)=\gamma_2(0,\mu)$, and $ \hat{\gamma}(C,\mu)=\hat{\gamma}_2(C,\mu)$ (cf. \thsref{thm:MGP},\ref{thm:Mixed}, and \lref{lem:gammaCHighDim} in the main text). Furthermore, we determine the optimal protocol and the optimal two-setting protocol.
The main results are summarized in \tref{tab:ProtocolGHZDH} in the main text. 
Many protocols presented here can  easily be adapted for the verification of GHZ states.

\subsection*{\label{app:two}Two-setting protocols}
In the simplest verification protocol, Alice can perform two projective measurements $\bm{r}_1$ and $\bm{r}_2$ with probabilities $p_1$ and $p_2$, respectively. By replacing $\bm{r}_2$ with $-\bm{r}_2$ if necessary, we can assume that $\bm{r}_1\cdot \bm{r}_2\geq0$.
To minimize the guessing probability of Bob, Alice can choose $p_1=p_2=1/2$ according to \lref{lem:Transitive}; then the verification matrix reads
\begin{equation}
\Xi=\frac{1}{2}(\bm{r}_1 \bm{r}_1^\rmT+\bm{r}_2 \bm{r}_2^\rmT).
\end{equation}
Let $\alpha=\arccos(\bm{r}_1\cdot\bm{r}_2)$; then  $0\leq \alpha\leq \pi/2$ and
\begin{equation}
\|\Xi\|=\frac{1+\cos\alpha}{2}.
\end{equation}
In addition,  \lref{lem:gCUB} implies that
\begin{align}
g(C,\mu)&= \sqrt{\frac{1+\cos\alpha+C^2(1-\cos\alpha)}{2}}\nonumber\\
&=\sqrt{\frac{1+C^2+(1-C^2)\cos\alpha}{2}}; \label{eq:g2setting}
\end{align}
meanwhile, any unit vector parallel to $\bm{r}_1+\bm{r}_2$ represents an intelligent direction. 
Here  $g(C,\mu)$ is minimized when $\alpha=\pi/2$, so that $\bm{r}_1$ and $\bm{r}_2$ are orthogonal,  in which case we 
obtain the optimal two-setting protocol with
\begin{equation}\label{eq:ggammaXYapp}
g(C,\mu)=\sqrt{\frac{1+C^2}{2}},\;\; \gamma_2(C,\mu)=\frac{1}{2}+\frac{1}{2}\sqrt{\frac{1+C^2}{2}} .
\end{equation}
The special case $XY$ protocol was proposed  previously in  \rcite{Pappa2012SM}. However, \rcite{Pappa2012SM} neither derived the exact formula for the guessing probability  nor proved the optimality of the $XY$ protocol among all two-setting protocols.

When $C=0$, \eref{eq:ggammaXYapp} yields
\begin{equation}
g^\ast(\mu)=\frac{1}{\sqrt{2}},\quad \gamma_2^\ast(\mu)=\frac{1}{2}+\frac{1}{2\sqrt{2}}.
\end{equation}
By virtue of \thref{thm:Mixed} we can further deduce that
\begin{align}
\hat{\gamma}_2(C,\mu)
&=\frac{1+C}{2}+\frac{1-C}{2\sqrt{2}}  =\frac{1}{4}[2+\sqrt{2}+(2-\sqrt{2})C\lsp].
\end{align}

\begin{proposition}\label{pro:Opt2setting}
	Suppose $\mu$ is a two-setting protocol specified by the weighted set $\{\bm{r}_j, p_j\}_{j=1}^2$. Then 
	\begin{gather}
	g(C,\mu)\geq \sqrt{\frac{1+C^2}{2}},\;\; \gamma_2(C,\mu)\geq \frac{1}{2}+\frac{1}{2}\sqrt{\frac{1+C^2}{2}}, \label{eq:ggammaLB2setting}\\
	\hat{\gamma}_2(C,\mu)
	\geq \frac{1}{4}[2+\sqrt{2}+(2-\sqrt{2})C\lsp]. \label{eq:gammahatLB2setting}
	\end{gather} 
	When $0\leq C<1$, the three inequalities  are saturated iff $\bm{r}_1,\bm{r}_2$ are  orthogonal and $p_1=p_2=1/2$.
\end{proposition}
\begin{proof}
	According to the above discussion, the three inequalities in \esref{eq:ggammaLB2setting} and \eqref{eq:gammahatLB2setting}  hold for all two-setting protocols. In addition they are saturated if $\bm{r}_1,\bm{r}_2$ are  orthogonal and $p_1=p_2=1/2$. 
	When $0\leq C<1$, if one of the 
	inequalities is saturated, then the other two  are also saturated thanks to \eref{eq:gammaCmu} and \thref{thm:Mixed} in the main text, so we have $g(C,\mu)= \sqrt{(1+C^2)/2}$. According to \eref{eq:g2setting} and the above discussion, $\bm{r}_1,\bm{r}_2$  must be orthogonal. By \eref{eq:gCrp} we have
	\begin{align}
	g(C,\mu)&=\max_{x^2+y^2\leq 1}\Bigl(p_1\sqrt{C^2+(1-C^2)x^2}\nonumber\\
	&\quad +p_2\sqrt{C^2+(1-C^2)y^2}\lsp\Bigr)\nonumber\\
	&\geq \max_{x^2+y^2\leq 1}\sqrt{C^2+(1-C^2)(p_1 x+p_2 y)^2}\nonumber\\
	&=\sqrt{C^2+(1-C^2)(p_1^2+p_2^2)}\geq  \sqrt{\frac{1+C^2}{2}},
	\end{align}
	where the first inequality follows from the fact that $\sqrt{C^2+(1-C^2)x^2}$ is convex in $x$. The last inequality is saturated iff $p_1=p_2=1/2$. This observation completes the proof of \pref{pro:Opt2setting}.
\end{proof}

\subsection*{\label{asec:XYZ}Three-setting protocols}
Here  we consider verification protocols based on three measurement settings. We are particularly interested in the case in which the Bloch vectors $\bm{r}_1, \bm{r}_2, \bm{r}_3$ that specify the three projective measurements are mutually orthogonal, so that the corresponding projective measurements are mutually unbiased. For example, the measurement bases can be chosen  to be the eigenbases of the three Pauli operators $X$, $Y$, $Z$, which lead to the $XYZ$ protocol. To minimize the guessing probability, the three measurements should be performed with the equal probability of $1/3$ according to \lref{lem:Transitive}.
In this case, the verification matrix reads $\Xi(\mu)=\mathbb{I}/3$, and the bound in \lref{lem:gCUB} can be saturated, so we have
\begin{align}
g(C,XYZ)
&=\sqrt{\frac{1+2C^2}{3}},\\
\gamma_2(C,XYZ)
&=\frac{1}{2}+\frac{1}{2}\sqrt{\frac{1+2C^2}{3}}.
\end{align}
When $0\leq C<1$, there are eight intelligent directions, namely, $(\pm1, \pm 1, \pm 1)^\rmT/\sqrt{3}$. 
When $C=0$, the above two equations reduce to
\begin{align}
g^\ast(XYZ)=\frac{1}{\sqrt{3}},\quad
\gamma_2^\ast(XYZ)=\frac{1}{2}+\frac{1}{2\sqrt{3}}.
\end{align}
By virtue of \thref{thm:Mixed} we can further deduce that
\begin{align}
\hat{\gamma}_2(C,XYZ)
&=\frac{1}{6}[3+\sqrt{3}+(3-\sqrt{3})C\lsp].
\end{align}
Thanks to \pref{pro:Opt3setting} below, the $XYZ$ protocol is optimal among all three-setting protocols with respect to $\gamma_2^\ast(\mu)$ and $\hat{\gamma}_2(C,\mu)$.

\begin{proposition}\label{pro:Opt3setting}
	Suppose $\mu$ is a three-setting protocol specified by the weighted set $\{\bm{r}_j, p_j\}_{j=1}^3$. Then 
	\begin{align}
	g^\ast(\mu)&\geq \frac{1}{\sqrt{3}},\quad
	\gamma_2^\ast(\mu)\geq \frac{3+\sqrt{3}}{6},    \label{eq:ggammaLB3setting}\\
	\hat{\gamma}_2(C,\mu)
	&\geq \frac{1}{6}[3+\sqrt{3}+(3-\sqrt{3})C\lsp].\label{eq:gammahatLB3setting}
	\end{align}
	The  inequalities in \eref{eq:ggammaLB3setting} are saturated iff 
	$\bm{r}_1,\bm{r}_2,\bm{r}_3$ are mutually orthogonal and $p_1=p_2=p_3=1/3$. When $0\leq C<1$, the  inequality in \eref{eq:gammahatLB3setting} is saturated iff the same conditions hold. 
\end{proposition}

\begin{proof}
	The two inequalities in \eref{eq:ggammaLB3setting} are equivalent by \eref{eq:gammaCmu} in the main text and they imply 
	\eref{eq:gammahatLB3setting} thanks to \thref{thm:Mixed}.  In addition, $\hat{\gamma}_2(C=1,\mu)=1$, and the  inequality in \eref{eq:gammahatLB3setting} for any $0\leq C<1$ implies the two inequalities in \eref{eq:ggammaLB3setting}.
	So  it suffices to consider the first inequality $g^\ast(\mu)\geq 1/\sqrt{3}$. 
	If one of the three probabilities $p_1,p_2,p_3$ is equal to 0, then $\mu$ reduces to a two-setting protocol, in which case we have
	\begin{align}
	g^\ast(\mu)\geq\frac{1}{\sqrt{2}}> g^\ast(XYZ)=\frac{1}{\sqrt{3}}.
	\end{align}
	So the inequalities in \eref{eq:ggammaLB3setting} hold and cannot be saturated.

	Next, we consider the case in which $p_1,p_2,p_3>0$.
	Note that the protocol $\mu$ is equivalent to the protocol $\mu'$ specified by the following weighted set
	\begin{align}
	\{&\{\bm{r}_1,p_1/2\},\{-\bm{r}_1,p_1/2\},
	\{\bm{r}_2,p_2/2\},  \nonumber\\
	&\{-\bm{r}_2,p_2/2\},
	\{\bm{r}_3,p_3/2\},\{-\bm{r}_3,p_3/2\}\}, 
	\end{align}
	which satisfies  \eref{eq:CenterSymProt} with $\mu$ replaced by  $\mu'$.
	According to \lref{lem:g*(mu)=2max},  we have
	\begin{align}\label{eq:g0of3setting}
	&g^\ast(\mu)=g^\ast(\mu')   \nonumber\\
	&= \max\{
	|p_1\bm{r}_1+p_2\bm{r}_2+p_3\bm{r}_3|,\,
	|p_1\bm{r}_1+p_2\bm{r}_2-p_3\bm{r}_3|,\nonumber\\
	&\quad\quad\quad\;\;\,                 |p_1\bm{r}_1-p_2\bm{r}_2+p_3\bm{r}_3|,\,
	|p_1\bm{r}_1-p_2\bm{r}_2-p_3\bm{r}_3|
	\}\nonumber\\
	&= \max\Big\{
	\sqrt{(p_1\bm{r}_1+p_2\bm{r}_2)^2+p_3^2+2p_3\bm{r}_3\cdot(p_1\bm{r}_1+p_2\bm{r}_2)},\nonumber\\
	&\qquad\qquad\sqrt{(p_1\bm{r}_1+p_2\bm{r}_2)^2+p_3^2-2p_3\bm{r}_3\cdot(p_1\bm{r}_1+p_2\bm{r}_2)},\nonumber\\
	&\qquad\qquad\sqrt{(p_1\bm{r}_1-p_2\bm{r}_2)^2+p_3^2+2p_3\bm{r}_3\cdot(p_1\bm{r}_1-p_2\bm{r}_2)},\nonumber\\
	&\qquad\qquad\sqrt{(p_1\bm{r}_1-p_2\bm{r}_2)^2+p_3^2-2p_3\bm{r}_3\cdot(p_1\bm{r}_1-p_2\bm{r}_2)}\lsp
	\Big\}\nonumber\\
	&\geq \max\Big\{
	\sqrt{(p_1\bm{r}_1+p_2\bm{r}_2)^2+p_3^2},\,
	\sqrt{(p_1\bm{r}_1-p_2\bm{r}_2)^2+p_3^2}
	\lsp \Big\}\nonumber\\
	&= \max\Big\{
	\sqrt{p_1^2+p_2^2+p_3^2+2p_1p_2 (\bm{r}_1\cdot\bm{r}_2)},\nonumber\\
	&\qquad\qquad  \sqrt{p_1^2+p_2^2+p_3^2-2p_1p_2 (\bm{r}_1\cdot\bm{r}_2)}
	\lsp \Big\} \nonumber\\
	&\geq\sqrt{p_1^2+p_2^2+p_3^2}
	\geq\frac{p_1+p_2+p_3}{\sqrt{3}}
	=\frac{1}{\sqrt{3}},
	\end{align}
	which confirms \eref{eq:ggammaLB3setting}.
	Here the last inequality is saturated iff $p_1=p_2=p_3=1/3$; the second inequality is saturated iff $\bm{r}_1\perp\bm{r}_2$ given that $p_1, p_2, p_3>0$. 
	When $\bm{r}_1\perp\bm{r}_2$, the first inequality is saturated iff $\bm{r}_3\perp\bm{r}_1$ and $\bm{r}_3\perp\bm{r}_2$. 
	Therefore, for a three-setting protocol $\mu$, the inequality  $g^\ast(\mu)\geq 1/\sqrt{3}$ is saturated  iff $\bm{r}_1,\bm{r}_2,\bm{r}_3$ are mutually orthogonal and $p_1=p_2=p_3=1/3$.
\end{proof}

\subsection*{Optimal protocol}
Since $g(C,\mu)$ is convex in $\mu$ by \lref{lem:Convex} and is invariant under rotations on the Bloch sphere, it is minimized when $\mu$ is the uniform distribution on the Bloch sphere, which leads to the  isotropic protocol. Then  the integration in \eref{eq:gCmu} in the main text is independent of the unit vector $\bm{v}$, so any direction is an intelligent direction. For simplicity we can choose $\bm{v}=(0,0,1)^\rmT$, which yields
\begin{align}
&g(C, \mu )=\int\rmd\mu(\bm{r}) \sqrt{C^2+(1-C^2)(\bm{r}\cdot\bm{v})^2}\nonumber\\
&=\frac{1}{4\pi}\int_{\alpha=0}^\pi \rmd\alpha\int_{\phi=0}^{2\pi}\rmd\phi\sin\alpha \sqrt{C^2+(1-C^2)\cos^2\alpha}\nonumber\\
&=\int_{0}^1\rmd x \sqrt{C^2+(1-C^2)x^2}\nonumber\\
&=\frac{1}{2}+\frac{C^2\mathrm{arcsinh}(\frac{\sqrt{1-C^2}}{C})}{2\sqrt{1-C^2}}.
\end{align}
In the limit  $C\rightarrow 0$, this equation yields
\begin{align}
g^\ast(\mu)=\frac{1}{2},\quad \gamma_2^\ast( \mu )=\frac{3}{4}.
\end{align}
By virtue of \thref{thm:Mixed} we can further deduce that
\begin{align}
\hat{\gamma}_2(C,\mu)
&=\frac{1+C}{2}+\frac{1-C}{4}  =\frac{3+C}{4}.
\end{align}

The above discussion yields the following proposition.
\begin{proposition}
	Every verification protocol $\mu$ of the Bell state satisfies
	\begin{align}
	g(C, \mu)&\geq \frac{1}{2}+\frac{C^2\mathrm{arcsinh}(\frac{\sqrt{1-C^2}}{C})}{2\sqrt{1-C^2}}, \\
	\gamma_2(C, \mu)&\geq \frac{3}{4}+\frac{C^2\mathrm{arcsinh}(\frac{\sqrt{1-C^2}}{C})}{4\sqrt{1-C^2}},\\
	\hat{\gamma}_2(C,\mu)
	&\geq \frac{3+C}{4}.
	\end{align}
	All inequalities are saturated for the isotropic protocol. 
\end{proposition}

\subsection*{\label{asec:Equator}Equator protocol}
When the support of $\mu$ is contained in the equator, $g(C,\mu)$ is minimized when $\mu$ is the uniform distribution on the equator. The resulting protocol is
called the equator protocol,  which is equivalent to the $\theta$-protocol proposed in \rcite{McCut2016SM}. In this case, the integration in \eref{eq:gCmu}  in the main text is maximized when  $\bm{v}$ is any unit vector in the $xy$-plane. In other words, any direction in the $xy$-plane is an intelligent direction.  To simplify the calculation we can choose $\bm{v}=(1,0,0)^\rmT$, which yields
\begin{align}
g(C, \mu )&=\int\rmd\mu(\bm{r}) \sqrt{C^2+(1-C^2)(\bm{r}\cdot\bm{v})^2}\nonumber\\
&=\frac{1}{2\pi} \int_{0}^{2\pi}\rmd\phi \sqrt{C^2+(1-C^2)\cos^2\phi}\nonumber\\
&=\frac{2}{\pi} \int_{0}^{\pi/2}\rmd\phi \sqrt{1-(1-C^2)\sin^2\phi}\nonumber\\
&=\frac{2}{\pi}K\bigl(\sqrt{1-C^2}\lsp\bigr),
\end{align}
where $K(\sqrt{1-C^2}\lsp)$ is a complete elliptic integral of the second kind. In conjunction with \eref{eq:gammaCmu}  in the main text we get
\begin{align}
\gamma_2(C, \mu )&=\frac{1}{2}+\frac{1}{\pi}K\bigl(\sqrt{1-C^2}\lsp\bigr).
\end{align}
When $C=0$, we have $K(\sqrt{1-C^2}\lsp)=1$ and
\begin{align}
g^\ast( \mu )&=\frac{2}{\pi},\quad \gamma_2^\ast( \mu )=\frac{1}{2}+\frac{1}{\pi}.
\end{align}
By virtue of \thref{thm:Mixed} we can deduce that
\begin{align}
\hat{\gamma}_2(C,\mu)
&=\frac{1+C}{2}+\frac{1-C}{\pi}  =\frac{1}{2\pi}[\pi+2+(\pi-2)C\lsp].
\end{align}

The above discussion yields the following proposition. 
\begin{proposition}
	Suppose $\mu$ is supported on the equator of the Bloch sphere. Then 
	\begin{align}
	g(C, \mu )&\geq \frac{2}{\pi}K\bigl(\sqrt{1-C^2}\lsp\bigr),\\
	\gamma_2(C, \mu )&\geq\frac{1}{2}+\frac{1}{\pi}K\bigl(\sqrt{1-C^2}\lsp\bigr),\\
	\hat{\gamma}_2(C,\mu)
	&\geq\frac{1}{2\pi}[\pi+2+(\pi-2)C\lsp].
	\end{align}
	All inequalities are saturated for the equator protocol. 
\end{proposition}

\subsection*{Polygon protocols}
Here we consider verification protocols defined by regular polygons inscribed on the equator of the Bloch sphere. For simplicity, all  measurement settings are chosen with the same  probability,  which is the optimal choice according to  \lref{lem:Transitive}. The resulting protocols are called polygon protocols, which may be regarded as  discrete versions of the equator protocol.

Consider a polygon with $M$ vertices, without loss of generality, we can assume that the $M$ unit vectors defining the polygon have the form
\begin{equation}
\bm{r}_j=(\cos\theta_j,\sin\theta_j,0)^\rmT,
\end{equation}
where
\begin{equation}
\theta_j =\frac{2(j-1)\pi }{M},\quad  j=1,2,\ldots, M.
\end{equation}
Let $g(C,M):=g(C,\{r_j\}_{j=1}^M)$ and $g^\ast(M):=g(0,M)$. Then we have
\begin{align}
&g(C, M)
=\frac{1}{M}\max_{\bm{v}}\sum_{j=1}^M \sqrt{C^2+(1-C^2)(\bm{r}_j\cdot\bm{v})^2}\nonumber\\
&=\frac{1}{M}\max_{0\leq \varphi<2\pi}\sum_{j=0}^{M-1} \sqrt{C^2+(1-C^2)\cos^2{\left(\frac{ 2\pi j}{M}-\varphi\right)} }, \label{eq:gCM}
\end{align}
where the maximization in the first line is taken over all unit vectors on the Bloch sphere. When $M$ goes to infinity, the polygon protocol approaches the  equator protocol.

In general it is not easy  to derive an analytical formula
for $g(C,\mu)$. Here we focus on the special case $C=0$,
which yields the following proposition.
\begin{proposition}\label{pro:PolyC=0}
	For any integer $M\geq 3$ we have
	\begin{align}\label{eq:PolyC=0}
	g^\ast(M)=\begin{cases}
	2\left[M\sin\left(\frac{\pi}{M}\right)\right]^{-1}  & M \ \rm{ even},\\[0.5ex]
	\left[M\sin\left(\frac{\pi}{2M}\right)\right]^{-1}  & M \ \rm{ odd}.
	\end{cases}
	\end{align}
\end{proposition}

\Pref{pro:PolyC=0} above and \eref{eq:gammaCmu} in the main text together yield
\begin{align}
\gamma_2^\ast(M)
=\begin{cases}
\frac{1}{2}+\left[M\sin\left(\frac{\pi}{M}\right)\right]^{-1}  & M \ \rm{ even},\\[0.5ex]
\frac{1}{2}+\left[2M\sin\left(\frac{\pi}{2M}\right)\right]^{-1}  & M \ \rm{ odd}.
\end{cases}
\end{align}
By virtue of \thref{thm:Mixed} we can further deduce that
\begin{align}
\hat{\gamma}_2(C,M)
=\begin{cases}
\frac{1+C}{2}+\frac{1-C}{M \sin(\frac{\pi}{ M})}  & M \ \rm{ even},\\[1ex]
\frac{1+C}{2}+\frac{1-C}{2M\sin(\frac{\pi}{2M})}  & M \ \rm{ odd}.
\end{cases}
\end{align}

\begin{proof}[Proof of \pref{pro:PolyC=0}]
	When $C=0$, \eref{eq:gCM} yields 
	\begin{align}
	g^\ast(M)&= \frac{1}{M}\max_{\bm{v}}\sum_{j=1}^M |\bm{r_j}\cdot\bm{v}|\nonumber \\
	&=\frac{1}{M}
	\max_{0\leq \varphi<2\pi}  \sum_{j=0}^{M-1}
	\left|  \cos{\left(\frac{ 2\pi j}{M}-\varphi\right)} \right|. \label{g0Mproof2} 
	\end{align}
	When $M$ is even, let $S=\{\bm{r}_1, \bm{r}_2, \ldots, \bm{r}_{M/2}\}$ be a set of $M/2$ unit vectors, which correspond to $M/2$ consecutive vertices of the polygon. Then \lref{lem:CenterSym} implies that
	\begin{align}
	g^\ast(M)&=\frac{2}{M}|\bm{\eta}(S)|=\frac{2}{M}\left|\sum_{j=1}^{M/2} \bm{r}_j\right|=\frac{2}{M}\left|\sum_{j=1}^{M/2} \rme^{\rmi\theta_j}\right|\nonumber\\
	&=\frac{2}{M}\frac{2}{|1-\rme^{2\pi\rmi/M}|}=\frac{2}{M\sin\bigl(\frac{\pi}{ M}\bigr)},
	\end{align}
	which confirms \eref{eq:PolyC=0}. 
	Alternatively, this result can be derived directly by virtue of  \eref{g0Mproof2}.
	
	When $M$ is even but not divisible by 4, \lref{lem:CenterSym} also implies that each intelligent direction coincides with 
	some  $\bm{r}_j$. When $M$ is  divisible by 4, by contrast, each  intelligent direction passes through the middle point of some edge of the polygon.

	When $M$ is odd, the polygon protocol with $M$ measurement settings is equivalent to the polygon protocol  with $2M$ measurement settings. So each intelligent direction coincides with 
	some  $\bm{r}_j$ or $-\bm{r}_j$,  and we have
	\begin{align}
	g^\ast(M)=g^\ast(2M)
	=\left[M\sin\left(\frac{\pi}{2M}\right)\right]^{-1}.
	\end{align}
	This observation completes the proof of  \pref{pro:PolyC=0}.  
\end{proof}

When $C>0$,  numerical calculation indicates that each intelligent direction for $C=0$  is still an intelligent direction for $C>0$. 
To be concrete, when $M$ is not divisible by 4, any unit vector parallel or antiparallel to $\bm{r}_j$ is an intelligent direction.  When $M$ is divisible by 4, by contrast, any unit vector passing through the middle point of some edge of the polygon is an intelligent direction. This observation leads to the following conjecture.

\begin{widetext}
	\begin{conjecture}\label{con:gCM}
		Suppose $M\geq3$ is an integer and $0\leq C\leq1$; then we have
		\begin{align}\label{eq:PolygonM}
		g(C,M)=\begin{cases}
		\frac{1}{M}\sum_{j=0}^{M-1}\sqrt{C^2+(1-C^2)\cos^2\bigl(\frac{(2j-1)\pi}{M}\bigr)} &  4| M, \\[1ex]
		\frac{1}{M}\sum_{j=0}^{M-1}\sqrt{C^2+(1-C^2)\cos^2\bigl(\frac{2j\pi}{M}\bigr)}  & 4\nmid M. 
		\end{cases}
		\end{align}
	\end{conjecture}
\end{widetext}

When $M=3$, the maximization over $\varphi$ in \eref{eq:gCM} can be solved directly by considering the derivative over $\varphi$, which yields 
\begin{equation}\label{eq:Polygon3}
g(C,3)=\frac{1+\sqrt{1+3C^2}}{3},\quad \gamma_2(C,3)=\frac{4+\sqrt{1+3C^2}}{6}
\end{equation}
and confirms \cref{con:gCM}.
When $C=0$,  this equation reduces to 
\begin{equation}
g^\ast(3)=\frac{2}{3},\quad \gamma_2^\ast(3)=\frac{5}{6},
\end{equation}
which agrees with \pref{pro:PolyC=0}. 

In the limit $M\rightarrow\infty$, the polygon protocol approaches the equator protocol; accordingly, the above results converge to the counterparts for the equator protocol.

\subsection*{\label{asec:EquatorZ}Equator$+Z$ protocol}
Equator$+Z$ protocol is another protocol of special interest, especially in the verification of GHZ states.  The corresponding distribution $\mu$ is concentrated on the equator together with the north and south poles. To minimize $g(C,\mu)$, the distribution $\mu$ should be uniform on the equator. Let $p$ be the total probability assigned to  the north and south poles.  Thanks to the rotation symmetry around the $z$-axis, \eref{eq:gCmu} can be simplified as follows, 
\begin{align}
&g(C, \mu )=\max_{0\leq \alpha\leq \pi/2} \biggl[ p\sqrt{C^2+(1-C^2)\cos^2\alpha}\nonumber\\
&+ \frac{1-p}{2\pi} \int_{0}^{2\pi}\rmd\phi \sqrt{C^2+(1-C^2)\sin^2\alpha\cos^2\phi}\,\biggr].
\end{align}
In general, it is not easy to derive an analytical formula for $g(C,\mu)$. Here we focus on the special case $C=0$, which yields
\begin{align}
&g^\ast( \mu )=\max_{0\leq \alpha\leq \pi/2} \biggl[ p\cos\alpha +\frac{2(1-p)}{\pi}\sin\alpha \biggr]\nonumber\\
&=\sqrt{p^2+\frac{4(1-p)^2}{\pi^2}}=\sqrt{\frac{(4+\pi^2)p^2-8p+4}{\pi^2}}. \label{eq:g*muEquatorZ}
\end{align}
The minimum of $g^\ast(\mu)$ is attained when $p=4/(4+\pi^2)$, in which case we have
\begin{equation}
g^\ast(\mu)=\frac{2}{\sqrt{4+\pi^2}}, \quad \gamma_2^\ast( \mu )=\frac{1}{2}+\frac{1}{\sqrt{4+\pi^2}}.
\end{equation}
This protocol is referred to as the equator$+Z$ protocol.
In addition, by virtue of \thref{thm:Mixed} we can deduce that
\begin{align}\label{eq:gahatEquatorZopt}
\hat{\gamma}_2(C,\mu)
&=\frac{1+C}{2}+\frac{1-C}{\sqrt{4+\pi^2}}.
\end{align}

The above discussion yields the following proposition. 
\begin{proposition}
	Suppose $\mu$ is supported on the equator of the Bloch sphere together with the north and south poles. Then 
	\begin{align}
	g^\ast(\mu)&\geq \frac{2}{\sqrt{4+\pi^2}}, \quad \gamma_2^\ast( \mu )\geq \frac{1}{2}+\frac{1}{\sqrt{4+\pi^2}},\\
	\hat{\gamma}_2(C,\mu)
	&\geq\frac{1+C}{2}+\frac{1-C}{\sqrt{4+\pi^2}}.
	\end{align}
	All inequalities are saturated for the equator$+Z$ protocol. 
\end{proposition}

By contrast, the protocol corresponding to  $p=1/3$  is referred to as the equator$+{Z}_{II}$-protocol. This protocol is interesting because it is balanced and  is optimal in the limit $C\rightarrow 1$ if we only consider pure states. However, it is not optimal when $C=0$, in which case we have
\begin{equation}
g^\ast(\mu)=\frac{\sqrt{16+\pi^2}}{3\pi}, \quad \gamma_2^\ast( \mu )=\frac{1}{2}+\frac{\sqrt{16+\pi^2}}{6\pi}.
\end{equation}
This example shows that the optimal equator$+Z$ protocol depends on $C$, and no choice of $p$ is optimal for all values of $C$. Compared with  $\hat{\gamma}_2(C,\mu)$,  the behavior of $\gamma_2(C,\mu)$ is more complicated, and it may be difficult to compare the performances of different verification protocols based on $\gamma_2(C,\mu)$. 
Nevertheless, the choices $p=4/(4+\pi^2)$ and $p=1/3$ are both nearly optimal for all values of $C$.

\bigskip

\subsection*{\label{asec:PolyZ}Polygon$+Z$ protocols}
The polygon$+Z$ protocol is constructed by replacing the uniform distribution on the equator in the  equator$+Z$ protocol with  the uniform distribution on the vertices of a regular polygon. Let $p$ be the total probability assigned to  the north and south poles and let $M$ be the number of vertices of the polygon. Then $g(C,\mu)$ is determined by the three parameters $C,M,p$; therefore, it is more informative to write   $g(C,\{M,p\})$ and 	$g^\ast(\{M,p\})$ in place of $g(C,\mu)$ and $g^\ast(\mu)$, respectively. The following proposition determines the value of $g^\ast(\{M,p\})$.

\begin{proposition}\label{pro:Poly+ZC=0}
	Suppose  $0\leq p\leq 1$ and $M\geq 3$ is an integer. Then 
	\begin{align}\label{eq:Poly+Zg(0)}
	g^\ast(\{M,p\})
	=\begin{cases}
	\sqrt{p^2+ \frac{4(1-p)^2}{M^2\sin^2(\frac{\pi}{M})} }  & M \ \rm{ even},\\[1ex]
	\sqrt{p^2+ \frac{(1-p)^2}{M^2\sin^2(\frac{\pi}{2M})} }  & M \ \rm{ odd}.
	\end{cases}
	\end{align}
\end{proposition}
In the special case $p=0$, the Polygon$+Z$ protocol reduces to the Polygon protocol, and \pref{pro:Poly+ZC=0} reduces  to \pref{pro:PolyC=0} as expected.

\begin{proof}
	By virtue of \eref{eq:gC0rp} we can deduce the following result,
	\begin{align}\label{eq:Poly+ZMaxgC=0Meven}
	&\!\!\! g^\ast(\{M,p\}) = \max_{\bm{v}}  \Biggl(p \left|\hat{z}\cdot\bm{v}\right| + \frac{1-p}{M}  \sum_{j=1}^{M}
	\left| \bm{r_j}\cdot\bm{v} \right| \Biggr) \nonumber\\
	&\!\!\! = \max_{0\leq \alpha\leq \frac{\pi}{2}}  \Biggl[p \cos\alpha  + \frac{1-p}{M} (\sin\alpha) \max_{\bm{u}}\sum_{j=1}^{M}
	\left| \bm{r_j}\cdot\bm{u} \right| \Biggr], \end{align}
	where the maximization in the brackets in the second line is taken over all unit vectors in the $xy$-plane, and the result is tied to the counterpart for  the polygon protocol. Therefore,
	\begin{align}
	g^\ast(\{M,p\})
	& = \max_{0\leq \alpha\leq \pi/2}  \left[\lsp p \cos\alpha  +(1-p)  g^\ast(M)\sin\alpha \lsp\right]\nonumber\\
	&=\sqrt{p^2 +(1-p)^2g^\ast(M)^2},
	\end{align}	
	where $ g^\ast(M)$ is presented in \eref{eq:PolyC=0}. By inserting the explicit expression for $g^\ast(M)$, we can derive
	\eref{eq:Poly+Zg(0)} immediately, which completes the proof of \pref{pro:Poly+ZC=0}.
\end{proof}

\bigskip

To determine the optimal Polygon$+Z$ protocol and the corresponding  probability $p_{\rm{opt}}$  of performing $Z$ measurement in the case $C=0$, we need to minimize $g^\ast(\{M,p\})$ over $p$. Simple calculation based on \pref{pro:Poly+ZC=0} shows that
\begin{align}
&\min_p \, g^\ast(\{M,p\})
= g^\ast(\{M,p_{\rm{opt}}\}) \nonumber\\
&=\begin{cases}
\left[1+ \frac{1}{4}M^2\sin^2\left(\frac{\pi}{M}\right)\right]^{-1/2}   & M \ \rm{ even},\\[0.5ex]
\left[1+ M^2\sin^2\left(\frac{\pi}{2M}\right)\right]^{-1/2}            & M \ \rm{ odd},
\end{cases}
\end{align}
where 
\begin{align}\label{eq:Poly+ZC=0optp}
p_{\rm{opt}}:=\begin{cases}
\left[1+ \frac{1}{4}M^2\sin^2\left(\frac{\pi}{M}\right)\right]^{-1}   & M \ \rm{ even},\\[0.8ex]
\left[1+ M^2\sin^2\left(\frac{\pi}{2M}\right)\right]^{-1}            & M \ \rm{ odd}, 
\end{cases}
\end{align}
is the desired  optimal probability.

\bigskip

In conjunction with \thref{thm:MGP}  in the main text, we can now determine the optimal threshold in the guessing probability, with the result
\begin{widetext}
	\begin{align}
	&\gamma_2^\ast(\{M,p_{\rm{opt}}\})=\begin{cases}
	\frac{1}{2}+\left[4+ M^2\sin^2\left(\frac{\pi}{M}\right)\right]^{-1/2}  & M \ \rm{ even},\\[0.5ex]
	\frac{1}{2}+\left[4+ 4M^2\sin^2\left(\frac{\pi}{2M}\right)\right]^{-1/2}   & M \ \rm{ odd}.
	\end{cases}
	\end{align}
	By virtue of \thref{thm:Mixed} we can further deduce that
	\begin{align}
	&\hat{\gamma}_2(C,\{M,p_{\rm{opt}}\})=\begin{cases}
	\frac{1+C}{2}+(1-C)\left[4+ M^2\sin^2\left(\frac{\pi}{M}\right)\right]^{-1/2}  & M \ \rm{ even},\\[0.5ex]
	\frac{1+C}{2}+(1-C)\left[4+ 4M^2\sin^2\left(\frac{\pi}{2M}\right)\right]^{-1/2}  & M \ \rm{ odd}.
	\end{cases}
	\end{align}

	In the limit $M\rightarrow\infty$, the polygon$+Z$ protocol approaches the equator$+Z$ protocol; accordingly, the above results converge to the counterparts of the equator$+Z$ protocol. For example, 
	\begin{gather}
	\lim_{M\rightarrow \infty} g^\ast(\{M,p\})=\sqrt{p^2+ \frac{4(1-p)^2}{\pi^2 }},  \qquad   \lim_{M\rightarrow \infty} p_{\rm{opt}}=\frac{4}{4+\pi^2}, \qquad  \lim_{M\rightarrow \infty} g^\ast(\{M,p_{\rm{opt}}\})=\frac{2}{\sqrt{4+\pi^2}},\\
	\lim_{M\rightarrow \infty}\gamma_2^\ast(\{M,p_{\rm{opt}}\}) =\frac{1}{2}+\frac{1}{\sqrt{4+\pi^2}}, \qquad 
	\lim_{M\rightarrow \infty}\hat{\gamma}_2(C,\{M,p_{\rm{opt}}\})=
	\frac{1+C}{2}+\frac{1-C}{\sqrt{4+\pi^2}}.
	\end{gather}
	All these limits coincide with the corresponding results for the equator$+Z$ protocol  presented in  \esref{eq:g*muEquatorZ}-\eqref{eq:gahatEquatorZopt} as expected.
\end{widetext}

\subsection*{Tetrahedron,  cube, and octahedron protocols}
In this and the following subsections we consider verification protocols based on  platonic solids. A platonic solid with $M$ vertices can be specified by a set $\{\bm{r}_j\}_{j=1}^M$ of  $M$ unit vectors that correspond to the vertices. Here we assume that all measurements associated with these unit vectors are chosen with the same probability of $1/M$, which is the optimal choice according to  \lref{lem:Transitive}. Note that the octahedron protocol is equivalent to the $XYZ$ protocol. In addition, the cube protocol is equivalent to the tetrahedron protocol, so it suffices to consider the tetrahedron protocol here.

For the tetrahedron protocol, the verification matrix reads $\Xi(\mu)=\mathbb{I}/3$, and the bound in \lref{lem:gCUB} is saturated if  $\bm{v}$ passes through the middle point of an edge. Therefore, we have
\begin{align}
g(C,\mu)
&=\sqrt{\frac{1+2C^2}{3}},
\quad \gamma_2(C,\mu)
=\frac{1}{2}+\frac{1}{2}\sqrt{\frac{1+2C^2}{3}},
\end{align} 
which are identical to the results on the $XYZ$ protocol. On the other hand, there are only six intelligent directions instead of eight, assuming  $0\leq C<1$.  When $C=0$, the above  equation reduces to 
\begin{align}
g^\ast(\mu)&=\frac{1}{\sqrt{3}},\quad
\gamma_2^\ast(\mu)=\frac{1}{2}+\frac{1}{2\sqrt{3}}.
\end{align}
By virtue of \thref{thm:Mixed} we can further deduce that
\begin{align}
\hat{\gamma}_2(C,\mu)
&=\frac{1}{6} [3+\sqrt{3}+(3-\sqrt{3})C\lsp] .
\end{align}
Here it is worth pointing out  that $\gamma_2^\ast(\mu)$ is larger than the counterpart of  the optimal polygon(3)$+Z$ protocol (cf.~\tref{tab:ProtocolGHZDH}). So  the  tetrahedron protocol is not the  optimal four-setting protocol. In general, the verification protocol based on a platonic solid is not necessarily optimal among protocols with the same number of measurement settings.

\subsection*{Icosahedron and dodecahedron protocols}
Here we consider verification protocols based on the two remaining platonic solids, namely, icosahedron and dodecahedron.
Both  platonic solids are center symmetric, so $g^\ast(\mu)$ can be determined by virtue of \lref{lem:CenterSym}.

To be concrete, the 12 vertices of the icosahedron are chosen to be
\begin{align}
\frac{(0,\pm1,\pm\tau)}{\sqrt{1+\tau^2}}^\rmT, \quad
\frac{(\pm\tau,0,\pm1)}{\sqrt{1+\tau^2}}^\rmT, \quad
\frac{(\pm1,\pm\tau,0)}{\sqrt{1+\tau^2}}^\rmT,
\end{align}
where $\tau=(1+\sqrt{5})/2$ is the golden ratio. According to \lref{lem:CenterSym}, to derive $g^\ast(\mu)$, it suffices to compare $\left|\bm{\eta}(S)\right|$ [cf. \eref{eq:etaS}]
for all subsets $S$ that contain six vertices on the same side of a plane passing through the origin.
Note that an optimal set $S$ contains one and only one vertex in each pair of antipodal vertices.
For any plane passing through the origin, there exists at least one face of the icosahedron on each side of the plane.
Therefore, we can choose three vertices on the same face as fixed elements of $S$ and then choose one
vertex from each of the remaining three pairs of antipodal vertices. There are $2^3=8$ different choices, but only three equivalent classes under
orthogonal transformations.
Moreover, in one of the  classes,
the vertices in each set are not on the same side of a plane passing through the origin. Therefore,
it suffices  to compute $\left|\bm{\eta}(S)\right|$ for the remaining two cases:
\begin{enumerate}
	\item[1.] $S$ contains six vertices, one of which is adjacent to the other five. In this case we have
	\begin{align}
	\left|\bm{\eta}(S)\right|
	&=2\sqrt{\frac{2\tau^2+2\tau+1}{\tau^2+1}}
	=1+\sqrt{5}
	\approx3.236.
	\end{align}
	
	\item[2.] The six vertices contained in $S$ are on one face (labeled by $A$) of the icosahedron and the three faces adjacent to $A$.
	In this case we have
	\begin{align}
	\left|\bm{\eta}(S)\right|
	&=2\sqrt{\frac{(\tau+1)^2+1}{\tau^2+1}}
	\approx2.947.
	\end{align}
\end{enumerate}
By comparison, $\left|\bm{\eta}(S)\right|$ attains the maximum in the first case, so we have
\begin{align}\label{eq:g0Icosahedron}
g^\ast(\mu)
&=\frac{2}{12}\max_{S}\,\left|\bm{\eta}(S) \right|
=\frac{1+\sqrt{5}}{6}
\approx0.539
\end{align}
by \eref{eq:g*(mu)CSS}.
Accordingly, a unit vector  is an intelligent direction at $C=0$ iff
it corresponds to one of the 12 vertices of the icosahedron.

As immediate corollaries of \eref{eq:g0Icosahedron}, we can derive
\begin{align}
\gamma_2^\ast(\mu)&=\frac{7+\sqrt{5}}{12} \approx 0.770 , \\
\hat{\gamma}_2(C,\mu)
&=\frac{1}{12}[7+\sqrt{5}+(5-\sqrt{5})C\lsp]  .
\end{align}
Incidentally,  these results are very close to the  counterparts of the equator+$Z_{II}$-protocol.

For the icosahedron protocol, numerical calculation indicates that  each intelligent direction at $C=0$  is still an intelligent direction for  $C>0$.
In other words, \eref{eq:gCmu} in the main text is maximized when $\bm{v}$ corresponds to one vertex of the icosahedron. This observation leads to the following conjecture. 
\begin{conjecture}\label{con:Icosahedron}
	The icosahedron protocol $\mu$ satisfies the following equations,
	\begin{align}
	g(C,\mu)
	&=\frac{1+\sqrt{5(1+4C^2)}}{6},\\
	\gamma_2(C,\mu)
	&=\frac{7+\sqrt{5(1+4C^2)}}{12}.
	\end{align}
	When $0\leq C<1$, any intelligent direction corresponds to a vertex of the icosahedron, and vice versa. 
\end{conjecture}

Next,  we consider the dodecahedron protocol.
To be concrete, the 10 pairs of  antipodal vertices of the  dodecahedron are chosen to be 
\begin{equation}
\begin{aligned}
& \frac{1}{\sqrt{3}}(0,\pm\tau,\pm\frac{1}{\tau})^\rmT, \quad
\frac{1}{\sqrt{3}}(\pm\frac{1}{\tau},0,\pm\tau)^\rmT, \\
&\frac{1}{\sqrt{3}}(\pm\tau,\pm\frac{1}{\tau},0)^\rmT, \quad
\frac{1}{\sqrt{3}}(\pm1,\pm1,\pm1)^\rmT.
\end{aligned}
\end{equation}
Thanks to \lref{lem:CenterSym} again,  to derive $g^\ast(\mu)$ for this protocol, it suffices to compare $\left|\bm{\eta}(S)\right|$
for all subsets $S$ that contain 10 vertices located on the same side of a plane passing through the origin.
In addition, any optimal  set $S$ contains one and only one vertex in each of the 10 pairs of antipodal vertices.
Note that for any plane that passes through the origin, there is at least one face of the dodecahedron on each side of the plane.
Therefore, we can choose five vertices on the same face as fixed elements of $S$, and then choose one
vertex from each of the remaining five pairs. There are $2^5=32$ different choices, but only  seven equivalent classes.
Moreover, in four of the seven classes, the vertices in each set are not  on the same side of a plane passing through the origin.
So it suffices  to compare $\left|\bm{\eta}(S)\right|$ for the remaining three classes:
\begin{enumerate}
	\item[1.] The 10 vertices contained in $S$ are on three faces of the dodecahedron which are mutually adjacent. In this case we have
	\begin{align}
	\left|\bm{\eta}(S)\right|
	&=2\sqrt{\frac{(\tau\!+\!2\!+\!1/\tau)^2+\tau^2}{3}}
	=3+\sqrt{5}
	\approx5.236.
	\end{align}
	
	\item[2.] $S$ contains 10 vertices, five of which are on one face (labeled by $A$) of the dodecahedron, and each of the other five vertices is adjacent to a vertex of $A$. 
	In this case we have
	\begin{align}
	\left|\bm{\eta}(S)\right|
	&=2\sqrt{\frac{(\tau\!+\!1/\tau)^2+(\tau\!+\!2)^2}{3}}
	\approx4.911.
	\end{align}
	
	\item[2.] $S$ contains 10 vertices. Five of them are on one face (labeled by $A$) of the dodecahedron; four of them are adjacent to four vertices of $A$, respectively.
	The last element of $S$ is the antipodal point of a vertex that is adjacent to $A$.
	In this case we have
	\begin{align}
	\left|\bm{\eta}(S)\right|
	&=2\sqrt{\frac{(\tau+2)^2+2\tau^2}{3}}
	\approx4.943.
	\end{align}
\end{enumerate}
By comparison, $\left|\bm{\eta}(S)\right|$ attains the maximum in the first case, so
we have 
\begin{align}\label{eq:g0Dodecahedron}
g^\ast(\mu)
&=\frac{2}{20}\max_{S}\,\left|\bm{\eta}(S) \right|
=\frac{3+\sqrt{5}}{10}
\approx0.524
\end{align}
by  \eref{eq:g*(mu)CSS}.
Accordingly, a unit vector  is an intelligent direction at $C=0$ iff
it corresponds to one of the 20 vertices of the dodecahedron.

As immediate corollaries of \eref{eq:g0Dodecahedron}, we can deduce
\begin{align}
\gamma_2^\ast(\mu)&=\frac{13+\sqrt{5}}{20}  \approx 0.762 , \\
\hat{\gamma}_2(C,\mu)
&=\frac{1}{20}[13+\sqrt{5}+(7-\sqrt{5})C\lsp]  .
\end{align}

Numerical calculation indicates that each intelligent direction at $C=0$  is still an intelligent direction for  $C>0$. In other words, \eref{eq:gCmu} in the main text is maximized when $\bm{v}$ corresponds to one vertex of the dodecahedron. This observation leads to the following conjecture in analogy to \cref{con:Icosahedron}. 
\begin{conjecture}
	The dodecahedron protocol $\mu$ satisfies the following equations,
	\begin{align}
	g(C,\mu)
	&=\frac{1+\sqrt{5+4C^2}+2\sqrt{1+8C^2}}{10},\\
	\gamma_2(C,\mu)
	&=\frac{11+\sqrt{5+4C^2}+2\sqrt{1+8C^2}}{20}.
	\end{align}
	When $0\leq C<1$, any intelligent direction corresponds to a vertex of the dodecahedron, and vice versa.
\end{conjecture}

\section*{\label{app:GHZverify}SUPPLEMENTARY Note 10: Verification of GHZ states}
According to the main text, each verification protocol of the Bell state is determined by a probability distribution on the Bloch sphere, which specifies the probability of performing each projective measurement. In addition, each projective measurement is specified by a unit vector on the Bloch sphere. To verify the
$n$-qubit GHZ state
\begin{equation}
|G^n\>=\frac{1}{\sqrt{2}}(|0\>^{\otimes n}+| 1\>^{\otimes n}),\quad n\geq 3,
\end{equation}
we can simulate projective tests for the Bell state by local projective measurements on individual qubits. However, not all projective tests for the Bell state can be simulated in this way. To clarify this limitation, we need to introduce an additional concept.

\subsection*{\label{app:GHZcompatibleMeas}Compatible measurements}
Local projective measurements on the $n$ qubits of the  GHZ state $|G^n\>$ can be specified by  a set of $n$ unit vectors $\{\bm{r}_1, \bm{r}_2, \ldots, \bm{r}_n\}$, where $\bm{r}_j$ determines the local projective measurement on qubit $j$.  The set $\{\bm{r}_1, \bm{r}_2, \ldots, \bm{r}_n\}$ and the corresponding measurements are \emph{compatible} if each qubit $j$ has only two possible reduced states conditioned on the outcomes of all projective measurements on the other qubits when the target GHZ state is prepared; in addition,  the two possible reduced states happen to be the eigenstates of $\bm{r}_j\cdot\bm{\sigma}$, that is, $(\openone\pm\bm{r}_j\cdot\bm{\sigma})/2$.
Here the compatibility requirement  guarantees that the two reduced states of qubit $j$ can be distinguished with certainty by performing the  projective measurement associated with the Bloch vector $\bm{r}_j$. In addition, given the measurement outcomes of other qubits, measurement of the remaining qubit will yield one outcome with certainty, which corresponds to passing the test.
In this way, the local projective measurements specified by the set $\{\bm{r}_1, \bm{r}_2, \ldots, \bm{r}_n\}$ can be used to construct a nontrivial test for the  GHZ state $|G^n\>$ in which no projective measurement on any qubit is redundant.

Conversely, if the set $\{\bm{r}_1, \bm{r}_2, \ldots, \bm{r}_n\}$ is not compatible, then the measurement outcome of some party is unpredictable even if the measurement outcomes of all other parties are known. Consequently,  any outcome reported by this party would pass the test, and such a test is not useful in the presence of dishonest parties. To construct tests for verifying the GHZ state,  therefore, we can focus on compatible measurements, whose properties are clarified in the following lemma.
\begin{lemma}\label{lem:CompatibleMeas}
	The set $\{\bm{r}_1, \bm{r}_2, \ldots, \bm{r}_n\}$ of local projective measurements is compatible iff either one of the following two conditions holds:
	\begin{enumerate}
		\item $\bm{r}_j=(0,0,\pm 1)^\rmT$ for each $j$.
		
		\item $\bm{r}_j=(\cos \phi_j,\sin\phi_j,0)^\rmT$ with $\sum_j \phi_j=0\mod \pi$.
	\end{enumerate}
\end{lemma}
Here the condition $\sum_j \phi_j=0\mod \pi$ means $\sum_j \phi_j$ is an integer multiple of $\pi$. Note that $-\bm{r}_j$ and $\bm{r}_j$ correspond to the same projective measurement, except for the labeling of the two outcomes. By replacing some $\bm{r}_j$ with $-\bm{r}_j$ if necessary, the two conditions in \lref{lem:CompatibleMeas} can be simplified as
\begin{enumerate}
	\item $\bm{r}_j=(0,0, 1)^\rmT$ for each $j$.
	
	\item $\bm{r}_j=(\cos \phi_j,\sin\phi_j,0)^\rmT$ with $\sum_j \phi_j=0\mod 2\pi$.
\end{enumerate}
The two types of compatible measurements determined in \lref{lem:CompatibleMeas} can be used to construct two types of tests for verifying the GHZ state.

\begin{proof}[Proof of \lref{lem:CompatibleMeas}]
	If either one of the  two conditions in \lref{lem:CompatibleMeas} holds, then it is easy to verify that the set $\{\bm{r}_1, \bm{r}_2, \ldots, \bm{r}_n\}$ of local projective measurements is compatible. 
	
	Conversely, suppose the set $\{\bm{r}_1, \bm{r}_2, \ldots, \bm{r}_n\}$  is compatible. If each Bloch vector $\bm{r}_j$ lies on the equator and thus has the form $(\cos \phi_j,\sin\phi_j,0)^\rmT$ for some azimuthal angle $\phi_j$, then the  two eigenstates of the operator $\bm{r}_j\cdot\bm{\sigma}$ read $(|0\rangle \pm \rme^{\rmi\phi_j }|1\rangle)/\sqrt{2}$. After the projective measurements of all parties except for party $k$, the two possible reduced states of party $k$ read
	\begin{equation}
	\frac{1}{\sqrt{2}}(|0\rangle \pm \rme^{\rmi\phi_k-\rmi \varphi}|1\rangle),\quad \varphi=\sum_{j=1}^n\phi_j.
	\end{equation}
	Note that the two states are the eigenstates of $\bm{r}_k\cdot\bm{\sigma}$ iff $\varphi=0\mod \pi$. So condition 2 in \lref{lem:CompatibleMeas} holds.

	Next, suppose at least one $\bm{r}_j$, say $\bm{r}_1$, does not lie on the equator. Then  $\bm{r}_1=(\sin\theta\cos\phi,\sin\theta\sin\phi,\cos\theta)^\rmT$ with $0\leq \theta\leq \pi$, $0\leq \phi<2\pi$, and $\theta\neq \pi/2$.
	The two eigenstates of the operator $\bm{r}_1\cdot\bm{\sigma}$ have the form
	\begin{align}
	|\psi_+\rangle &=\cos\frac{\theta}{2}|0\rangle +\sin\frac{\theta}{2}\rme^{\rmi\phi}|1\rangle,\\
	|\psi_-\rangle &=-\sin\frac{\theta}{2}|0\rangle +\cos\frac{\theta}{2}\rme^{\rmi\phi}|1\rangle.
	\end{align}
	After the projective measurement on the first qubit, the two possible reduced states of the remaining $n-1$ qubits are given by
	\begin{align}
	|\Psi_+\rangle &=\cos\frac{\theta}{2}|0\rangle^{\otimes (n-1)} +\sin\frac{\theta}{2}\rme^{-\rmi\phi}|1\rangle^{\otimes (n-1)},\\
	|\Psi_-\rangle &=-\sin\frac{\theta}{2}|0\rangle^{\otimes (n-1)} +\cos\frac{\theta}{2}\rme^{-\rmi\phi}|1\rangle^{\otimes (n-1)}.
	\end{align}
	The reduced state of $|\Psi_+\rangle$ for each qubit $j=2,3,\ldots,n$ reads
	\begin{equation}
	\rho_+=\cos^2\frac{\theta}{2}|0\rangle\langle 0|+\sin^2\frac{\theta}{2}|1\rangle\langle1|.
	\end{equation}
	By contrast, the reduced state of $|\Psi_-\rangle$ for each qubit $j=2,3,\ldots,n$ reads
	\begin{equation}
	\rho_-=\sin^2\frac{\theta}{2}|0\rangle\langle 0|+\cos^2\frac{\theta}{2}|1\rangle\langle1|.
	\end{equation}
	
	By assumption $\rho_+$ is a convex combination of the two eigenstates of $\bm{r}_j\cdot\bm{\sigma}$ for each $j=2,3,\ldots, n$, and so is $\rho_-$. It follows that $\bm{r}_j=(0,0,\pm 1)^\rmT$ for $j=2,3,\ldots,n$. Accordingly, the two possible reduced states of qubit 1 conditioned on the measurement outcomes of the other $n-1$ qubits are $|0\rangle\langle0|$ and $|1\rangle\langle 1|$, which implies that $\bm{r}_1=(0,0,\pm 1)^\rmT$. Therefore, all $\bm{r}_j$ are parallel to the unit vector $(0,0,1)^\rmT$, and  condition 1 in \lref{lem:CompatibleMeas} holds.
	This observation completes the proof of \lref{lem:CompatibleMeas}.
\end{proof}

\subsection*{\label{app:GHZprotocols} Tests and protocols  for verifying  GHZ states}
Recall that $\scrH$ is the set of honest parties, and  $V_\scrH$ is the two-dimensional subspace spanned by $|0\rangle_\scrH$ and $|1\rangle_\scrH$, where
\begin{equation}
|0\rangle_\scrH=\bigotimes_{j\in \scrH}|0\rangle_j,  \quad |1\rangle_\scrH=\bigotimes_{j\in \scrH}|1\rangle_j. 
\end{equation}
Let $P_\scrH$ be the projector onto $V_\scrH$ and let $\mathbb{I}_\scrH$ be the identity operator on the Hilbert space associated with $\scrH$. To simplify the terminology, all parties other than the honest parties will be referred to as the adversary. Without loss of generality, we can assume that the actual state $\rho$ to be verified is prepared by the adversary. Let $\rho_\scrH:=\tr_\scrD(\rho)$ be the reduced state of $\rho$  for the honest parties.

The two types of compatible measurements determined in \lref{lem:CompatibleMeas} can be used to construct two types of tests for verifying the GHZ state. 
In the first type, all parties perform $Z$ measurements, and the test is passed if they obtain the same outcome. In the perspective of the  honest parties, including the verifier, these measurements realize the  three-outcome projective measurement 
\begin{equation}
\{|0\rangle_\scrH\langle 0|,\;|1\rangle_\scrH\langle 1|, \;\mathbb{I}_\scrH-P_\scrH \}. 
\end{equation}
The state $\rho$ prepared by the adversary  is rejected immediately if the third outcome occurs. To maximize the guessing probability, it is advantageous for the adversary to prepare the state $\rho$ such that $\rho_\scrH$  is supported in $V_\scrH$.  Then the verifier  effectively realizes the $Z$ measurement on $V_\scrH$.

In the second type of tests, party $j$  performs the $X(\phi_j)$ measurement with $\sum_j \phi_j=0\mod 2\pi$, where
\begin{equation}
X(\phi_j):=\begin{pmatrix}
0& \rme^{-\rmi \phi_j} \\
\rme^{\rmi\phi_j} &0
\end{pmatrix}
\end{equation}
corresponds to the Bloch vector $(\cos\phi_j,\sin\phi_j,0)^\rmT$.
The test is passed if the number of outcome $-1$ is even. In the perspective of the  honest parties, these measurements realize the  projective measurement with two outcomes 
\begin{equation}
P_\pm(\{\phi_j\}_j)=\frac{\mathbb{I}_\scrH\pm\bigotimes_{j\in \scrH} X(\phi_j)}{2}.
\end{equation}
Let  $\phi_\scrH=\sum_{j\in \scrH}\phi_j \mod 2\pi$. In the subspace $V_\scrH$, this measurement reduces to a projective measurement with two outcomes
\begin{equation}
P_\pm^\scrH=\frac{P_\scrH\pm X_\scrH(\phi_\scrH)}{2},
\end{equation}
where 
\begin{equation}
X_\scrH(\phi_\scrH)=\rme^{-\rmi\phi_\scrH}|0\rangle_\scrH\langle 1|+\rme^{\rmi\phi_\scrH}|1\rangle_\scrH\langle 0|. 
\end{equation}
So the verifier can   effectively realize the $X(\phi_\scrH)$ measurement on  $V_\scrH$. If $\rho_\scrH$ is supported in $V_\scrH$, then the guessing probability is the same as in the verification of the Bell state. However, it is not clear whether the adversary can increase the guessing probability if $\rho_\scrH$ is not supported in $V_\scrH$. To eliminate this problem, we can introduce some randomness in each $\phi_j$.

Suppose $\phi_1,\phi_2,\ldots, \phi_{n}$ are chosen independently and uniformly at  random from  the interval $[0,2\pi)$;  then $\phi_\scrH$ is uniformly distributed in $[0,2\pi)$. Given $\phi\in[0,2\pi)$,  the average of         $\bigotimes_{j\in \scrH}  X(\phi_j)$ under the condition $\phi_\scrH =\phi$ reads
\begin{align}\label{eq:Xphi}
\biggl\<\bigotimes_{j\in \scrH} X(\phi_j)\biggr\>_\phi=X_\scrH(\phi),
\end{align}
which implies that 
\begin{align}
Q_\pm:=&\langle P_\pm(\{\phi_j\}_j)\rangle_\phi =\frac{\mathbb{I}_\scrH\pm X_\scrH(\phi_\scrH)}{2}\nonumber\\
=&\frac{P_\scrH\pm X_\scrH(\phi_\scrH)}{2}+\frac{\mathbb{I}_\scrH-P_\scrH}{2}.\label{eq:Pphi}
\end{align}
In this way the verifier can   effectively realize the $X(\phi)$ measurement on  $V_\scrH$, where $\phi$ is completely random. The resulting verification protocol corresponds to the equator protocol in the verification of the Bell state. Thanks to the following equality
\begin{equation}
(\mathbb{I}-P_\scrH)Q_+(\mathbb{I}-P_\scrH)=(\mathbb{I}-P_\scrH)Q_-(\mathbb{I}-P_\scrH),
\end{equation}
the adversary cannot gain any advantage if $\rho_\scrH$ is not supported in $V_\scrH$, and the maximum guessing probability is the same as in the verification of the Bell state. Equator$+Z$ protocol can be constructed by adding the $Z$ measurement. 

Let $M\geq 3$ be an integer. If   $\phi_1,\phi_2,\ldots, \phi_{n}$ are chosen independently and uniformly at  random from  the discrete set  $\{2k\pi/M\}_{k=0}^{M-1}$;  then $\phi_\scrH$ is uniformly distributed in the same set. In addition, \esref{eq:Xphi} and \eqref{eq:Pphi} still hold for any $\phi\in \{2k\pi/M\}_{k=0}^{M-1}$. 
In this way the verifier can  effectively realize the polygon protocol for verifying the GHZ state. Again the adversary cannot gain any advantage if $\rho_\scrH$ is not supported in $V_\scrH$, and the maximum guessing probability is the same as in the verification of the Bell state.  Polygon$+Z$ protocol can be constructed by adding the $Z$ measurement.

In summary, each verification protocol of the $n$-qubit GHZ state corresponds to a probability distribution  on the Bloch sphere which is supported on the equator together with the north and south poles. Moreover,  the equator protocol, equator$+Z$ protocol, polygon protocols, and polygon$+Z$ protocols (including the $XYZ$ protocol) can be generalized to GHZ states 
such that the guessing probabilities and sample efficiencies are identical to  the counterparts for the verification  of the Bell state. Notably, GHZ states can be verified almost as efficiently as Bell states. The simplest verification protocol is the $XY$ protocol, which is equivalent to the polygon protocol with four vertices. The optimal protocol is an equator$+Z$ protocol, in which the optimal probability $p_Z$ for performing the $Z$ measurement depends on the target fidelity $F$ as in the verification  of the Bell state. 
 
\bigskip

\onecolumngrid
\section*{Supplementary References}
\vspace{-3ex}
\twocolumngrid

\end{document}